\documentclass[11pt,letterpaper]{article}
\usepackage[margin=1in]{geometry}
\usepackage[T1]{fontenc}
\usepackage[dvipsnames]{xcolor}
\usepackage{amsthm}
\usepackage{enumitem}
\usepackage[linesnumbered,ruled,noline,noend,algonl]{algorithm2e}
\usepackage{booktabs}
\usepackage{tabularx}
\usepackage{amsmath, amssymb}
\usepackage{mathtools}
\usepackage{xspace}
\usepackage[framemethod=TikZ]{mdframed}
\usepackage{amsfonts}
\usepackage{mathtools}
\usepackage{thm-restate}
\usepackage{cases}
\usepackage{tikz}
\usepackage{varwidth}

\usepackage{footnote}
\usepackage[textwidth=1.9cm,textsize=small
]{todonotes}

\usepackage[hidelinks]{hyperref}
\hypersetup{colorlinks=false,breaklinks=true}
\usepackage[capitalise,noabbrev]{cleveref}

\usetikzlibrary{automata,positioning,arrows.meta,arrows,decorations.pathmorphing,backgrounds,calc}
\tikzstyle{oscillate} = [decorate, decoration={snake, amplitude=.3mm, segment length=3mm, post length=2mm}]

\makeatletter
    \tikzset{
    from/.style args={#1 to #2}{
        above right={0cm of #1},
        /utils/exec=\pgfpointdiff
            {\tikz@scan@one@point\pgfutil@firstofone(#1)\relax}
            {\tikz@scan@one@point\pgfutil@firstofone(#2)\relax},
        minimum width/.expanded=\the\pgf@x,
        minimum height/.expanded=\the\pgf@y}}
\makeatother

\graphicspath{{./figures}}

\makesavenoteenv{table}

\newenvironment{claimproof}[1][\unskip]{\noindent {\emph{Proof.\space}}}{\hfill$\triangleleft$ \smallskip}

\newcommand{\eps}{\varepsilon}
\newcommand{\nat}{\mathbb{N}}

\newtheorem{theorem}{Theorem}[section]
\newtheorem{lemma}[theorem]{Lemma}
\newtheorem{definition}[theorem]{Definition}
\newtheorem{corollary}[theorem]{Corollary}

\newtheorem{observation}[theorem]{Observation}

\newtheorem{fact}[theorem]{Fact}

\newtheorem{property}{Property}
\newtheorem{claim}{Claim}[lemma]
\newtheorem{pty}{Property}[claim]

\newcommand{\norm}[1]{\left\lVert#1\right\rVert}

\newcommand{\floor}[1]{\left\lfloor #1 \right\rfloor}
\newcommand{\ceil}[1]{\left\lceil #1 \right\rceil}

\newcommand{\Oh}{\mathcal{O}}

\newcommand{\Otilde}{\widetilde{\Oh}}
\newcommand{\Ot}{\Otilde}

\newcommand{\tOh}{\Otilde}

\newcommand{\N}{\mathbb{N}}
\newcommand{\Z}{\mathbb{Z}}

\newcommand{\Mod}[1]{\ (\mathrm{mod}\ #1)}

\newcommand{\bincode}[1]{\mathsf{#1}_2}

\renewcommand{\leq}{\leqslant}
\renewcommand{\geq}{\geqslant}
\renewcommand{\le}{\leqslant}
\renewcommand{\ge}{\geqslant}

\mdfdefinestyle{MyFrame}{%
		linecolor=black,
		middlelinewidth=1pt,
		outerlinewidth=0pt,
		roundcorner=5pt,
		innertopmargin=0pt,
		innerbottommargin=2pt,
		innerrightmargin=2pt,
		innerleftmargin=2pt,
	leftmargin = 2pt,
	rightmargin = 2pt
}

\newcommand{\Problem}[4]{
\begin{center}
\par\addvspace{.5\baselineskip}
	\noindent
	\begin{tabularx}{\textwidth}{ @{\hspace{4ex}}l X c} 
		\multicolumn{2}{@{\hspace{\parindent}}l}{#1} \\
		\textbf{Input:} & {#2} \\
		\textbf{Parameter:} & {#4} \\
		\textbf{Question:} & {#3} \\ 
	\end{tabularx}
	\par\addvspace{.5\baselineskip}
\end{center}
}

\Crefname{equation}{Equation}{Equations}
\Crefname{assumption}{Assumption}{Assumptions}
\Crefname{fact}{Fact}{Facts}
\Crefname{enumi}{Property}{Properties}
\Crefname{observation}{Observation}{Observations}

\makeatletter
\@ifpackageloaded{algorithm2e}{%
\PackageInfo{cleveref}{`algorithm2e' support loaded}%
\crefalias{algocf}{algorithm}%
\crefalias{algocfline}{line}%
\crefalias{AlgoLine}{line}%
\let\cref@old@algocf@nl@sethref\algocf@nl@sethref%
\renewcommand{\algocf@nl@sethref}[1]{%
	\cref@old@algocf@nl@sethref{#1}%
	\cref@constructprefix{AlgoLine}{\cref@result}%
	\@ifundefined{cref@AlgoLine@alias}%
		{\def\@tempa{AlgoLine}}%
		{\def\@tempa{\csname cref@AlgoLine@alias\endcsname}}%
	\xdef\cref@currentlabel{%
		[\@tempa][\arabic{AlgoLine}][\cref@result]%
		\csname p@AlgoLine\endcsname\csname theAlgoLine\endcsname}}%
}{}
\makeatother

\newlist{yesenum}{enumerate}{2}
\setlist[yesenum,1]{%
		label= \textbf{\textsc{Yes}},
		ref= \textsc{Yes},
		leftmargin=*, labelindent=0pt, itemsep=0.5em,topsep=0.25em}
\setlist[yesenum,2]{%
		label= \arabic*.,
		ref= {\theyesenumi}.\arabic*,
		leftmargin=1.25em, labelindent=0pt, itemsep=0.25em,topsep=0pt}
\newlist{yesenump}{enumerate}{2}
\setlist[yesenump,1]{%
		label= \textbf{\textsc{Yes}},
		ref= \textsc{Yes},
		leftmargin=*, labelindent=0pt, itemsep=0.5em,topsep=0.25em}
\setlist[yesenump,2]{%
		label= \roman*.,
		ref= {\theyesenumpi}.\roman*,
		leftmargin=1.25em, labelindent=0pt, itemsep=0.25em,topsep=0pt}

\newlist{noenum}{enumerate}{2}
\setlist[noenum,1]{%
		label= \textbf{\textsc{No}},
		ref= \textsc{No},
		leftmargin=*, labelindent=0pt, itemsep=0.5em,topsep=0.25em}
\setlist[noenum,2]{%
		label= \arabic*.,
		ref= {\thenoenumi}.\arabic*,
		leftmargin=1.25em, labelindent=0pt, itemsep=0.25em,topsep=0pt}
\newlist{noenump}{enumerate}{2}
\setlist[noenump,1]{%
		label= \textbf{\textsc{No}},
		ref= \textsc{No},
		leftmargin=*, labelindent=0pt, itemsep=0.5em,topsep=0.25em}
\setlist[noenump,2]{%
		label= \roman*.,
		ref= {\thenoenumpi}.\roman*,
		leftmargin=1.25em, labelindent=0pt, itemsep=0.25em,topsep=0pt}

\Crefname{noenumi}{Property}{Properties}
\Crefname{noenumii}{Property}{Properties}
\Crefname{yesenumi}{Property}{Properties}
\Crefname{yesenumii}{Property}{Properties}

\Crefname{yesenumpi}{Property}{Properties}
\Crefname{yesenumpii}{Property}{Properties}
\Crefname{noenumpi}{Property}{Properties}
\Crefname{noenumpii}{Property}{Properties}

\renewcommand{\subset}{\subseteq}
\renewcommand{\le}{\leqslant}
\renewcommand{\leq}{\leqslant}
\renewcommand{\ge}{\geqslant}
\renewcommand{\geq}{\geqslant}
\renewcommand{\Tilde}{\widetilde}

\renewcommand{\epsilon}{\ensuremath\varepsilon}

\renewcommand{\phi}{\ensuremath{\varphi}}

\newcommand{\mc}[1]{\mathcal{#1}}

\newcommand{\Cmax}{C_{\max}}

\newcommand\restr[2]{{
	\left.\kern-\nulldelimiterspace 
	#1 
	\littletaller 
	\right|_{#2} 
	}}
\newcommand{\littletaller}{\mathchoice{\vphantom{\big|}}{}{}{}}

\newcommand{\subsetsum}{{\textup{\textsc{Subset Sum}}}\xspace}
\newcommand{\vectorsubsetsum}{{\textup{\textsc{Vector Subset Sum}}}\xspace}

\newcommand{\binpacking}[1][k]{{\textup{\textsc{${#1}$-Bin Packing}}}\xspace}
\newcommand{\binpackingMultisets}[1][k]{{\textup{\textsc{Multiset ${#1}$-Bin Packing}}}\xspace}

\newcommand{\kpartition}[1][k]{{\textup{\textsc{${#1}$-way Partition}}}\xspace}
\newcommand{\kpartitionTargets}[1][k]{{\textup{\textsc{${#1}$-way Partition with targets}}}\xspace}
\newcommand{\kpartitionBounded}[1][k]{{\textup{\textsc{Bounded ${#1}$-way Partition}}}\xspace}
\newcommand{\kpartitionMultisets}[1][k]{{\textup{\textsc{Multiset ${#1}$-way Partition}}}\xspace}
\newcommand{\kpartitionBoundedMultisets}[1][k]{{\textup{\textsc{Multiset Bounded ${#1}$-way Partition}}}\xspace}
\newcommand{\kpartitionTargetsMultisets}[1][k]{{\textup{\textsc{Multiset ${#1}$-way Partition with targets}}}\xspace}

\newcommand{\pbtwoYes}[1][k]{{\textup{\textsc{Grouped ${#1}$-way Partition}}}\xspace}
\newcommand{\pbtwo}[1][k]{{\textup{\textsc{Weak Grouped ${#1}$-way Partition}}}\xspace}
\newcommand{\pbtwoTargetsMultisets}[1][k]{{\textup{\textsc{Multiset Weak Grouped ${#1}$-way Partition with targets}}}\xspace}
\newcommand{\pbtwoMultisets}[1][k]{{\textup{\textsc{Multiset Weak Grouped ${#1}$-way Partition}}}\xspace}

\newcommand{\PSwjUj}[1][k]{{\ensuremath{P_{#1} || \Sigma w_j U_j}}\xspace}
\newcommand{\QSwjUj}[1][k]{{\ensuremath{Q_{#1} || \Sigma w_j U_j}}\xspace}
\newcommand{\PSpjUj}[1][k]{{\ensuremath{P_{#1 } || \Sigma p_j U_j}}\xspace}
\newcommand{\QSpjUj}[1][k]{{\ensuremath{Q_{#1} || \Sigma p_j U_j}}\xspace}
\newcommand{\PSwjCj}[1][k]{{\ensuremath{P_{#1} || \Sigma w_j C_j}}\xspace}
\newcommand{\QSwjCj}[1][k]{{\ensuremath{Q_{#1} || \Sigma w_j C_j}}\xspace}
\newcommand{\PSUj}[1][k]{{\ensuremath{P_{#1} || \Sigma U_j}}\xspace}
\newcommand{\QSUj}[1][k]{{\ensuremath{Q_{#1} || \Sigma U_j}}\xspace}
\newcommand{\PrjCmax}[1][k]{{\ensuremath{P_{#1} |r_j| \Cmax}}\xspace}
\newcommand{\QrjCmax}[1][k]{{\ensuremath{Q_{#1} |r_j| \Cmax}}\xspace}
\newcommand{\PCmax}[1][k]{{\ensuremath{P_{#1} || \Cmax}}\xspace}
\newcommand{\QCmax}[1][k]{{\ensuremath{Q_{#1} || \Cmax}}\xspace}
\newcommand{\PLmax}[1][k]{{\ensuremath{P_{#1} || L_{\text{max}}}}\xspace}
\newcommand{\PTmax}[1][k]{{\ensuremath{P_{#1} || T_{\text{max}}}}\xspace}

\newcommand{\tmp}[1]{}

\definecolor{cb_yellow}{RGB}{255,176,0}
\definecolor{cb_orange}{RGB}{254,97,0}
\definecolor{cb_red}{RGB}{220,038,127}
\definecolor{cb_purple}{RGB}{120,94,240}
\definecolor{cb_blue}{RGB}{100,143,255}

\definecolor{OI_green}{RGB}{0,158,115}
\definecolor{OI_blue}{RGB}{0,114,178}
\definecolor{OI_red}{RGB}{213,94,0}
\definecolor{OI_yellow}{RGB}{240,228,66}
\definecolor{OI_lightblue}{RGB}{86,180,233}
\definecolor{OI_pink}{RGB}{204,121,167}
\definecolor{OI_orange}{RGB}{230,159,0}

\title{Tight (S)ETH-based Lower Bounds for Pseudopolynomial Algorithms for Bin Packing and Multi-Machine Scheduling}
\author{Karl Bringmann\footnote{ETH Zurich, \texttt{\{karl.bringmann,anita.duerr\}@inf.ethz.ch}. Part of this work was done while affiliated to Saarland University and Max Planck Institute for Informatics, Saarbrücken, Germany, where this work was part of the project TIPEA that has received funding from the European Research Council (ERC) under the European Unions Horizon 2020 research and innovation programme (grant agreement No. 850979).}\and Anita D\"{u}rr$^*$
\and Karol W\k{e}grzycki\footnote{Max Planck Institute for Informatics, Saarbr\"ucken, Germany, \texttt{kwegrzyc@mpi-inf.mpg.de}. Supported by 
    the Deutsche Forschungsgemeinschaft (DFG, German Research Foundation) grant number
    559177164.}}

\date{}

\begin{document}
\maketitle

\thispagestyle{empty}
\begin{abstract}


Bin Packing with $k$ bins is a fundamental optimisation problem in which we are given a set of $n$ integers and a capacity~$T$ and the goal is to partition the set into $k$ subsets, each of total sum at most~$T$. Bin Packing is NP-hard already for $k=2$ and a textbook dynamic programming algorithm solves it in pseudopolynomial time $\Oh(n T^{k-1})$. Jansen, Kratsch, Marx, and Schlotter [JCSS'13] proved that this time cannot be improved to $(nT)^{o(k / \log k)}$ assuming the Exponential Time Hypothesis (ETH). Their result has become an important building block, explaining the hardness of many problems in parameterised complexity. Note that their result is one log-factor short of being tight. In this paper, we prove a tight ETH-based lower bound for Bin Packing, ruling out time $2^{o(n)} T^{o(k)}$. This answers an open problem of Jansen et al.\ and yields improved lower bounds for many applications in parameterised complexity.

Since Bin Packing is an example of multi-machine scheduling, it is natural to next study other scheduling problems. We prove tight lower bounds based on the Strong Exponential Time Hypothesis (SETH) for several classic $k$-machine scheduling problems, including makespan minimisation with release dates (\PrjCmax), minimizing the number of tardy jobs (\PSUj), and minimizing the weighted sum of completion times (\PSwjCj). For all these problems, we rule out time $2^{o(n)} T^{k-1-\varepsilon}$ for any $\varepsilon > 0$ assuming SETH, where $T$ is the total processing time; this matches classic $n^{\Oh(1)} T^{k-1}$-time algorithms from the 60s and 70s. Moreover, we rule out time $2^{o(n)} T^{k-\varepsilon}$ for minimizing the total processing time of tardy jobs (\PSpjUj), which matches a classic $\Oh(n T^{k})$-time algorithm and answers an open problem of Fischer and Wennmann [TheoretiCS'25].

\end{abstract}

\clearpage
\setcounter{page}{1}

\section{Introduction}

%

Bin Packing is a fundamental optimisation problem that has many applications both in theory and practice and is widely studied in theoretical computer science, mathematical optimisation, and operations research.
In this problem, given $n$ items $a_1,\ldots,a_n \in \mathbb{N}$, a capacity $T \in \mathbb{N}$, and a number of bins $k \ge 2$, the task is to partition the items into $k$ subsets, each summing to at most $T$. Note that this is the decision version of the problem (in the optimisation version the goal is to minimise the number of bins~$k$). In this paper we study \binpacking, where the number of bins $k$ is fixed (a closely related view is to study the parameterised complexity of {{Bin Packing}} with parameter~$k$). 
Since the problem is NP-hard already for $k=2$, we focus on pseudopolynomial-time algorithms, i.e.~we allow polynomial time in terms of $n$ and~$T$. 
%
%
\binpacking can be solved in $\Oh(n T^{k-1})$ time by a simple dynamic programming algorithm that dates back to the
60s~\cite{rothkopf1966scheduling,lawler1969functional}. This running time has not been improved, despite significant interest in the problem.

A natural special case of \binpacking is the \kpartition problem, where the goal is to partition a given set of $n$ positive integers into $k$ subsets of \emph{equal} sum. In fact, we prove that the two problems are equivalent with respect to the questions studied in this paper (see~\cref{thm:equivalence-bin-packing}), so the reader may think of either of these problems in what follows.

Another important special case is the \subsetsum\footnote{Given a set of $n$ positive integers and a target $T$, does some subset of the input integers sum to $T$?} problem, as it is equivalent to \binpacking[2] by a standard reduction. The vast literature on \subsetsum includes a classic $\Oh(nT)$-time dynamic programming algorithm~\cite{bellman1957dynamic} and a more recent improvement to time $\tOh(n + T)$~\cite{DBLP:conf/soda/Bringmann17}.\footnote{The $\Ot(\cdot)$ notation hides polylogarithmic factors.} 
Lower bounds for \subsetsum are a recent success story of fine-grained complexity theory: The Exponential Time Hypothesis (ETH) can be used to prove conditional lower bounds that are tight up to constants in the exponent, and such an ETH-based lower bound ruling out time $2^{o(n)} T^{o(1)}$ was shown in~\cite{BuhrmanLT15,JansenLL16}. More precise lower bounds that even determine the optimal constant in the exponent can be shown using the Strong Exponential Time Hypothesis (SETH), and indeed a tight SETH-based lower bound ruling out time $2^{o(n)} T^{1-\eps}$ for any $\eps>0$ was shown by Abboud et al.~\cite{AbboudBHS19}.




\subsection{Fine-grained Complexity of Bin Packing}
In this paper we ask whether a similarly precise understanding as for \subsetsum can be obtained for \binpacking. 
The state of the art lower bound for \binpacking was proven by Jansen, Kratsch, Marx, and Schlotter~\cite{jansen2013bin}, and shows that it cannot be solved in time $(nT)^{o(k / \log k)}$ assuming ETH. Jansen et al.\ also asked whether their result can be improved to rule out time $(nT)^{o(k)}$, thus avoiding the log-factor loss of their reduction.
Our first main result answers their question by indeed providing a \emph{tight ETH-based lower bound}, in fact with a better dependence on~$n$:

\begin{theorem}\label{thm:bin-packing-lowerbound}
    Assuming ETH, \binpacking cannot be solved in time $2^{o(n)} T^{o(k)}$.
\end{theorem}


This shows that any algorithm for \binpacking requires either exponential time $2^{\Omega(n)}$, which matches the known $2^{\Oh(n)}$-time algorithms~\cite{BjorklundHK09,NederlofPSW23}, or time $T^{\Omega(k)}$, which matches the $\Oh(n T^{k-1})$-time dynamic programming algorithm. Since ETH ignores constants in the exponent, what we obtain is the best possible ETH-based lower bound in terms of $n$ and $T$. 
We prove our lower bound by a creative new encoding of communication between bins, see \Cref{sec:technical-overview} for a discussion.


\smallskip

The lower bound of Jansen et al.~\cite{jansen2013bin} has  become an important starting point of reductions to prove lower bounds in parameterised complexity, see, e.g.~\cite{blavzej2024parameterized,
blavzej2024equitable,
bliem2016complexity,dreier2019complexity,
froese2024disentangling,gima2024extended,
hanaka2024parameterized,heeger2023single,
JAVADI20241,
lafond2025cluster}. However, the log-factor loss in their exponent is inherited by all of these lower bounds. Our tight lower bound for \binpacking yields, via the same reductions, tight ETH-based lower bounds for many applications in parameterised complexity. 
In particular, \cref{thm:bin-packing-lowerbound} immediately gives tight lower
bounds for 
$p$-Cluster Editing on cographs~\cite{lafond2025cluster}, Sparsest Cut on bounded treewidth graphs~\cite{JAVADI20241},
Eulerian Strong Component Arc Deletion~\cite{blavzej2024parameterized}, EEF-Allocation~\cite{bliem2016complexity}, Equitable Connected
Partition~\cite{blavzej2024equitable}, Exact Path Packing~\cite{dreier2019complexity}
and two single-machine scheduling problems~\cite{heeger2023single}.
See~\cref{cor:cluster-editing,cor:sparsest-cut,cor:escad,cor:eef-allocation,cor:ecp,cor:path-packing,cor:single-machine} in~\cref{app:further-applications} for a detailed discussion of the improved lower bounds implied by \cref{thm:bin-packing-lowerbound}.



\paragraph{Towards more precise lower bounds?}
Analogously to the success story of \subsetsum, now that we have a tight ETH-based lower bound for \binpacking it is natural to ask for a tight SETH-based lower bound: Can we show that \binpacking cannot be solved in time $2^{o(n)} T^{k-1-\eps}$ for any $\eps > 0$ assuming SETH (or some other popular hypothesis)?
Alas, we leave this as an intriguing open problem. It turns out that reductions to \binpacking are significantly more difficult than reductions to \subsetsum, see~\cref{sec:technical-overview}. 
Since there have been surprising algorithmic breakthroughs for \binpacking in some settings\footnote{In the context of exact algorithms it would have been natural to ask for a lower bound ruling out time $\Oh((2-\eps)^n)$ for any $\eps>0$ for, say, \binpacking[10] -- until Nederlof et al.~\cite{NederlofPSW23} showed that for every $k \ge 2$ there exists $\eps_k >0$ such that \binpacking can be solved in time
$\Oh((2-\eps_k)^n)$.}, it is certainly possible that \binpacking can be solved in time $2^{o(n)} T^{k-1-\eps}$ after all.

\subsection{Fine-grained Complexity of Multi-machine Scheduling}

\binpacking can be reduced to several classic scheduling problems on $k$ machines. As these problems offer more structure for lower bounds, we can bypass the difficulties encountered in our search for more precise lower bounds for \binpacking,  by proving, as our second main result, \emph{tight SETH-based lower bounds for several multi-machine scheduling problems}.
We summarise our results in \cref{tab:summary}.




\begin{table}[!t]
\centering
\begin{tabular}{lllll}
\toprule
Problem & Upper Bound & & Our Result  & \\
\midrule
\binpacking $\equiv$ \PCmax & $\Oh(n T^{k-1})$ & \cite{rothkopf1966scheduling,lawler1969functional} & no $2^{o(n)} T^{o(k)}$ & \cref{thm:bin-packing-lowerbound} \\
\PrjCmax & $\Oh(n T^{k-1})$ & \cite{sahni1976algorithms} & no $2^{o(n)} T^{k-1-\eps}$ & \cref{thm:SETH_to_PrjCmax} \\
\PSwjCj & $\Oh(n T^{k-1})$ & \cite{rothkopf1966scheduling,lawler1969functional} & no $2^{o(n)} T^{k-1-\eps}$ & \cref{thm:SETH_to_PSwjCj} \\
$\PSUj$ & $\Oh(n^2 T^{k-1})$ & [{folklore}] & no $2^{o(n)} T^{k-1-\eps}$ & \cref{cor:SETH_to_PSUj} \\
\PSpjUj & $\Otilde(n + T^{k})$ & \cite{fischer2025minimizing} & no $2^{o(n)} T^{k-\eps}$ & \cref{thm:SETH_to_PSpjUj} \\
\PSwjUj & $\Oh(n T^{k})$ & \cite{rothkopf1966scheduling,lawler1969functional} & no $2^{o(n)} T^{k-\eps}$ & \cref{thm:SETH_to_PSpjUj} \\
\bottomrule
\end{tabular}
\caption{Summary of our lower bounds for \binpacking and various multi-machine scheduling problems. 
The lower bound of \cref{thm:bin-packing-lowerbound} assumes ETH, while all other listed results assume SETH. 
Prior works only exclude time $(nT)^{o(k/\log k)}$ assume ETH~\cite{jansen2013bin}. 
}
\label{tab:summary}
\end{table}

In order to describe these results in detail, let us start with some background on scheduling problems on parallel machines. 
Generally, in scheduling problems the input consists of $n$ jobs with processing times $p_1,\ldots,p_n \in \mathbb{N}$. A schedule is an assignment of each job to a machine and a starting time such that no two jobs are processed at the same time by the same machine. 
The three-field notation $\alpha|\beta|\gamma$ introduced by Graham et al.~\cite{graham1979optimization} is a shorthand notation for scheduling problems. The first field~$\alpha$ describes the machine environment, which in our case will be either ``$P_k$'', denoting $k$ identical parallel machines where executing job $j$ takes time $p_j$ on any machine, or ``$Q_k$'', denoting $k$ uniform machines with given speeds $s_1,\ldots,s_k \in (0,1]$ where executing job $j$ on machine $i$ takes time $p_j/s_i$. The second field $\beta$ describes additional constraints, which in this paper will either be empty, meaning no additional constraints, or ``$r_j$'', denoting that job $j$ can be processed no earlier than its given release date $r_j$. 
In some problems, job $j$ also comes with a due date $d_j$ or with a weight $w_j$ used in the objective function. As they are implicit from the third field, those characteristics are not specified in the notation. 
The third field~$\gamma$ describes the objective function, which in this paper will be one of the following classic choices, where $C_j$ denotes the completion time of job $j$, and $U_j$ is 1 if the job $j$ completes after its due date $d_j$ (i.e.~job $j$ is a \emph{tardy} job) and 0 otherwise:
\begin{itemize}[noitemsep]
	\item Maximum completion time (makespan) $\Cmax = \max_{j} C_j$,
	\item Total weighted completion time $\Sigma w_j C_j$,
	\item Total number of tardy jobs $\Sigma U_j$,
	\item Total processing time of tardy jobs $\Sigma p_jU_j$,
	\item Total weighted number of tardy jobs $\Sigma w_jU_j$.
\end{itemize}
Note that the studied scheduling problems are defined as optimisation problems where the task is to minimise the respective objective function. However, each of them has a natural decision variant, in which the task is to decide if the objective value is at most $\Lambda$ for some given $\Lambda \in \N$.

Several parameters have been used to analyse pseudopolynomial algorithms for scheduling problems, including the total processing time $T \coloneq \sum_{1 \le j \le n} p_j$, the maximum completion time $C \coloneq \max_{1 \le j \le n} C_j$, and, for problems with due dates, the maximum due date $D \coloneq \max_{1 \le j \le n} d_j$. 
We observe that on identical machines for the scheduling problems studied in this paper we have $C,D \le T$.\footnote{Indeed, for problems without release dates we have $C, D \le T$ without loss of generality since all jobs can trivially be processed by time $T$.
The only problem with release dates studied in this paper is \PrjCmax, for which parameter $D$ is irrelevant as there are no due dates, and we prove in \Cref{obs:PrjCmax_wlog_small_release_dates} that without loss of generality $C \le T$.}
In particular, a lower bound in terms of $T$ implies lower bounds in terms of $C$ and $D$. This will suffice to also prove tight lower bounds on uniform machines ($Q_k$). We will use $T$ as the default parameter.

%
%

\paragraph{Makespan with Release Dates}
In the problems \PrjCmax and $\QrjCmax$, each job has a processing time and a release date, and the task is to minimise the maximum completion time. 
To the best of our knowledge, Sahni~\cite{sahni1976algorithms} was the first
to design a $\Oh(n C^{k-1})$-time algorithm (see~\cite{lawler}). For \PrjCmax, since without loss of generality we can assume that $C \le T$ (see \Cref{obs:PrjCmax_wlog_small_release_dates}), this implies time $\Oh(n T^{k-1})$.

As usual for scheduling problems, there is a simple reduction from \subsetsum; here such a reduction exists for $\PrjCmax[2]$. The hardness of \subsetsum transfers via this reduction, and it follows that for any $k \ge 2$ the problem \PrjCmax is NP-hard and has no $2^{o(n)} T^{1-\eps}$-time algorithm for any $\eps>0$ assuming SETH; see~\cite[Theorem 13]{jansen2023complexity} for details. 

There is also a simple reduction from \binpacking to \PrjCmax, since the special case \PCmax (where all release dates are $r_j=0$) is equivalent to \binpacking. The hardness of \binpacking transfers via this reduction, and from \Cref{thm:bin-packing-lowerbound} it follows that \PrjCmax cannot be solved in time $2^{o(n)} T^{o(k)}$ assuming ETH.


This leaves open whether \PrjCmax can be solved in time $2^{o(n)} T^{k-1-\eps}$. In contrast to \binpacking, the \PrjCmax problem has additional structure that allows to bypass the difficulties encountered when studying \binpacking. We indeed prove a tight SETH-based lower bound:


\begin{theorem}\label{thm:SETH_to_PrjCmax}
		Assuming SETH, \PrjCmax cannot be solved in time $2^{o(n)} T^{k-1-\eps}$ for any $\eps > 0$.
\end{theorem}
Since on identical machines we can assume without loss of generality that $C \le T$ (see \Cref{obs:PrjCmax_wlog_small_release_dates}), we also rule out time $2^{o(n)} C^{k-1-\eps}$ for \PrjCmax, and, since it is a more general problem, also for \QrjCmax.\footnote{
	Since \PrjCmax is parameter-preserving equivalent to \PLmax and to \PTmax, the lower bound also applies for the latter problems. The definitions of \PLmax and \PTmax as well as the equivalence to \PrjCmax is discussed in \cref{app:folklore}.} This matches the classic $\Oh(nC^{k-1})$-time algorithms.

\medskip
In what follows we discuss several further examples of multi-machine scheduling problems that have pseudopolynomial-time algorithms according to the classical survey~\cite{lawler}.


\paragraph{Weighted Sum of Completion Times}
In the problems \PSwjCj and \QSwjCj, each job has a processing time $p_j$ and a weight $w_j$, and the task is to minimise $\sum_{1 \le j \le n} w_j C_j$, where $C_j$ is the completion time of job $j$.
The problem \PSwjCj was one of the first scheduling problems considered in the literature. Lawler, Moore, and
Rothkopf~\cite{lawler1969functional,
rothkopf1966scheduling} designed an $\Oh(n C^{k-1})$-time algorithm, where $C$ is the maximum completion time in an optimal solution; see~\cite[Section 8.2, page 468]{lawler} for the same running time for \QSwjCj. 
For \PSwjCj, since $C \le T$, this yields running time $\Oh(n T^{k-1})$.
As before, time $2^{o(n)} T^{1-\eps}$ can be ruled out for $k \ge 2$ by a reduction from \subsetsum~\cite[Theorem 15]{jansen2023complexity}, and time $2^{o(n)} T^{o(k)}$ can be ruled out via \Cref{thm:bin-packing-lowerbound}, assuming SETH and ETH respectively.
We prove a tight SETH-based lower bound:
\begin{theorem}\label{thm:SETH_to_PSwjCj}
		Assuming SETH, \PSwjCj cannot be solved in time $2^{o(n)} T^{k-1-\eps}$ for any $\eps > 0$.
\end{theorem}
Since $C \le T$, this rules out time $2^{o(n)} C^{k-1-\eps}$ for \PSwjCj, and thus also for \QSwjCj, again matching the classic $\Oh(n C^{k-1})$-time algorithm. 
We remark that these lower bounds cannot be extended to
the special case where $w_j = 1$ holds for every job $j$, because $Qk || \Sigma C_j$ can be
solved in polynomial time by a reduction to bipartite
matching~\cite{bruno1974scheduling,horn1973minimizing}.



\paragraph{Number of Tardy Jobs}
In the problems $\PSUj$ and $\QSUj$, each job has a processing time~$p_j$ and a due date $d_j$, and the task is to minimise $\sum_{1 \le j \le n} U_j$, where $U_j$ is 1 if job $j$ finishes after its due date~$d_j$, and 0 otherwise.
The problems can be solved in time $\Oh(n^2 T^{k-1})$. This result seems to be folklore, as we were not able to find a reference. We provide a proof in \cref{lem:PSUj_folkloreUB} for completeness. 
By a simple reduction from \Cref{thm:SETH_to_PrjCmax}, we prove a tight SETH-based lower bound:
\begin{corollary}\label{cor:SETH_to_PSUj}
		Assuming SETH, $\PSUj$ cannot be solved in time $2^{o(n)} T^{k-1-\eps}$ for any $\eps > 0$.
\end{corollary}
The same lower bound applies to the more general $\QSUj$, and for both problems we match the $\Oh(n^2 T^{k-1})$-time algorithm. 


\paragraph{Weighted Number of Tardy Jobs and Total Processing Time of Tardy Jobs}
In the problems \PSwjUj and \QSwjUj, each job has a processing time~$p_j$, a due date~$d_j$ and a weight~$w_j$, and the task is to minimise $\sum_{1 \le j \le n} w_j U_j$, where $U_j$ is 1 if job $j$ finishes after its due date~$d_j$, and 0 otherwise. 
These problems can be solved in time $\Oh(n D^{k})$, where $D$ is the maximum due date (see~\cite{lawler1969functional,rothkopf1966scheduling}).
Since we can assume $D \le T$, this also runs in time $\Oh(n T^{k})$.
Special cases are the problems \PSpjUj and \QSpjUj, i.e.~$w_j = p_j$ for every job. In particular, recently the problem \PSpjUj has been thoroughly studied~\cite{bringmann2022faster,klein2023minimizing,schieber2023quick,fischer2025minimizing}, and Fischer and Wennmann~\cite{fischer2025minimizing} designed an $\Ot(n + T^k)$-time algorithm.
They wrote that a conditional lower bound for this 
problem ``\emph{is a challenging question which is not resolved yet}''.
We answer this question and show optimality of their algorithm:
\begin{theorem}\label{thm:SETH_to_PSpjUj}
		Assuming SETH, \PSpjUj cannot be solved in time $2^{o(n)} T^{k-\eps}$ for any $\eps > 0$.
\end{theorem}
Note that this has exponent $k$, in contrast to all previous examples which had exponent $k-1$.
Since we can assume $D \le T$, this also rules out time $2^{o(n)} D^{k-\eps}$. The same lower bound holds for the more general problems \QSpjUj, \PSwjUj, and \QSwjUj, for which it matches the classic $\Oh(n D^{k})$-time algorithm.

\section{Technical Overview}
\label{sec:technical-overview}



In this section, we give a high-level overview of our results. 

\subsection{ETH-hardness of \boldmath$k$-Bin Packing}

Jansen et al.~\cite{jansen2013bin} proved their lower bound for \binpacking by a reduction from the Subgraph Isomorphism problem: Given graphs $G$ and $H$, find a subgraph of $G$ that is isomorphic to~$H$. 
Their reduction maps an instance of Subgraph Isomorphism with $n = |V(G)|$ and $k = |E(H)|$ to a {Bin Packing} instance with $n^{\Oh(1)}$ items and $\Oh(k)$ bins of capacity $T = f(k) n^{\Oh(1)}$, for some computable function~$f$. Since current lower bounds exclude $f(k)n^{o(k/\log{k})}$-time algorithms for Subgraph Isomorphism assuming ETH~\cite{marx2010can}, this implies that \binpacking cannot be solved in time $(nT)^{o(k/\log{k})}$. 

To avoid the barrier posed by the $\log(k)$-factor gap between upper and lower bounds for Subgraph Isomorphism~\cite{marx2010can}, we directly reduce from the 3-SAT problem.
Specifically, to simplify the exposition of the reduction, we will reduce 3-SAT to the following variant of {Bin Packing}\footnote{For a (multi)set $X$, we denote the sum of its elements by $\Sigma(X) = \sum_{x \in X} x$.}:

\Problem{\hypertarget{prob:kpartitiontargetsmultisets}{\kpartitionTargetsMultisets}}
{Multiset of $n$ integers $X \subset \N$; targets $t_1, \dots, t_k \in \N$.}
{Is there a partition of $X$ into $X_1, \dots, X_k$ such that $\Sigma(X_\ell) = t_\ell$ for every $\ell \in [k]$?}
{$T \coloneq \max \{t_1, \dots, t_k\}$}
This problem can be viewed as a variant of \binpacking where (i) the given integers form a multiset instead of a set, (ii) every bin has a different capacity instead of all bins having the same capacity, and (iii) the load of every bin is required to match its capacity instead of only asking the load to not exceed the capacity.
We prove both problems to be equivalent, in the sense that if one of them can be solved in time $2^{o(n)}T^{o(k)}$ then so can the other. (In fact, this equivalence extends to several further problem variants obtained by toggling (i)-(iii), see \cref{thm:equivalence-bin-packing}. We note that some of those reductions are folklore or can be found in the literature, but we include them for completeness in \cref{sec:binpacking_equivalence}.) 

It suffices therefore to reduce 3-SAT to \kpartitionTargetsMultisets.
Let $\phi$ be a 3-CNF formula with $N$ variables and $M$ clauses. By the Sparsification Lemma (\cref{lem:sparsification}) we can assume $M = \Oh(N)$ and each variable appears in $\Oh(1)$ clauses. Fix $k \geq 2$. We will construct a multiset $X \subset \mathbb N$ and bin capacities $t_1, \dots, t_k \in \mathbb N$ with $n = |X| = \Oh(N)$ and $T \coloneq \max \{t_1, \dots, t_k\} = 2^{\Oh(N/k)} N^{\Oh(1)}$ such that there exists a solution to \kpartitionTargetsMultisets if and only if $\phi$ is satisfiable. 
Hence, if \kpartitionTargetsMultisets could be solved in time $2^{o(n)} T^{o(k)}$ then 3-SAT could be solved in time $2^{o(N)}$, contradicting ETH.


We remark that in this overview we treat $k$ as if it would be a function of $n$ tending to $\infty$ and we hide all other constants in the $\Oh(\cdot)$ notation (e.g.~when we write $T = 2^{\Oh(N/k)}$). In reality, $k$ is also a constant, so we need to be very careful with dependencies between different constants.


We view a solution to the \kpartitionTargetsMultisets instance, i.e.~a partition $X_1, \dots, X_k$ of $X$ such that $\Sigma(X_i) = t_\ell$ for all $\ell \in [k]$, as a packing of all items into $k$ bins such that the $\ell$-th bin contains the items of $X_\ell$ and has load $\Sigma(X_\ell)$. We will restrict our attention to the first $k-1$ bins, viewing bin $X_k$ as a ``dumpster'', and thus by abuse of notation we say that an item is \emph{packed} if it is assigned to one of the first $k-1$ bins. 


Note that the goal of $T = 2^{\Oh(N/k)}$ allows us to construct numbers using up to $\Theta(N/k) = \Theta(M/(k-1))$ bits, which means that the sum of each bin can encode assignments to a $1/(k-1)$-fraction of all clauses. 
We therefore construct clause-assignment items as follows. For every clause $C_i$ and every assignment $\alpha \in \{0, 1\}^3$ of its variables that satisfies $C_i$, we create an item $z(i, \alpha) \in \mathbb N$. Furthermore, we equally distribute the clauses over the bins by assigning clause $C_i$ to  bin $\lceil (k-1) i / M \rceil$. We ensure:
\begin{enumerate}[label=(P\arabic*), ref=(P\arabic*), leftmargin=*, align=left]
	\itemsep0em
	\item\label{prop:1-block} Item $z(i, \alpha)$ can only be packed into the bin assigned to $C_i$, or into the dumpster bin.
	\item\label{prop:2-block} For every clause $C_i$ exactly one item $z(i,\alpha)$ is packed into the bin assigned to $C_i$.
\end{enumerate}
These properties imply that for each of the first $k-1$ bins the packed items encode satisfying assignments to the $1/(k-1)$-fraction of clauses assigned to the bin. To enforce these properties, we use a standard construction that adds $\Oh(N/k)$ bookkeeping bits to each item (see the I- and II-blocks in \cref{fig:ETH_to_kBP}).
%
%
It remains to ensure consistency: 
\begin{enumerate}[label=(P\arabic*), ref=(P\arabic*), leftmargin=*, align=left,start=3]
	\itemsep0em
	\item\label{prop:consistent} All items packed into the first $k-1$ bins agree on the assignment to each variable.
\end{enumerate}
Once we can enforce consistency, then the picked clause-assignment items together encode one assignment to the $N$ variables, and this assignment satisfies all clauses, which finishes the reduction.
However, enforcing consistency is the main challenge! Indeed, it requires checking for every pair of clauses $C_i$ and $C_j$ with a shared variable $x$ that the corresponding packed items $z(i, \alpha_i)$ and $z(j, \alpha_j)$ assign the same value to $x$, i.e.~$\alpha_i(x) = \alpha_j(x)$.
Since $C_i$ and $C_j$ may be assigned to different bins, we need some way to communicate between bins, despite the fact that any particular item only adds to the sum of a particular bin. 

We solve this issue in a novel way, using \emph{communication channels} and dummy items that communicate over these channels. 
For each pair of clauses $C_i$ and $C_j$, with $i < j$, that have a shared variable $x$, we assign a block of 3 bits as their communication channel. 
On this bit block, item $z(i, \alpha_i)$ takes value $\bincode{0 0 \alpha_i(x)}$ and item $z(j, \alpha_j)$ takes value $\bincode{0 0 \overline{\alpha_j(x)}}$ (where $\alpha_i(x)$ and $\alpha_j(x)$ are the values assigned to $x$ by $\alpha_i$ and $\alpha_j$, $\overline{z} \coloneq 1-z$ denotes negation, and the subscript 2 indicates numbers in binary). 
Let us focus on the case that $C_i$ and $C_j$ are assigned to different bins $\ell_i \ne \ell_j$, as the case of equal bins is easier. 
In this case, we also construct \emph{dummy items} $d(i,j,x), d'(i,j,x) \in \mathbb N$ that on the communication channel take the values $\bincode{010}$ and $\bincode{011}$, respectively. With some bookkeeping, we ensure that the dummy items can only be packed into bin $\ell_i$ or $\ell_j$. If the assignment items $z(i, \alpha_i)$ and $z(j, \alpha_j)$ agree on the value of $x$, i.e.~$\alpha_i(x) = \alpha_j(x)$, then the communication channel has value $\bincode{000}$ in one of the items and value $\bincode{001}$ in the other. 
By pairing $d(i, j, x)$ with the assignment item of value $\bincode{001}$, and pairing $d'(i, j, x)$ with the assignment item of value $\bincode{000}$, the communication channel sums up to $\bincode{011}$ in both pairs. Thus, by setting the communication channel in both bin capacities $t_{\ell_i}, t_{\ell_j}$ to the value $\bincode{011}$ we enforce \cref{prop:consistent}. See \cref{sec:bin-packing-lowerbound} for details. 

Recall that we can construct numbers with up to $\Theta(N/k) = \Theta(M/(k-1))$ bits. 
By the Sparsification Lemma, each variable appears in $\Oh(1)$ clauses, and thus there are $\Oh(M) = \Oh(N)$ tuples $(C_i, C_j, x)$ which need a communication channel. Hence, we cannot afford to construct one communication channel per tuple -- instead we need to reuse communication channels. 
The crucial observation is that two tuples $(C_i, C_j, x)$ and $(C_{i'}, C_{j'}, x')$ can use the same communication channel, i.e.~the same bit block, if the bins assigned to $C_i$ and $C_j$ are distinct from the bins assigned to $C_{i'}$ and~$C_{j'}$.
We thus need an assignment of communication tuples to channels that ensures this distinctness. We show that a simple greedy assignment works and requires only $\Oh(N/k)$ communication channels in total. 
Hence, in total every item and bin capacity can be described using $\Oh(N/k)$ bits. The complete proof is deferred to \cref{sec:bin-packing-lowerbound}.

\subsection{Generalisation Attempt: Vector Subset Sum}
\label{sec:techoverview_vectorsubsetsum}

The challenge in the reduction presented in the previous section was to ensure communication between bins, despite \emph{each item contributing to the sum of only one bin}. This is also the big obstacle preventing more precise lower bounds (for example, an SETH-based lower bound ruling out time $2^{o(n)} T^{k-1-\eps}$). 

To support our intuition that this is the big obstacle, in this section we consider a problem in which each item contributes to the sum of \emph{every} bin, and show a simple tight SETH-based lower bound. 
Specifically, let us consider the $k$-dimensional \vectorsubsetsum problem:
\Problem{\vectorsubsetsum}
{Set of $n$ vectors $X \subset \N^k$, target vector $t \in \N^k$.}
{Is there a subset $S \subset X$ such that $\Sigma(S) = t$?}
{$T \coloneq \lVert t \rVert_{1}$}
Here, vectors in the subset $S$ can be viewed as items being packed into $k$ bins simultaneously, such that each coordinate of the item contributes to the load of the respective bin.
This problem, can be solved in $\Oh(n
T^k)$ time by a straightforward dynamic programming algorithm. We show
that this is essentially optimal assuming SETH:
\begin{restatable}{theorem}{SETHtoVSS}\label{thm:SETH_to_VectorSubsetSum}
	Assuming SETH, for any $k \ge 1$, \vectorsubsetsum in $k$ dimensions cannot be solved in time $2^{o(n)} T^{k-\varepsilon}$ for any $\eps > 0$.
\end{restatable}
We show this lower bound by proving equivalence of \vectorsubsetsum with the \subsetsum problem, which in turn is known to be SETH-hard by~\cite{AbboudBHS19}. Our equivalence is quite simple and fits on two pages, see \cref{app:vector-subset-sum}. Thus, precise lower bounds are much easier to prove for a problem in which items contribute to every bin.

We remark that \cref{thm:SETH_to_VectorSubsetSum} improves an ETH-based lower bound of $(n + W)^{o(k/\log k)}$ for the $k$-dimensional Knapsack problem by Doron{-}Arad et al.~\cite{DBLP:journals/corr/abs-2407-10146} (specifically, we show a higher and more precise lower bound for a special case of their problem). See \cref{app:vector-subset-sum} for details.

\subsection{SETH-hardness for Scheduling Problems}

This digression motivates us to consider other, less extreme, relaxations of \kpartition. To ensure that items \emph{affect} the load of multiple bins without having to directly contribute to them, we assume that the items are grouped into $G_1, \dots, G_q \subset X$ and we require that items from a group $G_i$ are equally distributed among the $k$ bins. Intuitively, assigning an item to a bin now affects the other bins because we need to distribute the remainder of its group. 
More precisely, we define:
\Problem{\pbtwoYes}
{Sets of integers $G_1, \dots, G_q \subset \N \cap [W(1 - 1/n^{10}), W]$ with $|G_i| = ks$ for every $i \in [q]$ and $n = k s q$ for integers $s, q, W \in \N$.
}
{Decide whether for all $i \in [q]$ there is a partition of $G_i$ into $k$ subsets $S_{i, 1}, \dots, S_{i, k}$ such that
\begin{enumerate}
	\itemsep0em
	\item \label{enum:pb2yes_1a} $|S_{i, \ell}| = s$ for all $i \in [q]$ and $\ell \in [k]$,
	\item \label{enum:pb2yes_1b} $\sum_{i \in [q]} \Sigma(S_{i, \ell}) = \frac{1}{k} \sum_{i \in [q]} \Sigma(G_i)$ for all $\ell \in [k]$.
\end{enumerate}}
{$W$}
This problem admits a pseudopolynomial $n^{\Oh(1)} W^{k-1}$-time algorithm\footnote{
This can be solved by a simple dynamic programming algorithm, similar to~\cref{lem:PSUj_folkloreUB} for \PSUj. We omit the details, as this is also implied by the series of reductions in \cref{sec:scheduling-problems-seth}.}.
We prove the following tight lower bound.

\begin{restatable}{theorem}{SETHtoGkwP}\label{thm:SETH_to_Grouped_kWay_Partition}
	Assuming SETH, for any $k \ge 2$, \pbtwoYes cannot be solved in time $2^{o(n)} W^{k-1-\varepsilon}$ for any $\eps > 0$.
\end{restatable}
Before sketching this lower bound, we demonstrate that \pbtwoYes is an elegant intermediate problem to prove SETH-based lower bounds for several scheduling problems. This discussion will also clarify some of the technical assumptions in the problem definition. We believe that \pbtwoYes will find many further uses as a starting point of SETH-based lower bounds in the future.

\paragraph*{No Tardy Jobs}
Let us sketch a reduction from \pbtwoYes to the special case of \PSUj where we want to decide if $\sum_j U_j = 0$, i.e.~if every job can be finished before its due date.
Let $(G_1,\ldots,G_q)$ be an instance of \pbtwoYes with group size $sk$ and integer range $W$. We build an instance of $\PSUj$ with $n = ksq$ jobs as follows: for every $x \in G_i$ we construct a job with
processing time $x$ and due date $\min\{\mu, i\cdot sW\}$, where $\mu \coloneq \frac{1}{k} \sum_{i=1}^q \Sigma(G_i)$.

Note that the $sk$ jobs corresponding to $G_1$ need to be finished by time $sW$, and all have processing time in $[W(1 - 1/n^{10}), W]$. It follows that each of the $k$ machines needs to process exactly $s$ jobs from $G_1$. This argument can be repeated to show that each machine needs to process exactly $s$ jobs from each group $G_i$. The final due date of $\mu$ ensures that the total load on each bin is $\mu = \frac{1}{k} \sum_{i=1}^q \Sigma(G_i)$. Thus, any schedule that finishes each job before its due date yields a solution of \pbtwoYes. Similarly, it is easy to argue that any solution of \pbtwoYes yields a valid schedule. The complete analysis is provided in \cref{sec:scheduling-problems-seth}. 
We note that this argument crucially uses the technical-looking assumption $G_1, \dots, G_q \subset \N \cap [W(1 - 1/n^{10}), W]$, i.e.~that all items are equal to $W$ up to a factor close to 1.


\medskip
This proves a lower bound for the special case of \PSUj asking for $\sum_j U_j = 0$. 
A lower bound for \PrjCmax follows, because we can turn due dates into release dates by reversing time, which turns the special case of \PSUj asking for $\sum_j U_j = 0$ into \PrjCmax. We further elaborate on these reductions in \cref{sec:scheduling-problems-seth}. Additionally, in \cref{app:folklore} we define the scheduling variants \PLmax and \PTmax and show how a lower bound for those problems also follows.

\paragraph*{Weighted Completion Times}
Our reduction from \pbtwoYes to \PSwjCj is significantly more involved as we need to somehow ``simulate'' due dates by weights. 
We build an instance of \PSwjCj by constructing for every item $x \in G_i$ a job with processing time $x$ and weight $x \cdot n^{10} W + (q-i)$. 
In the analysis, we show that the higher order bits of $\sum_j w_j C_j$, i.e.~the terms corresponding to the $x \cdot n^{10} W$ parts of the weights, are minimised when all loads are perfectly balanced, ensuring \cref{enum:pb2yes_1b} of \pbtwoYes. The lower order bits, corresponding to the $(q-i)$ parts, then ensure balancedness of groups, i.e.~\cref{enum:pb2yes_1a}. This analysis is quite non-trivial, and we defer the details to~\cref{sec:scheduling-problems-seth}.

\paragraph*{Total Processing Time of Tardy Jobs}
For \PSpjUj, a lower bound of $2^{o(n)} T^{k-1-\eps}$ follows immediately now.\footnote{The special case of \PSpjUj asking for $\sum_j p_j U_j = 0$ is equivalent to the special case of \PSUj asking for $\sum_j U_j = 0$, for which we have shown the lower bound above.}
However, current algorithms for \PSpjUj require time $n^{\Oh(1)} T^{k}$~\cite{lawler1969functional,rothkopf1966scheduling}, so our goal is more ambitious: we want to show a lower bound that is higher by one factor $T$.

Intuitively, every machine corresponds to a bin and we want to ``gain one more bin'' by interpreting the set of tardy jobs as being assigned to the dumpster bin. Since we minimise the total processing time of tardy jobs ($\sum_j p_j U_j$), we can easily enforce a bound on the load of this dumpster bin. However, an intricate issue arises in this reduction: The tardy items do not need to adhere to any due dates. This breaks the structure that we previously used to enforce equal splitting of groups. As a matter of fact, we do not see a way to make this reduction work starting from \pbtwoYes.

We therefore consider a relaxation of \pbtwoYes, which we call \pbtwo. 
The main difference is that \pbtwoYes asks for a packing of all items into $k$ bins, while \pbtwo asks to differentiate between the existence of a partial packing of items into $k-1$ bins under certain conditions (\textsc{Yes} case), and the absence of such a packing even under relaxed conditions (\textsc{No} case). 
We note that these cases are non-exhaustive, so \pbtwo is a promise problem where we are guaranteed that the instance is in either one of the cases. 
The relaxed conditions are chosen to make the reduction to \PSpjUj work, proving \cref{thm:SETH_to_PSpjUj}.

\Problem{{\pbtwo}}
{Sets of integers $G_1, \dots, G_q \subset \N \cap [W(1 - 1/n^{10}), W]$ with $|G_i| = ks$ for every $i \in [q]$ and $n = k s q$ for integers $s, q, W \in \N$.
}
{Decide which of the following cases hold.
\begin{yesenum}
	\item $\forall i \in [q]$, there are disjoint subsets $S_{i, 1}, \dots, S_{i, k-1} \subset G_i$ such that
	\begin{yesenum}
		\item $|S_{i, \ell}| = s$ for all $i \in [q]$ and $\ell \in [k-1]$,
		\item $\sum_{i \in [q]} \Sigma(S_{i, \ell}) = \frac{1}{k} \sum_{i \in [q]} \Sigma(G_i)$ for all $\ell \in [k-1]$.
	\end{yesenum}
\end{yesenum}
\begin{noenum}
	\item It does not hold that $\forall i \in [q]$, there are disjoint subsets $S_{i, 1}, \dots, S_{i, k-1} \subset G_i$ such that
	\begin{noenum}
		\item $|S_{1, \ell}| + \dots + |S_{i, \ell}| \leq i \cdot s$ for all $i \in [q]$ and $\ell \in [k-1]$,
		\item $\sum_{i\in [q], \ell \in [k-1]} |S_{i, \ell}| \geq (k-1) q s$,
		\item $\sum_{i \in [q]} \Sigma(S_{i, \ell}) = \frac{1}{k} \sum_{i \in [q]} \Sigma(G_i)$ for all $\ell \in [k-1]$.
	\end{noenum}
\end{noenum}
}
{$W$}
It turns out that we can also prove SETH-tight lower bound for this relaxation of the problem.
\begin{restatable}{theorem}{SETHtoWGkwP}\label{thm:SETH_to_Weak_Grouped_kWay_Partition}
	Assuming SETH, for any $k \ge 2$, \pbtwo cannot be solved in time $2^{o(n)} W^{k-1-\varepsilon}$ for any $\eps > 0$.
\end{restatable}

\noindent
Again we believe that this theorem can easily find further uses as a starting point for lower bounds.

\paragraph{Lower bound for (\textsc{Weak}) \textsc{Grouped \boldmath$k$-way Partition}}
It remains to describe the SETH-based lower bound for (\textsc{Weak}) \pbtwoYes. For this overview, we ignore the differences between the two problems, and focus on \pbtwoYes.
As for \binpacking, we first show a reduction to a problem variant on multisets and with given targets for each bin load, and later prove an equivalence that removes multisets and targets.

We reduce from $K$-SAT, so let $\phi$ be a $K$-CNF formula with $N$ variables and $M$ clauses. 
We will construct an equivalent \pbtwoYes instance. In contrast to the ETH-based lower bound for \binpacking where we could afford $\Oh(N/k) = \Oh(N/(k-1))$ bits per constructed item, now we can only afford $(1+\eps) N/(k-1)$ bits, in particular the overhead from bookkeeping can only be $\eps N / (k-1)$ bits. 
To minimise the bookkeeping, we group the variables $x_1, \dots, x_N$ into supervariables $X_1, \dots, X_{N/a}$ for a large enough constant~$a$.
For each supervariable $X_i$, we create a group $G_i$ containing items $z(i, \alpha)$ for every assignment $\alpha \in \{0, 1\}^a$ of the variables in $X_i$. 
As in the \binpacking reduction, bin $k$ is viewed as a ``dumpster'', and each supervariable is assigned to one of the first $k-1$ bins, so that each bin is assigned $N/(a(k-1))$ supervariables. We use average-free sets from additive combinatorics to ensure consistency of assignments to supervariable~$X_i$, similarly to the \subsetsum lower bound~\cite{AbboudBHS19}.

The main point where we need to diverge from the \subsetsum lower bound is that we cannot group clauses into superclauses, because different clauses might be satisfied due to (super-)variables assigned to different bins. In that case, a superclause item would need to communicate with multiple bins. To avoid this, we therefore construct items for single clauses, specifically we create an item $y(j, i, \alpha)$ for every clause~$C_j$, supervariable $X_i$, and assignment $\alpha \in \{0, 1\}^a$ to $X_i$ that satisfies $C_j$. We need to ensure picking exactly one item $y(j,i,\alpha)$ for each $j$, meaning that each clause points to some supervariable that causes it to be satisfied. For every clause $C_j$, the items $y(j, i, \alpha)$ are put in the group $G_{N/a + j}$. Hence in total, we construct groups $G_1, \dots, G_{N/a + M}$.

The grouping into supervariables helps to reduce the number of bookkeeping bits that enforce picking exactly one item $z(i,\alpha)$ for each clause $C_i$. However, since we cannot consider superclauses, the number of bookkeeping bits to ensure picking exactly one item $y(j, i, \alpha)$ for each clause $C_j$ would be too much. 
Circumventing this is where we rely on the group structure of \pbtwoYes: We add $s(k-1)-1$ dummy items to each group $G_{N/a + j}$, which ensures that at least one item $y(j, i, \alpha)$ must be packed into the first $k-1$ bins for each $C_j$. We furthermore make the dummy items favourable, in the sense that any solution needs to pack all dummy items into the first $k-1$ bins. This then enforces that exactly one item $y(j, i, \alpha)$ must be packed for each $C_j$. 
The details of the reduction are deferred to \cref{sec:SETH_lowerbound}. 

\subsection{Organisation of the paper}
After preliminaries in \cref{sec:preliminaries}, the rest of the paper is dedicated to the reductions summarised in \cref{fig:graph_reduction}. 
In \cref{sec:bin-packing-lowerbound}, we show the ETH-based lower bound for \kpartitionTargetsMultisets. This is shown to be equivalent to \binpacking in \cref{sec:binpacking_equivalence}, thus proving \cref{thm:bin-packing-lowerbound}. Next, in \cref{sec:SETH_lowerbound} we present the SETH-based lower bound for a variant of \pbtwo and in \cref{sec:equivalences_pbtwo} we show equivalence of the variants. Finally, \cref{sec:scheduling-problems-seth} is dedicated to reducing \pbtwo and \pbtwoYes to various scheduling problems, thus proving \cref{thm:SETH_to_PSwjCj,thm:SETH_to_PSpjUj,thm:SETH_to_PrjCmax,cor:SETH_to_PSUj}. 

In \cref{app:problem-definitions}, we provide all problem definitions. 
Further applications of our results are discussed in \cref{app:further-applications}.
The equivalence between \vectorsubsetsum and \subsetsum is given in \cref{app:vector-subset-sum}. Additional observations on scheduling problems are discussed in \cref{app:folklore}.
Finally, \cref{app:technical-lemmas} contains proofs of technical lemmas.

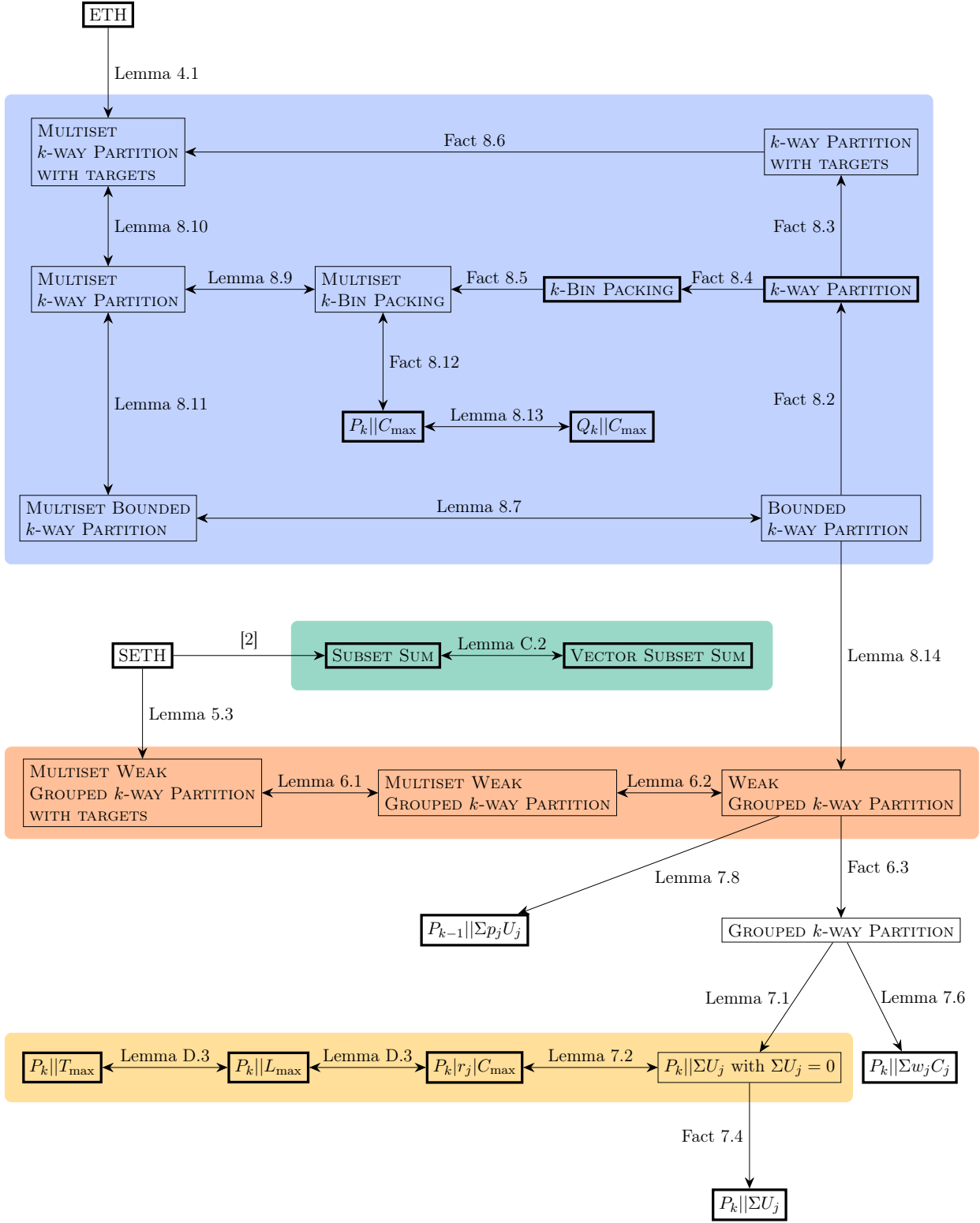
\begin{figure}[!t]
	\centering
		\resizebox{\textwidth}{!}{%
	\begin{tikzpicture}[yscale=-1, >={Stealth[length=2.4mm, width=2mm]}]
		\node[draw, ultra thick] at (-15,0) (ETH) {{ETH}};
		
		\node[draw] at (1,3) (kwayPtarget) {\begin{varwidth}{5cm}\hyperlink{prob:kpartitiontargets}{\textsc{$k$-way Partition} \\ \textsc{with targets}}
		\end{varwidth}};
		\node[draw] at (-15,3) (kPtargetMulti) {\begin{varwidth}{5cm}\hyperlink{prob:kpartitiontargetsmultisets}{\textsc{Multiset} \\ \textsc{$k$-way Partition} \\ \textsc{with targets}}\end{varwidth}};
		\node[draw,ultra thick] at (-4,9) (QkCmax) {\begin{varwidth}{5cm}\QCmax\end{varwidth}};
		\node[draw,ultra thick] at (-9,9) (PkCmax) {\begin{varwidth}{5cm}\PCmax\end{varwidth}};
		\node[draw] at (-9,6) (BPMulti) {\begin{varwidth}{5cm}\hyperlink{prob:binpacking}{\textsc{Multiset} \\ \textsc{$k$-Bin Packing}} \end{varwidth}};
		\node[draw,ultra thick] at (-4,6) (BP) {\begin{varwidth}{5cm}{\binpacking}\end{varwidth}};
		\node[draw,ultra thick] at (1,6) (Part) {\begin{varwidth}{5cm}{\kpartition}\end{varwidth}};
		\node[draw] at (-15,6) (PartMulti) {\begin{varwidth}{5cm}\hyperlink{prob:kpartition}{\textsc{Multiset} \\ \textsc{$k$-way Partition}} \end{varwidth}};
		\node[draw] at (1,11) (BoundedPart) {\begin{varwidth}{5cm}\hyperlink{prob:kpartitionbounded}{\textsc{Bounded} \\ \textsc{$k$-way Partition} }\end{varwidth}};
		\node[draw] at (-15,11) (BoundedPartMulti) {\begin{varwidth}{5cm} \hyperlink{prob:kpartitionbounded}{\textsc{Multiset} \textsc{Bounded} \\ \textsc{$k$-way Partition}} \end{varwidth}};

		\node[draw, ultra thick] at (-14.25,14) (SETH) {\begin{varwidth}{5cm}SETH\end{varwidth}};

		\node[draw,ultra thick] at (-12 +3,14) (SubsetSum) {{\subsetsum}};
		\node[draw,ultra thick] at (-6 +3,14) (VectorSubsetSum) {\vectorsubsetsum};

		\node[draw] at (-14.25,18-1) (Pb2TargetMultiset) {\begin{varwidth}{5cm}\hyperlink{prob:pbtwotargetsmultisets}{\textsc{Multiset Weak}  \\ \textsc{Grouped $k$-way Partition} \\ \textsc{with targets}}\end{varwidth}};
		\node[draw] at (-6.5,18-1) (Pb2Multiset) {\begin{varwidth}{5cm}\hyperlink{prob:pbtwo}{\textsc{Multiset Weak}  \\ \textsc{Grouped $k$-way Partition}}\end{varwidth}};
		\node[draw] at (1,18-1) (Pb2) {\begin{varwidth}{5cm}\hyperlink{prob:pbtwo}{\textsc{Weak} \\ \textsc{Grouped $k$-way Partition}}\end{varwidth}};
		\node[draw] at (1,21-1) (Pb2Yes) {\begin{varwidth}{5cm}{\pbtwoYes}\end{varwidth}};

		\node[draw] at (-1,23) (PSUjSpecial) {\begin{varwidth}{5cm}\PSUj with $\Sigma U_j = 0$\end{varwidth}};
		\node[draw,ultra thick] at (2.5,23) (PSwjUj) {\PSwjCj};
		\node[draw,ultra thick] at (-7, 20) (PSpjUj) {\PSpjUj[k-1]};
		\node[draw,ultra thick] at (-7,23) (PrjCmax) {\PrjCmax};
		\node[draw,ultra thick] at (-11.5,23) (PLmax) {\PLmax};
		\node[draw,ultra thick] at (-16,23) (PTmax) {\PTmax};
		\node[draw,ultra thick] at (-1,26) (PSUj) {\PSUj};

		\draw (SETH)   edge[->] node[above] {\cite{AbboudBHS19}} (SubsetSum);                    
		\draw (SubsetSum)   edge[<->] node[above] {{\cref{lem:SubsetSum_equiv_VectorSubsetSum}}} (VectorSubsetSum);

		\draw (ETH)   edge[->] node[right] {{\cref{lem:ETH_to_BPtargets}}} (kPtargetMulti);
		\draw (QkCmax)   edge[<->] node[above] {\cref{lem:QkCmax_equiv_PkCmax}} (PkCmax);
		\draw (PkCmax)   edge[<->] node[right] {\cref{lem:PkCmax_equiv_MultiBP}} (BPMulti);
		\draw (BPMulti)   edge[<->] node[above] {\cref{lem:MultiBP_to_MultiPart}} (PartMulti);
		\draw (PartMulti)   edge[<->] node[right] {\cref{lem:MultiPart_to_MultiBoundedPart}} (BoundedPartMulti);
		\draw (BoundedPartMulti)   edge[<->] node[above] {{\cref{lem:MultiBoundedPart_equiv_BoundedPart}}} (BoundedPart);
		\draw (BoundedPart)   edge[->] node[left] {\cref{lem:BoundedPart_to_Part}} (Part);
		\draw (Part)   edge[->] node[above] {\cref{lem:Part_to_BP}} (BP);
		\draw (Part)   edge[->] node[left] {\cref{lem:Part_to_TargetPart}} (kwayPtarget);
		\draw (BP)   edge[->] node[above] {\cref{lem:binpacking_to_binpackingMultisets}} (BPMulti);
		\draw (kwayPtarget)   edge[->] node[above] {\cref{lem:TargetPart_to_TargetPartMultisets}} (kPtargetMulti);
		\draw (kPtargetMulti)   edge[<->] node[right] {\cref{lem:MultiParttargets_equiv_MultiPart}} (PartMulti);

		\draw (BoundedPart)   edge[->] node[right] {\cref{lem:BoundedPart_to_WeakGroupedPart}} (Pb2);

		\draw (SETH)   edge[->] node[right] {{\cref{lem:SETH_to_GroupedkBP}}} (Pb2TargetMultiset);
		\draw (Pb2TargetMultiset)   edge[<->] node[above] {{\cref{lem:pbtwotargets_to_pbtwomulti}}} (Pb2Multiset);
		\draw (Pb2Multiset)   edge[<->] node[above] {{\cref{lem:pbtwomulti_to_pbtwo}}} (Pb2);
		\draw (Pb2)   edge[->] node[right] {{\cref{lem:pbtwo_to_pbtwoYES}}} (Pb2Yes);
		
		\draw (Pb2Yes)   edge[->] node[right] {{\cref{lem:GroupedkPart_to_PSwjCj}}} (PSwjUj);
		\draw (Pb2Yes)   edge[->] node[left] {{\cref{lem:GroupedkPart_to_special_schedule}}} (PSUjSpecial);
		\draw (PSUjSpecial)   edge[->] node[left] {\cref{lem:GroupedkPart_to_PSUj}} (PSUj);
		\draw (PSUjSpecial)   edge[<->] node[above] {{\cref{lem:GroupedkPart_to_PrjCmax}}} (PrjCmax);
		\draw (PrjCmax)   edge[<->] node[above] {{\cref{lem:PSUj0_Lmax_Tmax}}} (PLmax);
		\draw (PLmax)   edge[<->] node[above] {{\cref{lem:PSUj0_Lmax_Tmax}}} (PTmax);
		\draw (Pb2)   edge[->] node[below right] {{\cref{lem:WeakGroupedkPart_to_PSpjUj}}} (PSpjUj);

	\begin{scope}[on background layer]
		\draw[cb_blue!40, fill, rounded corners] (3, 1.75) rectangle (-17.25, 12);    
		\draw[cb_orange!40, fill, rounded corners] (4, 17-1) rectangle (-17.25, 19-1);    
		\draw[cb_yellow!40, fill, rounded corners] (-17.25, 22.25) rectangle (1.25, 23.75);    
		\draw[OI_green!40, fill, rounded corners] (-11, 13.25) rectangle (-.5, 14.75);    
	\end{scope}

	\end{tikzpicture}
	}
	\caption{Summary of parameter-preserving reductions. Equivalent problems are highlighted in the same colour.
	}\label{fig:graph_reduction}
\end{figure}


\section{Preliminaries}\label{sec:preliminaries}

For any integer $n$, we let $[n] \coloneq \{1, 2, \dots, n\}$. For a set or multiset $X$, we denote by $\Sigma(X)$ the sum of all elements in $X$. A \emph{partition} of a set or multiset $X$ is a collection of disjoint subsets whose union, denoted with the $\uplus$ operator, is $X$. All problems studied are formally defined in \cref{app:problem-definitions}.

Throughout the whole paper, we assume that the number of machines or bins $k$ is constant, and thus hide functions of $k$ in $\Oh(\cdot)$-notation. All logarithms are in base 2 unless stated otherwise. We use a subscript ``$\bincode{ }$'' to denote binary numbers, e.g.~$\bincode{1001} = 9$.
%

\smallskip
\begin{definition}[Parameter-preserving reduction]
Let $A$ and $B$ be decision problems with associated input sizes $n_A$ and $n_B$ and parameters $T_A$ and $T_B$, respectively. 
A \emph{parameter-preserving reduction} from $A$ to $B$ is an algorithm $\mc A$ that, given an instance $I$ of $A$ of size $n_A$ and parameter $T_A$, computes in time $(n_A+\log{T_A})^{\Oh(1)}$ an equivalent instance $\mc A(I)$ of $B$ of size $n_B = \Oh(n_A + \log T_A)$ and parameter $T_B 
	 \le T_A \cdot n_A^{\Oh(1)}$.
	
\end{definition}

\begin{lemma}\label{lem:parameter-preserving-reduction}
	Assume there is a parameter-preserving reduction from problem $A$ to problem $B$. Let $\gamma > 0$. If there exists $\eps > 0$ such that for all $\delta >0$
	problem $B$ on instances of size $n$ and parameter $T$ can be solved in time $\Oh(2^{\delta n} T^{\gamma-\eps})$, then there exists $\eps > 0$ such that for all $\delta > 0$ problem $A$ on instances of size $n$ and parameter $T$ can be solved in time $\Oh(2^{\delta n} T^{\gamma-\eps})$. 
\end{lemma}
\begin{proof}
	Let $I_A$ be an instance of problem $A$ of size $n_A$ and parameter $T_A$. 
	We want to show that for some $\eps_A > 0$ and for any $\delta_A>0$ 
	we can solve $I_A$ in time $\Oh(2^{\delta_A n_A} T_A^{\gamma-\eps_A})$.
	
	By assumption, there exists a global constant $C \ge 1$, such that we can compute in time $\Oh((n_A+\log{T_A})^{C})$ an equivalent instance $I_B = \mc A(I_A)$ of problem $B$ of size $n_B \le C\cdot (n_A+\log{T_A})$ and parameter $T_B \leq T_A \cdot n_A^{C}$. If problem $B$ on instances of size $n_B$ and parameter $T_B$ can be solved in time $\Oh(2^{\delta_B n_B} T_B^{\gamma-\eps_B})$ for some $\eps_B > 0$ and any $\delta_B > 0$, then in particular for $\delta_B := \frac{1}{C} \min\{\delta_A/2,\eps_A\}$ and $\eps_A := \eps_B/3$, we can solve $I_A$ in time
	\begin{align*}
		\Oh(2^{\delta_B n_B} T_B^{\gamma-\eps_B}) +  \Oh((n_A+\log{T_A})^C) &
		\le \Oh\left(2^{\delta_A n_A/2 + \eps_A \log(T_A)} \cdot T_A^{\gamma - 3\eps_A} \cdot n_A^{C(\gamma - 3\epsilon_A)} \cdot (n_A+\log{T_A})^{C}\right)\\
		&\le \Oh\left(2^{\delta_A n_A/2} \cdot T_A^{\epsilon_A}\cdot T_A^{\gamma - 3\eps_A} \cdot (n_A+\log{T_A})^{(\gamma+1)C}\right)\\
		&\le \Oh\left(2^{\delta_A n_A/2} \cdot T_A^{\gamma - 2\eps_A} \cdot n_A^{(\gamma+1)C}  \cdot (\log{T_A})^{(\gamma+1)C}\right).
	\end{align*}
	Since $\gamma, C, \delta_A > 0$ are constants, asymptotically in $n_A$ we have $n_A^{(\gamma + 1)C} \leq \Oh(2^{\delta_A n_A / 2})$. Similarly, we have $(\log{T_A})^{(\gamma+1)C} \leq \Oh(T_A^{\eps_A})$.
	Hence, the above running time is at most $\Oh(2^{\delta_A n_A} T_A^{\gamma-\eps_A})$. 
\end{proof}

\paragraph*{Hardness Assumptions} The Exponential Time Hypothesis (ETH) and the Strong Exponential Time Hypothesis (SETH) are hypothesis about the running time of algorithms for the $k$-SAT problem: Given a boolean CNF formula $\phi$ with $N$ variables, where each clause consists of at most $k$ literals, the task is to decide whether there is a satisfying assignment for $\phi$. 
The Exponential Time Hypothesis (ETH) is the conjecture that there exists $\delta >0$ such that there is no algorithm solving 3-SAT in time $\Oh(2^{\delta N})$~\cite{impagliazzo2001complexity}.
The Strong Exponential Time Hypothesis (SETH) states that for any $\eps > 0$, there exists an integer $K \geq 3$ such that $K$-SAT cannot be solved in time $\Oh(2^{(1-\eps)N})$~\cite{ImpagliazzoPZ01}.

We use the following standard tool by Impagliazzo, Paturi and Zane~\cite{ImpagliazzoPZ01}:

\begin{lemma}[Sparsification Lemma~{\cite{ImpagliazzoPZ01}, \cite[Lemma 16.17]{flum2006parameterized}}]
	\label{lem:sparsification} 
	For every integer $K \geq 3$ and every $\lambda > 0$, there exist a constant
	$\Delta(K,\lambda) > 0$ and an algorithm that, given a $K$-SAT instance $\Phi$ with $N$ variables,
	produces $K$-SAT instances $\phi_1,\ldots,\phi_r$ with $r \leq 2^{\lambda N}$ such that $\Phi$ is satisfiable if and only if at least one $\phi_i$ is satisfiable. Moreover, $\phi_i$ has $N$ variables and in $\phi_i$ each variable appears in at most $\Delta(K,\lambda)$ clauses, for each $i$. The algorithm runs in time $2^{\lambda N} |\Phi|^{\Oh(1)}$, where $|\Phi|$ is the number of bits needed to represent the input formula.
\end{lemma}

\medskip
Finally, we observe that if the input integers are bounded in a certain range, then we can leverage the following property. We will show that this assumption can be made without any loss of generality in \cref{thm:equivalence-bin-packing}.

\begin{observation}\label{obs:bounded_subset_size}
	Let $W,n,k$ be positive integers. Let $X \subset \N \cap [W(1 - 1/n^2), W]$ be a multiset of size $n$ 
	and let $X_1 \uplus \ldots \uplus X_k = X$ be a partition of $X$ such that $\Sigma(X_i) = \Sigma(X)/k$
	for every $i \in [k]$. Then $|X_i| = n/k$ for every $i \in [k]$. In particular, $n$ is divisible by $k$.
\end{observation}
\begin{proof}
	If $k=1$ or $n=1$ the statement is trivial. Assume $k \geq 2$ and $n \geq 2$.
	For every $i \in [k]$, we have
	\begin{align*}
		\frac{n}{k} W &\ge  \Sigma(X) / k = \Sigma(X_i) \ge |X_i| \cdot W(1 - 1/n^2)
	\intertext{and}
		\frac{n}{k} W(1-1/n^2) &\le  \Sigma(X)/k = \Sigma(X_i) \le |X_i| \cdot W.
	\end{align*}
	So for every $i \in [k]$, we deduce that $|X_i| \in 
	\left[\frac{n}{k} (1 - \frac{1}{n^2}), \frac{n}{k} / (1 - \frac{1}{n^2})\right]
	 =: I$.
	As $|I| < 1$ and $|X_i|$ is an integer, we have that $|X_1| = |X_2| = \ldots = |X_k|$. However, $n = |X| = \sum_{i=1}^k |X_i|$, which means that $|X_i| = n/k$ for every $i \in [k]$.
\end{proof}



\section{ETH-hardness of Bin Packing}
\label{sec:bin-packing-lowerbound}

In this section, we prove a tight ETH-based lower bound for \kpartitionTargetsMultisets[k], see \Cref{app:problem-definitions} for the formal problem definition. This is the main step in the proof of our tight ETH-based lower bound for \binpacking (\Cref{thm:bin-packing-lowerbound}); the remaining step of proving equivalence of \binpacking and its variant \kpartitionTargetsMultisets is postponed to \Cref{sec:binpacking_equivalence}.

\begin{lemma}\label{lem:ETH_to_BPtargets}
	Assuming ETH, there exists a constant $\delta >0$ such that for every $k \geq 2$
	\kpartitionTargetsMultisets cannot be solved in time $\Oh(T^{\delta k} \cdot 2^{\delta n})$, where $T$ is the maximum target and $n$ is the number of items.
\end{lemma}

\begin{proof}
	Let $\Phi$ be a 3-CNF formula on $N$ variables.  We apply the Sparsification
	Lemma (\Cref{lem:sparsification}) on $\Phi$ with parameter $\lambda > 0$ to be tuned later, i.e.~in
	$2^{\lambda N} \cdot |\Phi|^{\Oh(1)}$ time we construct 3-CNF formulas $\phi_1,
	\dots, \phi_r$ such that $\Phi = \phi_1 \vee \dots \vee \phi_r$, where $r \le
	2^{\lambda N}$, and in each formula $\phi_q$ 
	every variable appears in at most $\Delta = \Delta(3,\lambda)$
	clauses, in particular $\phi_q$ has $M_q \le
	\Delta N$ clauses.
	We can assume that each $M_q$ is divisible by $k-1$, 
	by adding trivial clauses if necessary.  
	Fix any $k \ge 2$. For every $q \in [r]$, we construct an instance of
	\kpartitionTargetsMultisets that is equivalent to $\phi_q$, in time polynomial in the input size $|\phi_q|$. To simplify notation, in what follows we write $\phi = \phi_q$ and $M = M_q$.

	At a high level, this construction works as follows. For every clause $C_i$ and for every assignment $\alpha \in \{0, 1\}^3$ of its variables satisfying $C_i$ we create an item $z(i, \alpha) \in \mathbb N$ called \emph{assignment item}. We refer to items assigned to one of the first $k-1$ bins with targets $t_1,\ldots,t_{k-1}$ as ``packed'' items. 
	If $\phi$ is satisfiable, we want to ensure that the packed assignment items correspond to an assignment of the $N$ variables that satisfies $\phi$. The $k$-th bin of value $t_k$ is viewed as a ``dumpster'' that contains all remaining assignment items.
	We therefore need to ensure that exactly one item $z(i, \alpha)$ per clause $C_i$ is packed, and that all packed assignment items have a consistent variable assignment. 
	The first property can be achieved with a standard bookkeeping bit encoding. 
	For the latter property, we consider all tuples $(C_i, C_j, x)$, where $x$ is a variable appearing in both clauses $C_i$ and $C_j$. If $z(i, \alpha_i)$ and $z(j, \alpha_j)$ are the corresponding packed assignment items, then we want to ensure $x$ is assigned the same value by $\alpha_i$ and by $\alpha_j$. Hence, we construct a \emph{communication channel} for every such tuple $(C_i, C_j, x)$ along with \emph{dummy items} $d(i, j, x)$ and $d'(i, j, x)$. Since every variable appears in at most $\Delta$ clauses, every clause appears in at most $3\Delta$ such tuples. This will allow us to bound the total number of communication channels. We next describe the full construction of the \kpartitionTargetsMultisets instance before analysing it. \Cref{tab:ETH_parameters} summarises notations used in the construction.

	\begin{table}[!t]
\centering
\begin{tabular}{ll@{\hskip .7cm}l}
\toprule
Notation & Description & Value \\
\midrule
$\lambda$ & parameter of the Sparsification Lemma  & $\in (0,1)$ \\
$\Delta$ & max \# of clauses per variable & $\Delta(\lambda) \in \mathbb{N}$ \\
$N$ & \# of variables & \\
$M$ & \# of clauses (after sparsification) & $\le \Delta N$ \\
& & \\
$k$ & \# of bins & $\ge 2$ \\
$n$ & total number of items & $\le 13 \Delta^2 N$ \\
$G_\ell$ & group of $M/(k-1)$ clauses & $\ell \in [k-1]$\\
$t_\ell$ & target value for the $\ell$-th bin & $\ell \in [k]$ \\
& & \\
$C_i$ & $i$-th clause & $i \in [M]$ \\
$\ell(i)$ & group assigned to $C_i$ & $= \lceil i \cdot \frac{k-1}{M} \rceil$ \\
$p(i)$ & position of $C_i$ in its group & $= i \bmod \frac{M}{k-1}$ \\
$\alpha_i$ & assignment of variables in $C_i$ & $\in \{0,1\}^3$ \\
$z(i, \alpha)$ & assignment item for $C_i$ and $\alpha$ & $\in \mathbb{N}$ \\
$d(i, j, x)$, $d'(i,j,x)$ & dummy items for tuple $(C_i, C_j, x)$ & $\in \mathbb{N}$ \\
\bottomrule
\end{tabular}
\caption{Notations used in the reduction from 3-SAT to \kpartitionTargetsMultisets in \cref{lem:ETH_to_BPtargets}.}
\label{tab:ETH_parameters}
\end{table}

	\paragraph*{Construction}
	First, we group the clauses into $k-1$ groups $G_1, \dots, G_{k-1}$ where 
	for every $\ell \in [k-1]$ we have
	$G_\ell =\{C_{(\ell-1)M/(k-1) + 1},\ldots, C_{\ell M/(k-1)}\}$. Let $\ell(i) \coloneq \lceil i \cdot \frac{k-1}{M} \rceil$ and $p(i) \coloneq (i \bmod \frac{M}{k-1})$ be the functions that assign clause $C_i$ to its group $G_{\ell(i)}$ and to its position $p(i)$ in the group $G_{\ell(i)}$, respectively.
	We call a tuple $(C_i, C_j, x)$ a \textbf{\emph{communication tuple}} if $i < j$ and the variable $x$ appears in both clauses $C_i$ and $C_j$. The communication tuple is said to be \textbf{\emph{internal}} if $C_i$ and $C_j$ belong to the same group, i.e.~$\ell(i) = \ell(j)$, and \textbf{\emph{external}} otherwise.
	For every clause $C_i$ and every assignment $\alpha \in \{0, 1\}^3$ of its variables that satisfies $C_i$ we create an \textbf{\emph{assignment item}} $z(i, \alpha)$. Additionally, for every external communication tuple $(C_i, C_j, x)$ we create two \textbf{\emph{dummy items}} $d(i, j, x)$ and $d'(i, j, x)$. Let $n$ be the total number of created items. 

	We describe the items and targets $t_1, \dots, t_{k-1}$  by describing blocks of their bits, from highest to lowest order. The last target $t_{k}$ is defined such that the total sum of all the items equals $t_1 + \dots + t_{k}$. See \cref{fig:ETH_to_kBP,fig:ETH_to_kBP_CC} for an illustration of the construction.

	\subparagraph*{I-blocks}
	The first $k-1$ blocks of $10 \lceil \log{M} \rceil$ bits, called \emph{I-blocks}, ensure that assignment items of the clause $C_i \in G_\ell$ can only be packed into the $\ell$-th bin (or the $k$-th bin). 
	Every assignment item $z(i, \alpha)$ has value 1 on the $\ell(i)$-th I-block and value 0 on every other I-block. 
	All dummy items have value 0 on all I-blocks.
	Every target $t_\ell$ for $\ell \in [k-1]$ has value $\frac{M}{k-1}$ on the $\ell$-th I-block and value 0 on every other I-block.

	\subparagraph*{II-blocks}
	The following $M/(k-1)$ bits are called \emph{II-blocks} and guarantee that exactly one assignment item per clause can be packed into the first $k-1$ bins.
	Every assignment item $z(i, \alpha)$ has value 1 on the $p(i)$-th II-block and value 0 on every other II-block. 
	All dummy items have all II-blocks set to 0. 
	All targets $t_1, \dots, t_{k-1}$ have all II-blocks set to 1.

	\subparagraph*{III-blocks}
	Next, $\binom{k-1}{2}$ \emph{III-blocks} of $10 \lceil \log{M} \rceil$ bits enforce that dummy items $d(i, j, x)$ and $d'(i, j, x)$ can only go into the $\ell(i)$-th or the $\ell(j)$-th bin (or the $k$-th bin).
	We index the $\binom{k-1}{2}$ III-blocks by pairs $(\ell_1, \ell_2) \in [k-1]^2$ with $\ell_1 < \ell_2$. 
	Let $E_{\ell_1, \ell_2}$ be the number of external communication tuples $(C_i, C_j, x)$ with $\ell(i) = \ell_1$ and $\ell(j) = \ell_2$. 
	All III-blocks of any assignment item have value 0. 
	Dummy items $d(i, j, x)$ and $d'(i, j, x)$ have value 1 on the III-block indexed by $(\ell(i), \ell(j))$ and value 0 on every other III-block.
	Every target $t_\ell$ for $\ell \in [k-1]$ has value $E_{\ell_1, \ell_2}$ in the III-block indexed by $(\ell_1, \ell_2)$ if $\ell \in \{\ell_1, \ell_2\}$, and value 0 otherwise.

	\input{figures/ETH_to_kBP.tex}	

		\begin{figure}
		\centering
		\resizebox{.5\textwidth}{!}{%
		 \begin{tikzpicture}[
			node distance=0pt,
        	box/.style={rectangle,draw,minimum width=1.1cm,minimum height=.8cm},
        	padding/.style={rectangle,draw,minimum width=.2cm, minimum height=.8cm, fill=gray!80},
        	value/.style={yshift=1cm}]

		\begin{scope}[yshift=-5.5cm,xshift=9cm]
			

		\node[box,fill=OI_green] (zi) at (0, 0) {$ 0 \ 0 \ 
		\makebox[0pt][l]{$\alpha(x)$}\phantom{\overline{\beta(x)}} 
		$
		};
		\node [left=1cm of zi] {$z(i, \alpha)$};
		\node [above=.8cm of zi] {$\ell(i) \neq \ell(j)$};

		\node[box,fill=OI_green] (zj) [below=.2cm of zi] {$ 0\ 0\ \overline{\beta(x)}$};
		\node [left=1cm of zj] {$z(j, \beta)$};

		\node[box,fill=OI_green] (dij) [below=.2cm of zj] { $0 \ 1 \ \makebox[0pt][l]{0}\phantom{\overline{\beta(x)}} $
		};
		\node [left=1cm of dij] {$d(i, j, x)$};

		\node[box,fill=OI_green] (dijp) [below=.2cm of dij] { $0 \ 1 \ \makebox[0pt][l]{1}\phantom{\overline{\beta(x)}} $
		};
		\node [left=1cm of dijp] {$d'(i, j, x)$};

		\node[box,fill=OI_green] (ti) [below=.2cm of dijp] { $0 \ 1 \ \makebox[0pt][l]{1}\phantom{\overline{\beta(x)}} $
		};
		\node [left=1cm of ti] {$t_{\ell(i)}, t_{\ell(j)}$};

		\node[box,fill=OI_green] (zie) [right=2cm of zi] {$ 0 \ 0 \ 
		\makebox[0pt][l]{$\alpha(x)$}\phantom{\overline{\beta(x)}} 
		$
		};
		\node [above= .8cm of zie] {$\ell(i) = \ell(j)$};

		\node[box,fill=OI_green] (zje) [below=.2cm of zie] {$ 0\ 0\ \overline{\beta(x)}$};

		\node [below=.2cm of zje,minimum width=1.1cm,minimum height=.8cm] { no dummy item 
		};

		\phantom{\node[box,fill=OI_green] (dije) [below=.2cm of zje] { $0 \ 1 \ \makebox[0pt][l]{0}\phantom{\overline{\beta(x)}} $
		};}

		\node [below=.2cm of dije,minimum width=1.1cm,minimum height=.8cm] { no dummy item 
		};

		\phantom{\node[box,fill=OI_green] (dijpe) [below=.2cm of dije] { $0 \ 1 \ \makebox[0pt][l]{1}\phantom{\overline{\beta(x)}} $
		};}

		\node[box,fill=OI_green] (tie) [below=.2cm of dijpe] { $0 \ 0 \ \makebox[0pt][l]{1}\phantom{\overline{\beta(x)}} $
		};
        
        \draw[dashed] ([yshift=(-.2cm), xshift=(-1cm)]tie.south west) -- ([yshift=(+.2cm), xshift=(-1cm)]zie.north west);

		\end{scope}
    
    \end{tikzpicture}
    }%
	\caption{
        The communication channel assigned to an external (left) and an internal (right) communication tuple $(C_i, C_j, x)$ in \cref{lem:ETH_to_BPtargets} for items $z(i, \alpha)$, $z(j, \beta)$, $d(i, j, x)$, $d'(i, j, x)$ and the targets $t_{\ell(i)}$ and $t_{\ell(j)}$. 
	}
	\label{fig:ETH_to_kBP_CC}
	\end{figure}

	\subparagraph*{Communication Channels}
	The last $6\Delta\frac{M}{k-1}$ blocks of $3$ bits are called \emph{communication channels} and ensure that assignment items packed into the first $k-1$ bins are consistent (see \cref{fig:ETH_to_kBP_CC}). 
	To assign communication tuples to communication channels we proceed as follows. 
	Each communication channel can be \emph{open} or \emph{full} on each bin. Initially, every communication channel is open on every bin. For every communication tuple $(C_i, C_j, x)$, greedily pick a communication channel that is open on bins $\ell(i)$ and $\ell(j)$, and assign the tuple to that channel. The channel then becomes full for the bins $\ell(i)$ and $\ell(j)$. We claim that an open communication channel can always be found. Indeed, every variable in $\phi$ is contained in at most $\Delta$ clauses, so every clause $C_i$ is contained in at most $3\Delta$ communication tuples, and so every bin $\ell \in [k-1]$ has at most $3\Delta \frac{M}{k-1}$ communication tuples $(C_i, C_j, x)$ with $\ell \in \{\ell(i), \ell(j)\}$. Since there are $2 \cdot 3\Delta \frac{M}{k-1}$ communication channels, for every communication tuple $(C_i, C_j, x)$ we can find a channel that is open on $\ell(i)$ and on $\ell(j)$. 
	
	After the above assignment procedure, for each communication channel left open on bin $\ell \in [k-1]$, the target $t_\ell$ and assignment and dummy items for the clauses in $G_\ell$ 
	take value $\bincode{000}$ (written in binary). We now define the values (written in binary) of the communication channel assigned to a communication tuple $(C_i, C_j, x)$ on the bins $\ell(i)$ and $\ell(j)$.
	Every assignment item $z(i, \alpha)$ of the clause $C_i$ has value $\bincode{00\mathnormal{\alpha(x)}}$, where $\alpha(x) = 1$ if the variable $x$ is assigned true by $\alpha$ and $\alpha(x) = 0$ otherwise.
	Every assignment item $z(j, \beta)$ of the clause $C_j$ has value $\bincode{00\overline{\mathnormal{\beta(x)}}}$, where $\overline{\beta(x)} = 1 - \beta(x)$ is the flipped value assigned to $x$ by $\beta$.
	If $(C_i, C_j, x)$ is an \emph{internal} communication tuple, then the target $t_{\ell(i)} = t_{\ell(j)}$ takes value $\bincode{001}$.
	If $(C_i, C_j, x)$ is an \emph{external} communication tuple, then both targets $t_{\ell(i)}$ and $t_{\ell(j)}$ take value $\bincode{011}$, and the dummy items $d(i, j, x)$ and $d'(i, j, x)$ take values $\bincode{010}$ and $\bincode{011}$, respectively (see~\cref{fig:ETH_to_kBP_CC}).

	\subparagraph*{Padding blocks}

	Finally, we add $10 \lceil \log M \rceil$  \emph{padding bits} after every I-block and III-block and after the last II-block. All the padding bits have value $0$ and the purpose is to prevent any overflow between blocks when summing up items. 

	\paragraph*{Correctness}
	We prove that the constructed \kpartitionTargetsMultisets instance has a solution if and only if $\phi$ is satisfiable.
	Suppose that there exists an assignment $\alpha \in \{0, 1\}^N$ satisfying $\phi$. For every $i \in [M]$, let $\alpha_i \in \{0, 1\}^3$ be the assignment of $\alpha$ to the variables in $C_i$. 
	We distribute assignment and dummy items into $k$ sets $S_1, \dots, S_k$ as follows. For every $i \in [M]$, put item $z(i, \alpha_i)$ into $S_{\ell(i)}$ and all items $z(i, \alpha_i')$ with $\alpha_i' \neq \alpha_i$ into $S_k$. For every external communication tuple $(C_i, C_j, x)$, if $\alpha_i(x) = 0$ put item $d(i, j, x)$ into $S_{\ell(i)}$ and item $d'(i, j, x)$ into $S_{\ell(j)}$; if $\alpha_i(x) = 1$ put item $d(i, j, x)$ into $S_{\ell(j)}$ and item $d'(i, j, x)$ into $S_{\ell(i)}$. We now verify that $\Sigma(S_\ell) = t_\ell$ for every $\ell \in [k-1]$. By definition of $t_k$, this implies that $\Sigma(S_{k}) = t_{k}$ as well and therefore that there exists a solution to the \kpartitionTargetsMultisets instance.
	Fix $\ell \in [k-1]$. 
	Since $S_\ell$ contains $M/(k-1)$ assignment items, all corresponding to clauses $C_i$ with $\ell(i) = \ell$, the I-blocks of $\Sigma(S_\ell)$ match the I-blocks of $t_\ell$. Additionally, for every clause $C_i$ with $\ell(i) = \ell$, $S_\ell$ contains exactly one assignment item, so $\Sigma(S_\ell)$ and $t_\ell$ also match on the II-blocks.
	Now for every external communication tuple $(C_i, C_j, x)$ with $\ell \in \{\ell(i), \ell(j)\}$, exactly one dummy item of $\{d(i, j, x), d'(i, j, x)\}$ is contained in $S_\ell$ and thus the III-blocks also match. Additionally, consider the assignment item $z(i, \alpha_i)$ of $C_i$ packed into $S_\ell$. If $\alpha_i(x) = 1$, then $S_\ell$ contains both $z(i, \alpha_i)$ and $d(i, j, x)$ that take respective values $\bincode{001}$ and $\bincode{010}$ in the assigned communication channel. If $\alpha_i(x) = 0$, then $S_\ell$ contains both $z(i, \alpha_i)$ and $d'(i, j, x)$ that take respective values $\bincode{000}$ and $\bincode{011}$ in the assigned communication channel. In both cases, the assigned communication channel of $\Sigma(S_\ell)$ has value $\bincode{011}$, matching the value of $t_\ell$.
	Finally, for every internal communication tuple $(C_i, C_j, x)$, consider the assignment items $z(i, \alpha_i)$ and $z(j, \alpha_j)$ packed in $S_\ell$. Since $\alpha_i(x) = \alpha_j(x)$, the assigned communication channel of $\Sigma(S_\ell)$ has value $\bincode{000} + \bincode{001} = \bincode{001}$, which matches the value in $t_\ell$.
	Since there is no overflow, padding blocks of $\Sigma(S_\ell)$ have value 0 and thus also match $t_\ell$. 

	\medskip
	Conversely, assume that there exists a partition of all the items into $k$ sets $S_1, \dots, S_{k}$ such that $\Sigma(S_\ell) = t_\ell$ for all $\ell \in [k]$. 
	Let $Z_i$ be the set of assignment items for the clause~$C_i$, and $Z \coloneq \bigcup_{i =1}^M Z_i$ be the set of all assignment items. Let $D_{\{\ell_1, \ell_2\}}$ be the set of dummy items for all external communication channels $(C_i, C_j, x)$ such that $\{\ell(i), \ell(j)\} = \{\ell_1, \ell_2\}$, and let $D \coloneq \bigcup_{\ell_1, \ell_2 \in [k-1] \text{ s.t.\ } \ell_1 \neq \ell_2} D_{\{\ell_1, \ell_2\}}$ be the set of all dummy items. The goal is to prove that items in $Z \cap (S_1 \cup \dots \cup S_{k-1})$ have a consistent assignment of the variables $x_1, \dots, x_N$ that satisfies $\phi$. Note that per definition  assignment items satisfy the corresponding clause. Hence, we only need to ensure that the assignments are consistent and that there is one assignment item per clause.

	Since all padding blocks have value 0 and length greater than
	$\log n$ (as $n \ll 2^{10 \lceil \log M \rceil}$), we deduce that there is no overflow from lower bits over any padding block. It follows that every I-block has the same value in $\Sigma(S_{\ell})$ and in $t_\ell$ for each $\ell \in [k-1]$. The same holds for every III-block, for the union of all II-blocks, and for the union of all communication channels.
	In what follows, we analyse these parts
	sequentially by starting with the higher order bits. 

	Fix $\ell \in [k-1]$. 
	Since every I-block of $\Sigma(S_\ell)$ is equal to the corresponding I-block of $t_\ell$, we have $|S_\ell \cap Z| = \frac{M}{k-1}$ and $S_{\ell} \cap Z_i = \emptyset$ for every $i \in [M]$ such that $\ell(i) \neq \ell$. In other words, each of the first $k-1$ bins packs exactly $\frac{M}{k-1}$ assignment items, and assignment items packed in the $\ell$-th bin correspond to clauses of the group $G_\ell$. 
	Furthermore, the padding block after the last II-block ensures that the only items in $S_{\ell}$ contributing to the II-blocks in $\Sigma(S_{\ell})$ are the $\frac{M}{k-1}$ packed assignment items. 
	Since each of these $\frac{M}{k-1}$ items contributes at most one 1-bit to the II-block and the target asks for $\frac{M}{k-1}$ 1-bits, these 1-bits must appear at distinct positions. It follows that $S_{\ell}$ contains exactly one assignment item for each clause from the group $G_\ell$.
	Let $z(i, \alpha_i)$ be the packed item for the clause $C_i$, i.e.~the sole element of $S_{\ell(i)} \cap Z_i$.
	Next, note that the padding blocks ensure that the only items in $S_{\ell}$ that contribute to III-blocks are dummy items, and every III-block of a dummy item can only contribute to the same III-block of $\Sigma(S_{\ell})$. The non-zero III-blocks of $t_\ell = \Sigma(S_{\ell})$ are indexed by $(\ell_1, \ell_2)$ such that $\ell \in \{\ell_1, \ell_2\}$. So we deduce that $D_{\{\ell_1, \ell_2\}} \subset S_{\ell_1} \cup S_{\ell_2}$ for every $\ell_1, \ell_2 \in [k-1]$ with $\ell_1 \neq \ell_2$.

	
	Next, we claim that there is no overflow between communication channels in each bin $S_{\ell}$.
	Indeed, by the distribution of channels, a channel can be used by two communication tuples $(C_i, C_j, x)$ and $(C_{i'}, C_{j'}, x')$ only if $\{\ell(i), \ell(j)\} \cap \{\ell(i'), \ell(j')\} = \emptyset$. This means that the packed assignment items (and the potential dummy items) corresponding to each tuple are packed in disjoint bins. 
	Hence, if we assume no overflow from the lower bits, the only items that can contribute to a given communication channel of $\Sigma(S_{\ell})$ are the packed assignment items (and potential dummy items) of the assigned communication tuple $(C_i, C_j, x)$. Note that, in each communication channel, assignment items take values at most $\bincode{001}$ and dummy items take values $\bincode{010}$ and $\bincode{011}$. In total, the contribution of $(C_i, C_j, x)$ to the assigned communication channel of $\Sigma(S_{\ell})$ is therefore at most $\bincode{111}$. 
	Hence, by induction on the communication channels from lower to higher order bits, there is no overflow between communication channels in $\Sigma(S_{\ell})$.
	
	Finally, fix a communication tuple $(C_i, C_j, x)$ and consider the assigned communication channel on $\ell(i)$ and $\ell(j)$. 
	If the communication channel is internal, i.e.~$\ell(i) = \ell(j) = \ell$, the only items contributing to the assigned channel in $\Sigma(S_\ell)$ are $z(i, \alpha_i)$ and $z(j, \alpha_j)$. Observe that $\Sigma(S_\ell)$ has value
	$\bincode{00\alpha_i(x)} + \bincode{0 0 (1-\alpha_j(x))}$ and $t_\ell$ has value $\bincode{001}$ on that channel, and these values are equal if and only if $\alpha_i(x) = \alpha_j(x)$. 
	Now consider an external communication channel, i.e.~$\ell(i) \neq \ell(j)$. Then the only items contributing to the assigned channel in $\Sigma(S_{\ell(i)})$ and $\Sigma(S_{\ell(j)})$ are $z(i, \alpha_i)$, $z(j, \alpha_j)$, $d(i, j, x)$ and $d'(i, j, x)$. 
	If $\alpha_i(x) \neq \alpha_j(x)$, then $z(i, \alpha_i)$ and $z(j, \alpha_j)$ have the same value in the communication channel: either both are $\bincode{000}$ or both are $\bincode{001}$. Since the dummy items have values $\bincode{010}$ and $\bincode{011}$, the sum of the items $z(i, \alpha_i)$, $z(j, \alpha_j)$, $d(i, j, x)$ and $d'(i, j, x)$ on the communication channel is an odd number. 
	Recall that the targets $t_{\ell(i)}$ and $t_{\ell(j)}$ have value $\bincode{011}$ in that channel, so they sum to an even number $\bincode{110}$.
	Hence, it is impossible to partition these items to add to the target values.
	Therefore, we necessarily have $\alpha_i(x) = \alpha_j(x)$.\footnote{In that case, one of the assignment items takes value $\bincode{001}$ and the other value $\bincode{000}$. By pairing them with the dummy items taking value $\bincode{010}$ and $\bincode{011}$ respectively, we get that $S_{\ell(i)}$ and $S_{\ell(j)}$ both evaluate to $\bincode{011}$ in the channel.} Since this applies for every communication tuple $(C_i, C_j, x)$, we deduce that all assignment items in $Z \cap (S_1 \cup \dots \cup S_{k-1})$ have a consistent variable assignment $\alpha \in \{0, 1\}^N$. Since assignment items are created only for variable assignments satisfying the corresponding clause, the assignment $\alpha$ satisfies all clauses of $\phi$, and thus $\phi$ is satisfiable.

	\paragraph*{Parameters of the reduction}
	Recall that every variable appears in at most $\Delta$ clauses of the 3-CNF formula $\phi$, so there are $M \le \Delta N$ clauses and at most $3 \Delta M \le 3 \Delta^2 N$ communication tuples. For each communication tuple we create at most 2 dummy items, and for each clause we create at most $7$ assignment items, so in total the number of items is $n \leq 7 \cdot M  + 2 \cdot 3 \Delta M \leq 13 \Delta^2 N$. 
	
	The number of bits in each constructed item is at most (i) $20 (k-1) \lceil \log{M} \rceil$ for I-blocks and their padding blocks, (ii) $M/(k-1) + 10 \lceil \log M \rceil$ for II-blocks and their single padding block, (iii) $20 \binom{k-1}{2} \lceil \log{M} \rceil$ for III-blocks and their padding blocks, and (iv) $6 \Delta \frac{M}{k-1} \cdot 3$ for communication channels. In total, this number of bits is at most 
	\begin{align*}
		 M / (k-1) + 18 \Delta M/(k-1) + \left(10+ 20 (k-1) + 20 \binom{k-1}{2} \right)\lceil \log M \rceil&\le\\ 20 \Delta^2 N/(k-1) + 50 (k-1)^2 \lceil \log M \rceil &=: B. 
	\end{align*}
	Therefore, every target has value at most 
	$$T \le n \cdot 2^B \le 13 \Delta^2 N \cdot 2^{20 \Delta^2 N/(k-1)} \cdot (2M)^{50 (k-1)^2} \le 2^{20 \Delta^2 N/(k-1)} \cdot (26 \Delta^2 N)^{51 (k-1)^2}.$$
	 
	Assuming ETH, there exists a constant $\delta_{\text{ETH}} > 0$ such that 3-SAT cannot be solved in time $2^{\delta_{\text{ETH}} \cdot N} \cdot |\Phi|^{\Oh(1)}$. 
	We set the parameter of the Sparsification Lemma to $\lambda \coloneq \delta_{\text{ETH}} / 3$, which determines $\Delta = \Delta(3,\lambda)$, and based on that we set $\delta \coloneq \delta_{\text{ETH}} / (100 \Delta^2)$. 
	For the sake of contradiction, assume that there exists $k \ge 2$ such that \kpartitionTargetsMultisets can be solved in time $\Oh(T^{\delta k} 2^{\delta n})$. Then our reduction yields an algorithm for 3-SAT with running time bounded by
	\begin{align*}
		2^{\lambda N} \cdot |\Phi|^{\Oh(1)} \cdot T^{\delta k} 2^{\delta n}
		&\leq 2^{\lambda N} \cdot |\Phi|^{\Oh(1)} \cdot T^{2 \delta (k-1)} 2^{\delta 13 \Delta^2 N} \\
		&\leq 2^{N (\lambda + 13 \Delta^2 \delta + 40 \Delta^2 \delta)} \cdot (26 \Delta^2 N)^{102 (k-1)^3 \delta} \cdot |\Phi|^{\Oh(1)} \\
		&\leq 2^{N (\lambda + 13 \Delta^2 \delta + 40 \Delta^2 \delta)} \cdot |\Phi|^{\Oh(1)} \tag{as $k,\delta, \Delta = \Oh(1)$ and $|\Phi| \ge N$}\\
		&\leq 2^{N (\delta_{\text{ETH}} / 3 + 53 \Delta^2 \cdot \delta_{\text{ETH}} / (100 \Delta^2))} \cdot |\Phi|^{\Oh(1)} \\
		&\leq 2^{0.9 \cdot \delta_{\text{ETH}} \cdot N} \cdot |\Phi|^{\Oh(1)}.
	\end{align*}
	This contradicts the initial assumption that 3-SAT cannot be solved in time $2^{\delta_{\text{ETH}} \cdot N} \cdot |\Phi|^{\Oh(1)}$, and thus we arrive at a contradiction with ETH.
\end{proof}

\begin{proof}[Proof of \cref{thm:bin-packing-lowerbound}]
	As \kpartitionTargetsMultisets and \binpacking are equivalent by~\cref{thm:equivalence-bin-packing}, the lower bound for \kpartitionTargetsMultisets given by \cref{lem:ETH_to_BPtargets} directly implies the same lower bound for \binpacking.
\end{proof}

The proof of~\cref{thm:equivalence-bin-packing} is deferred to \cref{sec:binpacking_equivalence}.

\section{SETH-hardness of a Variant of Weak Grouped \boldmath$k$-way Partition}\label{sec:SETH_lowerbound}

In this section, we prove a tight SETH-based lower bound for \pbtwoTargetsMultisets.
\Problem{\pbtwoTargetsMultisets}
{Multisets $G_1, \dots, G_q \subset \N \cap [W(1 - 1/n^{10}), W]$ with $|G_i| = ks$ for every $i \in [q]$ and $n = k s q$ for integers $s, q, W \in \N$; and targets $t_1, \dots, t_{k-1} \in \N$.}
{Decide which of the following cases hold.
\begin{yesenum}
	\item  $\forall i \in [q]$, there are disjoint subsets $S_{i, 1}, \dots, S_{i, k-1} \subset G_i$ such that:
	\begin{yesenum}
		\item $|S_{i, \ell}| = s$ for all $i \in [q]$ and $\ell \in [k-1]$,
		\item $\sum_{i \in [q]} \Sigma(S_{i, \ell}) = t_\ell$ for all $\ell \in [k-1]$.
	\end{yesenum}
\end{yesenum}
\begin{noenum}
	\item It does not hold that $\forall i \in [q]$, there are disjoint subsets $S_{i, 1}, \dots, S_{i, k-1} \subset G_i$ such that:
	\begin{noenum}
		\item $|S_{1, \ell}| + \dots + |S_{i, \ell}| \leq i \cdot s$ for all $i \in [q]$ and $\ell \in [k-1]$,
		\item $\sum_{i\in [q], \ell \in [k-1]} |S_{i, \ell}| \geq (k-1) q s$,
		\item $\sum_{i \in [q]} \Sigma(S_{i, \ell}) = t_\ell$ for all $\ell \in [k-1]$.
	\end{noenum}
\end{noenum}
}
{$W$}
Later, in \Cref{sec:equivalences_pbtwo}, we further reduce to the problems \pbtwo and \pbtwoYes, and in \Cref{sec:scheduling-problems-seth} we present reductions to multi-machine scheduling problems, cf.~\cref{fig:graph_reduction}.

We will need the following slight variation of the classic construction of Behrend sets~\cite{behrend1946sets}, which has seen many uses in algorithms and complexity, see, e.g.~\cite{DellM14,AbboudLW14,AbboudB17,AbboudBHS19}.

\begin{definition}[Strong $k$-average-free set]
	A set $B \subset \N$ is \emph{strong $k$-average-free} if for any $a,b \in \{0,1, \dots, k\}$ and $x_1, \dots, x_{a}, x \in B$, if $x_1 + \dots + x_{a} = b x$ then $a = b$ and $x_1 = \dots = x_{a} = x$. 
\end{definition}
\begin{restatable}{lemma}{lembehrend}\label{lem:behrend}
	For any $\mu \in (0, 1)$, $k \geq 2$ and $n \geq 1$, a
    strong $k$-average-free set $B \subset [U]$ of size $n$ with $U = n^{1 + \mu} k^{\Oh(1/\mu)}$ can be constructed in $n^{\Oh(1)}$ time.
\end{restatable}

As this is a simple modification of an already existing construction, we include the proof of \cref{lem:behrend} in \Cref{app:technical-lemmas} for completeness.

\begin{lemma}\label{lem:SETH_to_GroupedkBP}
	Assuming SETH, for any $\epsilon >0$ and $k \ge 2$ there exists $\delta > 0$ such that \pbtwoTargetsMultisets cannot be solved in $\Oh(2^{\delta n}W^{k -1 - \epsilon})$ time, where $n$ is the number of items and $W$ is the upper bound on the integer range.
\end{lemma}

\begin{proof}
Suppose for the sake of contradiction that there exist $k \ge2$ and $\eps \in (0,1)$ such that for all $\delta > 0$ \pbtwoTargetsMultisets can be solved in time $\Oh(2^{\delta n} W^{k-1-\eps})$. We will write the running time as $\Oh(2^{\delta n} W^{(k-1)(1-\bar{\eps})})$ for $\bar{\eps} \coloneq \eps/(k-1)$.

\begin{table}[!t]
\centering
\begin{tabular}{ll@{\hskip .7cm}l}
\toprule
Notation & Description & Value \\
\midrule
$K$ & \# of literals per clause & $\geq 3$\\
$\lambda$ & parameter of the Sparsification Lemma & $\in (0,1)$\\
$\Delta$ & max \# of clauses per variable & $\Delta = \Delta(K, \lambda) \in \mathbb{N}$ \\
$N$ & \# of variables &  \\
$M$ & \# of clauses (after sparsification) & $\le\Delta N$\\

& & \\

$X_i$ & supervariable (i.e.\ set of variables) & $i \in [N/a]$\\
$a$ & \# of variables in a supervariable &  \\
$k$ & \# of bins & $k$-th bin is a ``dumpster'' \\
$q$ & \# of groups & $=N/a + M$ \\
$s$ & group size divided by $k$ & $=2^a \cdot \max\{\Delta a, K\}$ \\
$n$ & total \# of items & $=ksq$ \\
$W$ & upper bound on all constructed items &  \\
$G_i$ & $i$-th group of items &  $i \in [q]$ \\
$S_{i,\ell}$ & items from $G_i$ packed in the $\ell$-th bin & $i \in [q], \ell \in [k-1]$ \\
$t_\ell$ & target value for the $\ell$-th bin& $\ell \in [k-1]$ \\

& & \\

$\ell(i)$ & bin assigned to $X_i$ & $=\ceil{i \frac{(k-1)a}{N}}$ \\
$p(i)$ & position of $X_i$ in its bin & $=i \bmod N/((k-1)a) $ \\
$\alpha_i$ & assignment of variables in $X_i$ &  $\in \{0,1\}^a$ \\
$\gamma_i$ & \# of clauses satisfied by variables in $X_i$ &  $\in \{0, \ldots,\Delta a\}$ \\
$z(i,\alpha)$ & assignment item for $X_i$ and $\alpha$ & $\in [W]$ \\
$y(j,i,\alpha)$ & assignment item for $C_j$ satisfied by $X_i \gets \alpha$ & $\in [W]$ \\
$d(i)$ & dummy item for $X_i$ & $\in [W]$ \\

& & \\

$\mu$ & parameter in strong average-free set & $\mu = \mu(K,\lambda) \in (0,1)$ \\
$U$ & maximum value in strong average-free set & $= 2^{a(1+\mu)}(\Delta a)^{\Oh(1/\mu)}$ \\
$B(\alpha)$ & $\alpha$-th element of strong average-free set & $\in [U]$ \\
\bottomrule
\end{tabular}
\caption{Notations used in the reduction from $K$-SAT to \pbtwoTargetsMultisets in \cref{lem:SETH_to_GroupedkBP}.}
\label{tab:SETH_parameters}
\end{table}

Let $K \geq 3$ and $\Phi$ be a $K$-CNF formula on $N$ variables. We apply the Sparsification
Lemma (\Cref{lem:sparsification}) on~$\Phi$ with parameter $\lambda \coloneq \bar{\eps}/10$, i.e.~in $2^{\lambda
N} |\Phi|^{\Oh(1)}$ time we construct $K$-CNF formulas $\phi_1, \dots, \phi_r$
such that $\Phi = \phi_1 \vee \dots \vee \phi_r$, where $r = 2^{\lambda N}$, and in each formula $\phi_q$ every variable appears in at most $\Delta = \Delta(K,\lambda)$ clauses. In particular, $\phi_q$ has at most $M \coloneq \Delta N$ clauses, and we can duplicate clauses to ensure that the number of clauses is equal to $M = \Delta N$.

We let $a \coloneq a(K,\lambda,\bar{\eps}) \in \mathbb{N}$, $ a \ge 2$, be a sufficiently large parameter depending only on $K,\lambda$, and~$\bar{\eps}$; we will tune $a$ later. We can assume that $N$ is a multiple of $(k-1) \cdot a$ by adding dummy variables.
Define:
\begin{align*}
	q \coloneq N/a +M, \quad s \coloneq 2^a \cdot \max\{\Delta a, K\}, \quad n \coloneq ksq. 
\end{align*}
Using \cref{lem:behrend} with $\mu \coloneq \bar{\eps}/10$, we
construct strong $(\Delta a)$-average-free set $B \subset [U]$ of size $|B| = 2^a$ with $U = 2^{a (1 + \mu)} (\Delta a)^{\Oh(1/\mu)}$. 
For $\alpha \in \{0, 1\}^a$, we denote by $B[\alpha]$ the $i$-th element of $B$, where $i = \sum_{j=1}^a \alpha[j] \cdot 2^{j-1}$. 

For every $\phi \in \{\phi_1, \dots, \phi_r\}$, and every tuple $(\gamma_1, \dots, \gamma_{N/a}) \in \{0, \dots, \Delta a\}^{N/a}$ with $\sum_{i \in [N/a]} \gamma_i = M$, we construct an instance $(G_1,\ldots,G_q)$ of \pbtwoTargetsMultisets of group size $ks$ such that $\phi$ is satisfiable if and only if at least one of the constructed instances is in the \ref{enum:pb2targets_yes} case. 
In the following, for any collection of disjoint subsets $S_{i, 1}, \dots, S_{i, k-1} \subset G_i$, we interpret the items in $\bigcup_{i = 1}^q S_{i, \ell}$ as \emph{packed} into the $\ell$-th bin, so that $\sum_{i=1}^q \Sigma(S_{i, \ell})$ is the \emph{load} of the $\ell$-th bin.  
We next describe the full construction of the instance $(G_1,\ldots,G_q)$ before analysing it. \cref{tab:SETH_parameters} summarises notations used in the construction.

\paragraph*{Construction}
First, group the variables into $N/a$ \textbf{\emph{supervariables}} $X_i \coloneq (x_{(i-1)a+1}, \dots, x_{ia})$ for $i \in [N/a]$. 
We say that the supervariable $X_i$ appears in the clause $C_j$ if $X_i$ contains a variable appearing in $C_j$. 
Note that every variable $x_i$ (for $i \in [N]$) appears in at most $\Delta$ clauses, so every supervariable $X_i$ (for $i \in [N/a]$) appears in at most $\Delta a$ clauses. 

We group the supervariables into $(k-1)$ parts such that the supervariable $X_i$ is at the $p(i)$-th position of the $\ell(i)$-th part, where $p(i) \coloneq (i \bmod \frac{N}{(k-1)a})$ and $\ell(i) \coloneq \lceil i \frac{(k-1)a}{N} \rceil$ for all $i \in [N/a]$. 
Intuitively, supervariables of the $\ell$-th part are assigned to the $\ell$-th bin of \pbtwoTargetsMultisets.
%
For every supervariable $X_i$, we construct a \textbf{\emph{dummy item}} $d(i) \in \N$, and for every assignment $\alpha \in \{0, 1\}^a$ of the variables in $X_i$ we construct an \textbf{\emph{assignment item}} $z(i, \alpha) \in \N$.
For $i \in [N/a]$, let $G_i$ be the multiset containing every assignment item $z(i, \alpha)$ and $(k-1)s - 1$ copies of $d(i)$. Note that there are $2^a \leq s$ assignment items for $X_i$, so we can add enough copies of an arbitrary assignment item $z(i, \alpha)$
to $G_i$ until $|G_i| = ks$. 
For every clause $C_j$, we construct a \textbf{\emph{dummy item}} $d(N/a +j) \in \N$, and for every supervariable $X_i$ appearing in $C_j$ and every assignment $\alpha \in \{0, 1\}^a$ of variables in $X_i$ that satisfies $C_j$ we construct an \textbf{\emph{assignment item}} $y(j, i, \alpha) \in \N$. 
For $j \in [M]$, let $G_{N/a + j}$ be the multiset containing every assignment item $y(j, i, \alpha)$ and $(k-1)s-1$ copies of $d(N/a + j)$. 
Note that there are at most $K 2^a \leq s$ assignment items for $C_j$ so we can add enough copies of an arbitrary assignment item to $G_{N/a +j}$ until $|G_{N/a + j}| = ks$. 
We now describe the items and the targets $t_1, \dots, t_{k-1} \in \N$ by describing blocks of their bits, from highest to lowest order. See \cref{fig:SETH_to_GroupedkBP} for an illustration.


\begin{figure}[!t]
    \centering
    \resizebox{\textwidth}{!}{%
        \begin{tikzpicture}[
        node distance=0pt,
        box/.style={rectangle,draw,minimum width=1.1cm,minimum height=1.1cm},
        padding/.style={rectangle,draw,minimum width=.2cm, minimum height=1.1cm, fill=gray!80},
        value/.style={yshift=1cm}]

    \begin{scope}
    \node[box,fill=cb_orange] (I) at (0,0) {$1$};
    \node[padding] (pI) [right=of I] { };
    \node [left=of I] {$z(i, \alpha) \in G_i$};

    \node[box,fill=OI_lightblue] (II) [right=of pI] {$\phantom{\sum_{b=1}^q bs} \makebox[0pt][r]{i}$};
    \node[padding] (pII) [right=of II] { };

    \node[box,fill=OI_yellow] (II1) [right=of pII] {$0$};
    \node[padding] (pII1) [right=of II1] { };
    \node[value] [above of=II1] {1};

    \node[box,fill=OI_yellow] (II2) [right=of pII1] {$0$};
    \node[value] [above of=II2] {2};
    \node[padding] (pII2) [right=of II2] { };

    \node[rectangle,minimum height=.8cm] (II5) [right=of pII2] {$\cdots$};
    \node[padding] (pII5) [right=of II5] { };

    \node[box,fill=OI_yellow] (II3) [right=of pII5] {$\phantom{\frac{N}{(k-1)a} + \sum\limits_{i'  \ : \ \ell(i') = \ell(i)} \gamma_i} \makebox[0pt][r]{1}$};
    \node[value] [above of=II3] {{$\ell(i)$}};
    \node[padding] (pII3) [right=of II3] { };

    \node[rectangle] (II6) [right=of pII3] {$\cdots$};
    \node[padding] (pII6) [right=of II6] { };

    \node[box,fill=OI_yellow] (II4) [right=of pII6] {$0$};
    \node[value] [above of=II4] {$k-1$};
    \node[padding] (pII4) [right=of II4] { };

    \node[value] [above of=pII4] {//};

    \node[box,fill=OI_green] (III1) [right=of pII4] {$0$};
    \node[value] [above of=III1] {$1$};

    \node[box,fill=OI_green] (III2) [right=of III1] {$0$};
    \node[value] [above of=III2] {$2$};

    \node[rectangle] (III5) [right=of III2] {$\cdots$};
    
    \node[box,fill=OI_green] (III3) [right=of III5] {$\Delta a U - \gamma_i B(\alpha)$};
    \node[value] [above of=III3] {{$p(i)$}};

    \node[rectangle] (III6) [right=of III3] {$\cdots$};

    \node[box,fill=OI_green] (III4) [right=of III6] {$0$};
    \node[value] [above of=III4] {$\frac{N}{(k-1)a}$};

    \end{scope}

    \begin{scope}[yshift=-1.5cm]

    \node[box,fill=cb_orange] (I) at (0,0) {$1$};
    \node[padding] (pI) [right=of I] { };
    \node [left=of I] {$y(j, i, \alpha) \in G_{\frac{N}{a} + j} $ };

    \node[box,fill=OI_lightblue] (II) [right=of pI] {$\phantom{\sum_{b=1}^q bs}\makebox[0pt][r]{$\frac{N}{a} + j$}$ };
    \node[padding] (pII) [right=of II] { };

    \node[box,fill=OI_yellow] (II1) [right=of pII] {$0$};
    \node[padding] (pII1) [right=of II1] { };

    \node[box,fill=OI_yellow] (II2) [right=of pII1] {$0$};
    \node[padding] (pII2) [right=of II2] { };

    \node[rectangle] (II5) [right=of pII2] {$\cdots$};
    \node[padding] (pII5) [right=of II5] { };

    \node[box,fill=OI_yellow] (II3) [right=of pII5] {$\phantom{\frac{N}{(k-1)a} + \sum\limits_{i'  \ : \ \ell(i') = \ell(i)} \gamma_i}\makebox[0pt][r]{1}$};
    \node[padding] (pII3) [right=of II3] { };

    \node[rectangle] (II6) [right=of pII3] {$\cdots$};
    \node[padding] (pII6) [right=of II6] { };

    \node[box,fill=OI_yellow] (II4) [right=of pII6] {$0$};
    \node[padding] (pII4) [right=of II4] { };

    \node[box,fill=OI_green] (III1) [right=of pII4] {$0$};

    \node[box,fill=OI_green] (III2) [right=of III1] {$0$};

    \node[rectangle] (III5) [right=of III2] {$\cdots$};
    
    \node[box,fill=OI_green] (III3) [right=of III5] {$\phantom{\Delta a U - \gamma_i B(\alpha)}\makebox[0pt][r]{$B(\alpha)$}$};

    \node[rectangle] (III6) [right=of III3] {$\cdots$};

    \node[box,fill=OI_green] (III4) [right=of III6] {$0$};
    \end{scope}

    \begin{scope}[yshift=-3cm]

    \node[box,fill=cb_orange] (I) at (0,0) {$1$};
    \node[padding] (pI) [right=of I] { };
    \node [left=of I] {$d(b) \in G_b$ };

    \node[box,fill=OI_lightblue] (II) [right=of pI] {$\phantom{\sum_{b=1}^q bs}\makebox[0pt][r]{$b$}$ };
    \node[padding] (pII) [right=of II] { };

    \node[box,fill=OI_yellow] (II1) [right=of pII] {$0$};
    \node[padding] (pII1) [right=of II1] { };

    \node[box,fill=OI_yellow] (II2) [right=of pII1] {$0$};
    \node[padding] (pII2) [right=of II2] { };

    \node[rectangle] (II5) [right=of pII2] {$\cdots$};
    \node[padding] (pII5) [right=of II5] { };

    \node[box,fill=OI_yellow] (II3) [right=of pII5] {$\phantom{\frac{N}{(k-1)a} + \sum\limits_{i'  \ : \ \ell(i') = \ell(i)} \gamma_i}\makebox[0pt][r]{0}$};
    \node[padding] (pII3) [right=of II3] { };

    \node[rectangle] (II6) [right=of pII3] {$\cdots$};
    \node[padding] (pII6) [right=of II6] { };

    \node[box,fill=OI_yellow] (II4) [right=of pII6] {$0$};
    \node[padding] (pII4) [right=of II4] { };

    \node[box,fill=OI_green] (III1) [right=of pII4] {$0$};

    \node[box,fill=OI_green] (III2) [right=of III1] {$0$};

    \node[rectangle] (III5) [right=of III2] {$\cdots$};

    \node[box,fill=OI_green] (III3) [right=of III5] {$\phantom{\Delta a U - \gamma_i B(\alpha)}\makebox[0pt][r]{0} $};

    \node[rectangle] (III6) [right=of III3] {$\cdots$};

    \node[box,fill=OI_green] (III4) [right=of III6] {$0$};
    \end{scope}

    \begin{scope}[yshift=-4.5cm]

    \node[box,fill=cb_orange] (I) at (0,0) {$sq$};
    \node[padding] (pI) [right=of I] { };
    \node [left=of I] {$t_{\ell(i)}$ };

    \node[box,fill=OI_lightblue] (II) [right=of pI] {$\sum_{b=1}^q bs$};
    \node[padding] (pII) [right=of II] { };

    \node[box,fill=OI_yellow] (II1) [right=of pII] {$0$};
    \node[padding] (pII1) [right=of II1] { };

    \node[box,fill=OI_yellow] (II2) [right=of pII1] {$0$};
    \node[padding] (pII2) [right=of II2] { };

    \node[rectangle] (II5) [right=of pII2] {$\cdots$};
    \node[padding] (pII5) [right=of II5] { };

    \node[box,fill=OI_yellow] (II3) [right=of pII5] {$\frac{N}{(k-1)a} + \sum\limits_{i'  \ : \ \ell(i') = \ell(i)} \gamma_i$};
    \node[padding] (pII3) [right=of II3] { };

    \node[rectangle] (II6) [right=of pII3] {$\cdots$};
    \node[padding] (pII6) [right=of II6] { };

    \node[box,fill=OI_yellow] (II4) [right=of pII6] {$0$};
    \node[padding] (pII4) [right=of II4] { };

    \node[box,fill=OI_green] (III1) [right=of pII4] {$\Delta a U$};

    \node[box,fill=OI_green] (III2) [right=of III1] {$\Delta a U$};

    \node[rectangle] (III5) [right=of III2] {$\cdots$};

    \node[box,fill=OI_green] (III3) [right=of III5] {$\phantom{\Delta a U - \gamma_i B(\alpha)}\makebox[0pt][r]{$\Delta a U$} $};
    
    \node[rectangle] (III6) [right=of III3] {$\cdots$};

    \node[box,fill=OI_green] (III4) [right=of III6] {$\Delta a U$};
    \end{scope}

    \begin{scope}[yshift=-6.5cm]

    \node[box,fill=cb_orange] (I) at (0, 0) { };
    \node [align=center, below=.5cm of I] {I-block\\ $10\ceil{\log n}$ bits};

    \node[box,fill=OI_lightblue] (II) [right=2.5cm of I] {};
    \node [align=center,below=.5cm of II] {II-block\\$10 \ceil{\log n}$ bits};

    \node[box,fill=OI_yellow] (III) [right=2.5cm of II] {};
    \node [align=center,below=.5cm of III] {III-block\\ $10 \ceil{\log n}$ bits};


    \node[box,fill=OI_green] (C) [right=2.5cm of III] {};
    \node [align=center,below=.5cm of C] {IV-block \\ $\lceil \log(2\Delta a U) \rceil$ bits};

    \node[padding] (p) [right=3cm of C] { };
    \node [align=center,below=.5cm of p] {Padding block \\ $10\ceil{\log n}$ bits};
    
\end{scope}

\end{tikzpicture}
}%

\caption{Construction of the \pbtwoTargetsMultisets instance in \cref{lem:SETH_to_GroupedkBP}. We show the bit blocks of items and targets, where the highest bits are on the left, for $i \in [N/a]$, $j \in [M]$, $\alpha \in \{0, 1\}^{N/a}$, $b \in [N/a +M] = [q]$.
}
\label{fig:SETH_to_GroupedkBP}
\end{figure}
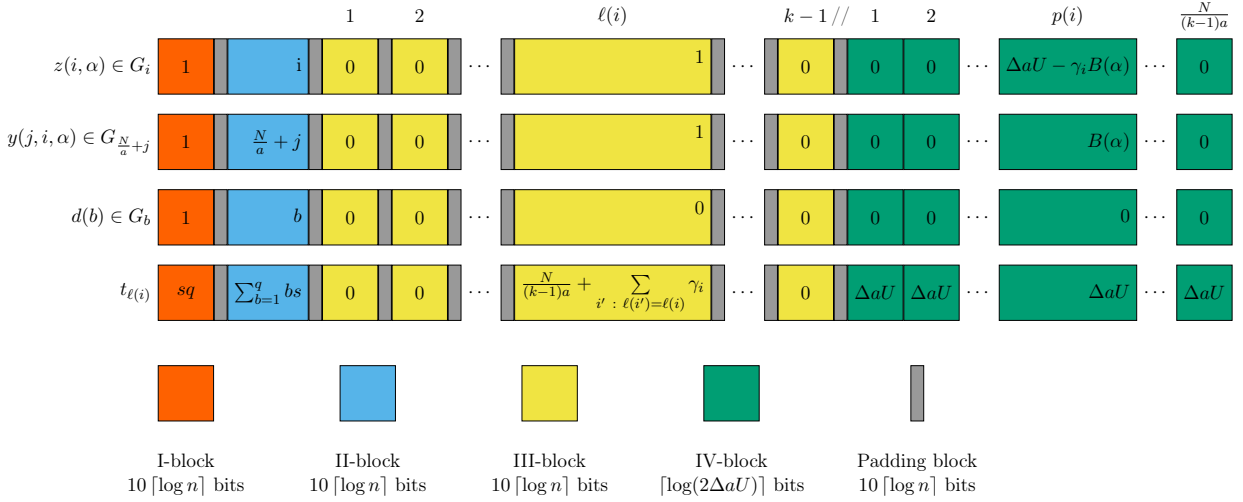

\begin{itemize}
	\item \textbf{I-block} To ensure that items in $G_i$ are all in the same range, in the first block of $10 \ceil{\log n}$ bits every item has value 1 and every target has value $sq$. 
	
	\item \textbf{II-block} In the next block of $10 \ceil{\log n}$ bits, each item in $G_b$ has value $b$ for any $b \in [q]$, and each target has value $\sum_{b \in [q]} b s$. This will ensure that at least one assignment item per group $G_b$ is packed. 
	 
	\item \textbf{III-blocks} The following $k-1$ blocks of $10 \ceil{\log n}$ bits ensure that assignment items involving the supervariable $X_i$ can only be packed in the $\ell(i)$-th bin. 
	Assignment items $z(i, \alpha)$ and $y(j, i, \alpha)$ take value 1 in the $\ell(i)$-th III-block and 0 in any other III-block.  All dummy items take value 0 in all III-blocks. 
	The target $t_\ell$ takes value $\frac{N}{(k-1)a} + \sum_{i : \ell(i) = \ell} \gamma_i$ in the $\ell$-th III-block and 0 in any other III-block, for all $\ell \in [k-1]$.
	
	\item \textbf{IV-blocks} The remaining $\frac{N}{(k-1)a}$ blocks of $\ceil{\log(2 \Delta a U)}$ bits guarantee that packed assignment items are consistent. 
	Every assignment item $z(i, \alpha)$ takes value $\Delta a U - \gamma_i B(\alpha)$ on the $p(i)$-th IV-block and value 0 on any other IV-block. 
	Every assignment item $y(j, i, \alpha)$ takes value $B(\alpha)$ on the $p(i)$-th IV-block and value 0 on any other IV-block. Dummy items take value 0 on all IV-blocks and targets take value $\Delta a U$ on all IV-blocks.
	
	\item \textbf{Padding blocks} Finally, we add $10 \ceil{\log n}$ \emph{padding bits} of value 0 after each I-, II- and III-block.
	
\end{itemize}

\paragraph*{Correctness}
We now prove the correctness of the reduction, which we split into two claims.
\begin{claim}
	Assume that $\phi$ is satisfiable. Then 
	there exists a tuple $(\gamma_1, \dots, \gamma_{N/a}) \in \{0, \dots, \Delta a\}^{N/a}$ with $\sum_{i \in [N/a]} \gamma_i = M$ such that
	the constructed instance of \pbtwoTargetsMultisets is in the \ref{enum:pb2targets_yes} case.

\end{claim}
\begin{claimproof}
We recall that \pbtwoTargetsMultisets is in the \ref{enum:pb2targets_yes} case if for every $b \in [q]$, there are disjoint subsets $S_{b, 1}, \dots, S_{b, k-1} \subset G_b$ satisfying the following \cref{enum:pb2targets_1a_SETH,enum:pb2targets_1b_SETH}:
\begin{yesenum}
	\item
\begin{yesenum}
	\item \label{enum:pb2targets_1a_SETH} $|S_{b, \ell}| = s$ for all $b \in [q]$ and $\ell \in [k-1]$,
	\item \label{enum:pb2targets_1b_SETH} $\sum_{b \in [q]} \Sigma(S_{b, \ell}) = t_\ell$ for all $\ell \in [k-1]$.
\end{yesenum}
\end{yesenum}

Suppose that $\phi$ admits a satisfying assignment $\alpha \in \{0, 1\}^N$. Let $\alpha_1, \dots, \alpha_{N/a} \in \{0, 1\}^a$ be the corresponding assignment of the supervariables $X_1, \dots, X_{N/a}$. For each clause we fix a supervariable whose assignment satisfies the clause, i.e.~we pick a function $f : [M] \mapsto [N/a]$ mapping every $j \in [M]$ to some $i \in [N/a]$ such that $X_i$ contains a variable whose assignment satisfies $C_j$. Let $\gamma_i$ be the cardinality of the inverse image of $f$ for $i \in [N/a]$, i.e.~the number of clauses that are satisfied by a variable in $X_i$ according to $f$. Since every variable appears in at most $\Delta$ clauses, we can bound $\gamma_1, \dots, \gamma_{N/a} \in \{0, \dots, \Delta a\}$, and since we fixed a satisfying variable for each clause we also have $\sum_{i \in [N/a]} \gamma_i = M$.
In what follows, we focus on the instance constructed for $\gamma_1, \dots, \gamma_{N/a}$. 

We construct the following disjoint subsets $S_{b, 1}, \dots, S_{b, k-1} \subset G_b$ for every $b \in [q]$. 
For every $i \in [N/a]$, let $S_{i, \ell(i)}$ be the multiset containing the assignment item $z(i, \alpha_i)$ and $s -1$ copies of the dummy item $d(i)$, and let $S_{i, \ell'}$ be the multiset containing $s$ copies of $d(i)$ for every $\ell' \in [k-1], \ell' \neq \ell(i)$. 
For every $j \in [M]$, let $S_{N/a + j, \ell(f(j))}$ be the multiset containing the assignment item $y(j, f(j), \alpha_{f(j)})$ and $s - 1$ copies of the dummy item $d(N/a + j)$, and let $S_{N/a + j, \ell'}$ be the multiset containing $s$ copies of $d(N/a +j)$ for every $\ell' \in [k-1], \ell' \neq \ell(f(j))$. 
Note that this is possible because for every $b \in [q]$ there are $(k-1)s -1$ copies of $d(b)$ in $G_b$.

Clearly, we have that $|S_{b, \ell}| = s$ for every $b \in [q]$ and $\ell \in [k-1]$, i.e.~\cref{enum:pb2targets_1a_SETH} holds for the subsets $S_{b, \ell}$. 
Fix $\ell \in [k-1]$ and let $L_\ell \coloneq \sum_{i = 1}^q \Sigma(S_{i, \ell})$ be the load of the $\ell$-th bin. We verify \cref{enum:pb2targets_1b_SETH}, i.e.~$L_\ell = t_\ell$ for every $\ell \in [k-1]$, by looking at the value of $L_\ell$ on every bit block.
For every $b \in [q]$, we have $|S_{b, \ell}| = s$, and every item has value 1 on the I-block, so $L_\ell$ has value $sq$ on the I-block, which coincides with $t_\ell$.
In the II-block, items in $G_b$ have value $b$, so $L_\ell$ has value $\sum_{b \in [q]} bs$ on the II-block, which coincides with $t_\ell$. 
Dummy items do not contribute to the blocks III and IV since they take value 0 on all these blocks. 
Hence, by construction of the subsets, the only items contributing to the value of $L_\ell$ in the III- and IV-blocks are assignment items in 
$S_{i, \ell}$ for $i \in [N/a]$ and in $S_{N/a + j, \ell}$ for $j \in [M]$ such that $f(j) = i$ and $\ell = \ell(i)$. Each such item contributes with value $1$ on the $\ell$-th III-block and value 0 on every other III-block.
There are $\frac{N}{(k-1)a}$ choices of $i \in [N/a]$ such that $\ell = \ell(i)$, and for each $i \in [N/a]$ there are $\gamma_i$ choices of $j \in [M]$ such that $f(j)=i$. So the value of $L_\ell$ in the $\ell$-th III-block is $\frac{N}{(k-1)a} + \sum_{i : \ell(i) = \ell} \gamma_i$, and 0 on every other III-block. 
Again, this coincides with the value of the III-blocks of $t_\ell$.
Finally, for every $p \in [\frac{N}{(k-1)a}]$, the only items contributing to the $p$-th IV-block of $L_\ell$ are assignment items in $S_{i, \ell}$ for $i \in [N/a]$ and in $S_{N/a + j, \ell}$ for $j \in [M]$ such that $p(i) = p$, $f(j) = i$ and $\ell = \ell(i)$. Note that, by the definition of $\ell(\cdot)$ and $p(\cdot)$, there is a unique $i \in [N/a]$ such that $\ell(i) = \ell$ and $p(i) = p$, and for that $i$ there are  $\gamma_i$ choices of $j \in [M]$ such that $f(j) = i$. On the $p$-th IV-block, the corresponding variable assignment item $z(i, \alpha_i) \in S_{i, \ell}$ takes value $\Delta a U - \gamma_i B(\alpha_i)$, and the corresponding $\gamma_i$ assignment items $y(j, i, \alpha_i)$ take value $B(\alpha_i)$. Hence, for every $p \in [\frac{N}{(k-1)a}]$, the $p$-th IV-block of $L_\ell$ has value $\Delta a U$, which coincides with $t_\ell$. 
In total, we indeed have that $L_\ell = \sum_{b \in [q]} \Sigma(S_{b, \ell}) = t_\ell$ for every $\ell \in [k-1]$, and thus $(G_1, \dots, G_q)$ is in the \ref{enum:pb2targets_yes} case of \pbtwoTargetsMultisets.
\end{claimproof}

It remains to prove the converse direction.
\begin{claim}
	Assume that $\phi$ is not satisfiable. Then, for every $(\gamma_1, \dots, \gamma_{N/a}) \in \{0, \dots, \Delta a\}^{N/a}$ with $\sum_{i \in [N/a]} \gamma_i = M$ the constructed instance of \pbtwoTargetsMultisets is in the \ref{enum:pb2targets_no} case.
\end{claim}
\begin{claimproof}
We prove the contrapositive. Suppose that for some $(\gamma_1,\ldots,\gamma_{N/a}) \in \{0, \dots, \Delta a\}^{N/a}$ with $\sum_{i \in [N/a]} \gamma_i = M$ the constructed instance $(G_1, \dots, G_q)$ is not in the \ref{enum:pb2targets_no} case of \pbtwoTargetsMultisets, i.e.~for every $b \in [q]$ there exist disjoint subsets $S_{b, 1}, \dots, S_{b, k-1} \subset G_b$ satisfying the following \cref{enum:pb2targets_2a_SETH,enum:pb2targets_2b_SETH,enum:pb2targets_2c_SETH}:
\begin{noenum}
	\item
	\begin{noenum}
		\item \label{enum:pb2targets_2a_SETH} $|S_{1, \ell}| + \dots + |S_{b, \ell}| \leq b \cdot s$ for all $b \in [q]$ and $\ell \in [k-1]$,
		\item \label{enum:pb2targets_2b_SETH} $\sum_{b \in [q], \ell \in [k-1]} |S_{b, \ell}| \geq (k-1) q s$,
		\item \label{enum:pb2targets_2c_SETH} $\sum_{b \in [q]} \Sigma(S_{b, \ell}) = t_\ell$ for all $\ell \in [k-1]$.
	\end{noenum}
\end{noenum}
We show that there exists an assignment $\alpha \in \{0, 1\}^N$ satisfying $\phi$. 
Let $L_\ell \coloneq \sum_{b \in [q]} \Sigma(S_{b, \ell})$ be the load of the $\ell$-th bin. 
By \cref{enum:pb2targets_2c_SETH}, we have $L_\ell = t_\ell$ for every $\ell \in [k-1]$. 
Since all padding blocks have value 0 and length greater than $\log n$, we deduce that there is no overflow from lower bits over any padding block. It follows that every I-block has the same value in $L_\ell$ and in $t_\ell$ for each $\ell \in [k-1]$. The same holds for every III-block, for every II-block, and for the union of all IV-blocks.
We analyse the blocks sequentially by starting from the higher order bits. We start by proving the following property based on the II-blocks.
\begin{pty}\label{cla:intermediate_Sbls}
	For every $b \in [q]$ and $\ell \in [k-1]$, we have $|S_{b, \ell}| = s$.
\end{pty}
\begin{claimproof}
For every $b \in [q]$ and $\ell \in [k-1]$ let $f(b, \ell) \coloneq s - |S_{b, \ell}|$. \cref{enum:pb2targets_2a_SETH} can be rewritten as:
\begin{align}
	&\sum_{b = 1}^{c} f(b, \ell) \geq 0 \quad \text{ for all } c \in [q] \text{ and }\ell \in [k-1].\label{eq:seth_sbl1}
	\intertext{Together, \Cref{enum:pb2targets_2a_SETH,enum:pb2targets_2b_SETH} imply that $\sum_{\ell \in [k-1]} \sum_{b \in [q]} |S_{b, \ell}| = (k-1)sq$, which is equivalent to}
	&\sum_{\ell = 1}^{k-1} \sum_{b = 1}^q f(b, \ell) = 0.\label{eq:seth_sbl2}
\intertext{Since the II-blocks of $t_\ell$ and $L_\ell$ have the same value, we have $\sum_{b = 1}^q b |S_{b, \ell}| = \sum_{b = 1}^q bs$ which is equivalent to $\sum_{b=1}^q b \cdot f(b, \ell) = 0$. This can be rewritten as:}
&\sum_{c=1}^q \sum_{b = c}^q f(b, \ell) = 0\quad \text{ for all } \ell \in [k-1].\label{eq:seth_sbl3}
\end{align}
We claim that~\eqref{eq:seth_sbl1},\eqref{eq:seth_sbl2} and \eqref{eq:seth_sbl3} imply $f(b, \ell) = 0$ for all $b \in [q]$ and $\ell \in [k-1]$.
Indeed, \eqref{eq:seth_sbl1} with $c=q$ gives $\sum_{b = 1}^{q} f(b, \ell) \geq 0$ for all $\ell \in [k-1]$, which combined with \eqref{eq:seth_sbl2} implies that $\sum_{b = 1}^{q} f(b, \ell) = 0$ for all $\ell \in [k-1]$ (as the sum of non-negative integers is $0$ if and only if all the integers are $0$). By using~\eqref{eq:seth_sbl1} again, we obtain
$\sum_{b = c}^{q} f(b, \ell) = \sum_{b = 1}^{q} f(b, \ell) - \sum_{b = 1}^{c-1} f(b, \ell)\leq 0$ for all $c \in [q]$ and $\ell \in [k-1]$. 
Next, we plug this into \eqref{eq:seth_sbl3} to conclude that $\sum_{b=c}^{q} f(b, \ell) = 0$ for all $c \in [q]$ and $\ell \in [k-1]$ (as the sum of non-positive numbers equals $0$ if and only if each of them is $0$). This implies that $f(c, \ell) = \sum_{b = c}^{q} f(b, \ell) - \sum_{b = c+1}^{q} f(b, \ell) = 0 - 0 = 0$ for all $c \in [q]$ and $\ell \in [k-1]$. Hence, $|S_{b,\ell}| = s$ for all $b \in [q]$ and $\ell \in [k-1]$.
\end{claimproof}

Since there are at most $(k-1)s -1$ dummy items per group, \Cref{cla:intermediate_Sbls} implies that \emph{at least one assignment item per group $G_b$ is packed}. 
By construction of the III-blocks, the total number of assignment packed items, i.e.~assignment items in $\bigcup_{b \in [q], \ell \in [k-1]} S_{i, \ell}$, is equal to the sum of all III-blocks of $\sum_{\ell = 1}^{k-1} L_\ell$. This is equal to the sum of all III-blocks of $\sum_{\ell = 1}^{k-1} t_\ell$, which is 
$$\sum_{\ell = 1}^{k-1} \Big(\frac{N}{(k-1)a} + \sum_{i : \ell(i) = \ell} \gamma_i\Big) = \frac{N}{a} + \sum_{i = 1}^{N/a} \gamma_i = \frac{N}{a} + M = q.$$ 
Hence, \emph{exactly one assignment item per group is packed}. 
For $i \in [N/a]$ and $j \in [M]$, denote by $z(i, \alpha_i) \in \bigcup_{\ell \in [k-1]} S_{i, \ell}$ and $y(j, i_j, \beta_j) \in \bigcup_{\ell \in [k-1]} S_{N/a + j, \ell}$ the unique packed assignment item of groups $G_i$ and $G_{N/a +j}$, respectively. By the values of the III-blocks, item $z(i, \alpha_i)$ can only be packed in $S_{i, \ell(i)}$ and item $y(j, i_j, \beta_j)$ can only be packed in $S_{N/a + j, \ell(i_j)}$. 

Finally, we analyse the IV-blocks of $L_\ell$ for every $\ell \in [k-1]$. By induction on $p \in [\frac{N}{(k-1)a}]$ in decreasing order, we show that there is no overflow among the IV-blocks of $L_\ell$. Along the way, we will show that $(\alpha_1,\ldots,\alpha_{N/a})$ is a satisfying assignment of $\phi$.
There is trivially no overflow into the $\frac{N}{(k-1)a}$-th IV-block. So consider $p \in [\frac{N}{(k-1)a}]$ and assume that there is no overflow into the $p$-th, $(p+1)$-th, $\dots$, $\frac{N}{(k-1)a}$-th IV-blocks of $L_\ell$. Then by construction, the only assignment items contributing to the value of the $p$-th IV-block of $L_\ell$ are $z(i, \alpha_i)$ and $y(j, i_j, \beta_j)$ such that $\ell = \ell(i) = \ell(i_j)$ and $p = p(i) = p(i_j)$. 
Note that, by the definition of $\ell(\cdot)$ and $p(\cdot)$, there is a unique $i \in [N/a]$ such that $\ell = \ell(i)$ and $p = p(i)$. Fix this value $i \in [N/a]$. 
Then assignment items contributing to the $p$-th IV block of $L_\ell$ are $z(i, \alpha_i)$ and all $y(j, i_j, \beta_j)$ with $i = i_j$. Let $Y_i \coloneq \{j \ : \ y(j, i_j, \beta_j) \in S_{N/a +j, \ell(i_j)} \text{ and } i_j = i\}$.
The $p$-th IV-blocks of items in $L_\ell$ sum up to
$\Delta a U - \gamma_i B(\alpha_i) + \sum_{j \in Y_i} B(\beta_j)$. Since every variable is contained in at most $\Delta$ clauses, there are at most $\Delta a$ clauses intersecting the supervariable $X_i$, i.e.~$|Y_i| \in \{0, 1, \dots, \Delta a\}$. Since $B \subseteq [U]$, we have $\Delta a U - \gamma_i B(\alpha_i) + \sum_{j \in Y_i} B(\beta_j) < 2 \Delta a U$. The IV-blocks have length $\lceil \log(2\Delta aU)\rceil$, so we deduce that there is no overflow from the $p$-th IV-block of $L_\ell$ into the $(p-1)$-th IV-block of $L_\ell$. In particular, the $p$-th IV-block of $L_\ell$ has value $\Delta a U - \gamma_i B(\alpha_i) + \sum_{j \in Y_i} B(\beta_j)$. As it is equal to the $p$-th IV-block of $t_\ell$, which has value $\Delta a U$, we deduce that 
$$\gamma_i B(\alpha_i) = \sum_{j \in Y_i} B(\beta_j).$$ 
Recall that $\gamma_i, |Y_i| \in \{0, 1, \dots, \Delta a\}$.
Since $B$ is strong $\Delta a$-average-free, the above equality implies that $\gamma_i = |Y_i|$ and $\beta_j = \alpha_i$ for every $j \in Y_i$. 
In other words, for every packed assignment item $y(j, i_j, \beta_j)$ we have $\alpha_{i_j} = \beta_j$ where $z(i_j, \alpha_{i_j})$ is also packed, i.e.~the packed assignment items are all consistent with the assignment $(\alpha_1,\ldots,\alpha_{N/a})$ of the supervariables $X_1,\ldots,X_{N/a}$. Since the item $y(j, i_j, \beta_j)$ is constructed only if $\beta_j$ satisfies the clause $C_j$, it follows that for every $j \in [M]$ there exists $i_j \in [N/a]$ such that $\alpha_{i_j}$ satisfies the clause $C_j$. In other words, $\phi$ is satisfied by the assignment $\alpha \in \{0, 1\}^N$ of variables $x_1, \dots, x_N$ obtained by aggregating the assignments $\alpha_1,\ldots,\alpha_{N/a}$ of supervariables $X_1,\ldots,X_{N/a}$.
\end{claimproof}

\paragraph{Parameters of the reduction}
It remains to prove that we have the desired size bounds. 
In total, the number of bits used to describe the constructed items is $(20(k-1) +40)\ceil{\log n} + \frac{N}{(k-1)a} \ceil{\log(2\Delta aU)}$. Let:
\begin{displaymath}
	W\coloneqq
	(1 + 2^{-10 \lceil \log n \rceil}) \cdot 2^{(20(k-1) +30)\ceil{\log n} + \frac{N}{(k-1)a} \ceil{\log(2\Delta aU)}}.
\end{displaymath}
Note that $W \in \N$. We claim that every constructed item is contained in $[W(1-1/n^{10}), W]$. Indeed, since the I-block of every item has value 1, all items have value at least 
$$2^{(20(k-1) + 30)\ceil{\log n} + \frac{N}{(k-1)a} \ceil{\log(2\Delta aU)}} \ge W (1 - 2^{-10 \lceil \log n \rceil}) \ge W (1-1/n^{10}).$$
Furthermore, since there are at least $10\ceil{\log n}$ padding bits between the I-block and the next 1-bits, all items have value at most $W$.

Recall that $s= \max \{\Delta a 2^a, K 2^a\}$,
$q = N/a + M$, and $M = \Delta N$, so $n = ksq \le c_1 \cdot N$, where $c_1 = c_1(K,k,\Delta,a)$ is some constant depending only on $K$, $k$, $\Delta$ and $a$. Recall that $U =  2^{a (1+\mu)} (\Delta a)^{\Oh(1/\mu)}$, so there exists an absolute constant $c_2 \geq 1$ such that $\ceil{\log(2\Delta aU)} \le (1+\mu) a + c_2 \log(\Delta a)/\mu$. Hence, assuming $N \ge 2$, we can bound:
\begin{align} \label{eq:bound_on_W}
	W \le  2^{(1+\mu) \frac{N}{k-1}} \cdot (\Delta a)^{\frac{c_2}{\mu a} \cdot \frac{N}{k-1}} \cdot N^{c_3},
\end{align}
for some constant $c_3 = c_3(K,k,\Delta,a)$ depending only on $K$, $k$, $\Delta$ and $a$. 

At the beginning of the proof we assumed that for any $\delta > 0$ \pbtwoTargetsMultisets can be solved in time $\Oh(2^{ \delta n} W^{k-1 - \eps}) = \Oh(2^{ \delta n} W^{(k-1)(1 - \bar{\eps})})$, where $\bar{\eps} = \eps/(k-1)$. Our reduction then implies that $K$-SAT can be solved in time
\begin{align*}
	2^{\lambda N} |\Phi|^{\Oh(1)} \cdot (\Delta a+1)^{N/a} \cdot 2^{\delta n} W^{(k-1)(1 - \bar{\eps})},
\end{align*}
where $2^{\lambda N} |\Phi|^{\Oh(1)}$ is the time required by the Sparsification
Lemma (\Cref{lem:sparsification}) and $(\Delta a+1)^{N/a}$ bounds the number of choices of $(\gamma_1, \dots, \gamma_{N/a})$.
We simplify this running time by plugging in $n \le c_1 N$ and (\ref{eq:bound_on_W}), loosely bounding $\Delta a + 1 \le (\Delta a)^2$, and noting that $N \le |\Phi|$ and thus $N^{(k-1) c_3} = |\Phi|^{\Oh(1)}$. This yields:
\begin{align*}
	2^{\lambda N} \cdot 2^{\delta c_1 N} \cdot 2^{(1-\bar{\eps}) (1+\mu) N} \cdot (\Delta a)^{(1-\bar{\eps}) \frac{c_2}{\mu a} \cdot N + \frac {2N} a} \cdot |\Phi|^{\Oh(1)}.
\end{align*}
Setting $c'_2 \coloneq c_2 + 2 \ge 1$, which is again an absolute constant, we can further bound the time by:
\begin{align*}
	2^{\lambda N} \cdot 2^{\delta c_1 N} \cdot 2^{(1-\bar{\eps}) (1+\mu) N} \cdot (\Delta a)^{\frac{c'_2}{\mu a} \cdot N} \cdot |\Phi|^{\Oh(1)}.
\end{align*}
To conclude the proof, we set $\eps_{\text{SAT}} \coloneq \bar{\eps} / 2$ and observe that assuming SETH there exists $K \ge 3$ such that $K$-SAT cannot be
solved in time $2^{(1 - \eps_{\text{SAT}}) N} |\Phi|^{\Oh(1)}$.
Recall that we set $\lambda = \mu = \bar{\eps}/10$. Based on $K$ and $\lambda$, we obtain $\Delta = \Delta(K, \lambda)$ from the Sparsification Lemma (\Cref{lem:sparsification}). 
Finally, we set $a \coloneq \max\{ \Delta, \lceil (200 c'_2 / \bar{\eps}^2)^2 \rceil \}$ and $\delta \coloneq \bar{\eps} / (10 c_1)$. 
We bound each factor of the time bound as follows: 
\begin{itemize}
	\setlength\itemsep{0em}
	\item $2^{\lambda N} \le 2^{\bar{\eps} N/10}$ because $\lambda = \bar{\eps}/10$. 
	\item $2^{\delta c_1 N} \le 2^{\bar{\eps} N/10}$ because $\delta = \bar{\eps}/(10c_1)$.
	\item $2^{(1-\bar{\eps})(1+\mu)N} = 2^{(1-\bar{\eps})(1+\bar{\eps}/10)N} \le 2^{(1 - 0.9\bar{\eps}) N}$ as $\mu = \bar{\eps}/10$.
	\item Using the fact that $x \le 2^{\sqrt{x}}$ for any $x \ge 20$, together with $a \ge \Delta$ we can bound $\Delta a \le a^2 \le 2^{2 \sqrt{a}}$. This yields $(\Delta a)^{c'_2 N / (\mu a)} \le 2^{2 c'_2 N / (\mu \sqrt{a})}$. By plugging in $\mu = \bar{\eps}/10$ and $a \ge (200 c'_2 / \bar{\eps}^2)^2$, we obtain an upper bound of $2^{\bar{\eps} N / 10}$.
\end{itemize}
In total, the running time is bounded by:
\begin{align*}
	2^{(1 - 0.6\bar{\eps}) N} \cdot |\Phi|^{\Oh(1)} \le 2^{(1 - \eps_{\text{SAT}}) N} \cdot |\Phi|^{\Oh(1)},
\end{align*}
which contradicts the assumption that $K$-SAT cannot be solved in time $2^{(1-\eps_{\text{SAT}}) N} |\Phi|^{\Oh(1)}$, and thus contradicts SETH.
\end{proof}

\section{Equivalences of Weak Grouped \boldmath$k$-way Partition}
\label{sec:equivalences_pbtwo}

Since in the previous section we proved SETH-hardness of \pbtwoTargetsMultisets, we now lift the targets and multisets assumptions to establish SETH-hardness of \pbtwo. We first prove that \pbtwoTargetsMultisets is equivalent to \pbtwoMultisets which in turn is proven to be equivalent to \pbtwo (see the orange-highlighted reductions in \cref{fig:graph_reduction}). We refer to \cref{app:problem-definitions} for the detailed problem definitions.

The following terminology is used. 
We consider having $k$ bins. 
Given (multi)sets of integers $G_1, \dots, G_q$ and collection of disjoint subsets $S_{i, 1}, \dots, S_{i, k-1} \subset G_i$, 
we interpret the items in $\bigcup_{i = 1}^q S_{i, \ell}$ as \emph{packed} into the $\ell$-th bin. Then we say that $\sum_{i=1}^q \Sigma(S_{i, \ell})$ is the \emph{load} of the $\ell$-th bin.  The \emph{average load} of the given integers is $\frac{1}{k} \sum_{i=1}^q \Sigma(G_i)$.


\begin{lemma}\label{lem:pbtwotargets_to_pbtwomulti}
	\pbtwoTargetsMultisets and \pbtwoMultisets are equivalent under parameter-preserving reductions. 
\end{lemma}
\begin{proof}
	\pbtwoMultisets is a special case of \pbtwoTargetsMultisets where the targets are all equal to the average load of the given integers, so it suffices to show a reduction from \pbtwoTargetsMultisets to \pbtwoMultisets.

	Let $G_1, \dots, G_q \subset \mathbb N \cap [W(1-1/n^{10}), W]$ and $t_1, \dots, t_{k-1} \in \mathbb N$ be an instance of  \pbtwoTargetsMultisets with group size $|G_i| = sk$ and total number of integers $n \coloneq ksq$. 
	Without loss of generality, we may assume that $n$ is large enough and $W \geq n^{10}$ as otherwise the instance can be solved in polynomial time.
	To construct an equivalent instance of \pbtwoMultisets we add new items that, when distributed among $k-1$ bins of loads $t_1, \dots, t_{k-1}$, increase the load of every bin to the average load of the constructed items. 
	More precisely, let $\mu \coloneq \frac{1}{k}\sum_{i = 1}^q \Sigma(G_i)$ be the average load of the given integers and let $t_{k} \coloneq k \mu - (t_1 +\dots + t_{k-1})$. 
		For each $i \in [q]$, let $G_i'$ be the multiset containing the integer
	\begin{align*}
		e' &\coloneq Wn^{20} + (e - W - 1)n^3 &\\
		\intertext{for every $e \in G_i$ with the same multiplicity. Let $G_{q+1}'$ be the multiset containing the \emph{filling items} $f_1, \dots, f_k$ and $k (s-1)$ copies of the \emph{dummy item} $d$ defined as}
		f_\ell &\coloneq Wn^{20} + (\mu - t_\ell - nW) n^3 + 1 & \text{for all } \ell \in [k] \\
		d &\coloneq W n^{20} - n^3 + n. &
	\end{align*}
	There are in total $m \coloneq n+sk \leq 2n$ constructed items.
	Define $W' \coloneq Wn^{20}$. We show that $G_i' \subset [W'(1-1/m^{10}), W']$ for every $i \in [q+1]$. For every $i \in [q]$, observe that, since every $e \in G_i$ is at most $W$, we have $e' \leq Wn^{20} -n^3 \leq W'$. On the other hand, since every $e \in G_i$ is at least $W(1-1/n^{10})$, we have 
	\begin{align*}
	e' 
	&\geq Wn^{20} + \left(W(1-1/n^{10}) - W - 1\right)n^3 \\
	&= Wn^{20} -W \frac{n^3}{n^{10}} - n^3 \\	
	&= Wn^{20}\left( 1 - \frac{n^3}{n^{20} n^{10}} - \frac{n^3}{Wn^{20}} \right)	 \\
	&\geq Wn^{20}\left(1 - 2\frac{n^3}{n^{30}} \right)	\tag{since $W \geq n^{10}$} \\
	&\geq Wn^{20}\left(1 - \frac{1}{m^{10}} \right)	\tag{since $n \geq 2$ and $2n \geq m$}.
	\end{align*}
	The dummy items can be bounded similarly: using the same assumptions that $n$ is large enough, $W \geq n^{10}$ and the observation that $2n \geq m$, we obtain the bounds $d \leq Wn^{20}$ and $d \geq Wn^{20} - n^3 = Wn^{20} \left(1 - \frac{n^3}{Wn^{20}}\right) \geq Wn^{20}(1 - 1/m^{10})$.
	It remains to bound the filling items. Note that we can assume that $t_1, t_2, \dots, t_{k}, \mu \in [sq W (1 - 1/n^{10}), sqW]$ as otherwise the instance is trivially in the \ref{enum:pb2targets_no} case. 
	In particular, for every $\ell \in [k]$ we have $|\mu - t_\ell| \leq sq W /n^{10} \leq sqW$. So for every $\ell \in [k]$
	\begin{align*}
	f_\ell \leq W n^{20} + (sqW - n W)n^3 + 1 \leq Wn^{20} - n^3 + 1 \leq d \leq Wn^{20}
	\end{align*}
	where the second inequality follows from $n= ksq \geq 2sq \geq sq+1/W$.
	On the other hand, since $n = ksq \geq 2$ and $2n \geq m$, we have
	$$
	f_\ell 
		\geq W n^{20} + (-sqW - n W)n^3 
		\geq W n^{20} -2W n^4 
		= W n^{20}\left( 1 - \frac{2n^4}{n^{10}}\right) 
		= W n^{20}\left( 1 - 1 / m^{10}\right).
	$$
	Hence, all the constructed integers are contained in the interval $[W'(1-1m^{10}), W']$. Therefore, $G_1', \dots, G_{q+1}' \subset \N \cap [W'(1-1/m^{10}), W']$ is an instance of \pbtwoMultisets with total number of items $m = n^{\Oh(1)}$. Since the construction takes time $\Oh(n)$ and $W' = W n^{\Oh(1)}$ it is parameter-preserving. 
	We observe that 
	\begin{align*}
		\Sigma(G_{q+1}') &= \sum_{\ell=1}^k f_\ell + k(s-1) \cdot d\\
		&= kWn^{20} + (k\mu - \sum_{i=1}^k t_\ell - knW)n^3 + k + k(s-1)Wn^{20} - k(s-1)n^3 + k(s-1)n\\
		&= ksWn^{20} + (-knW - k(s-1))n^3 + k(s-1)n +k \tag{by the definition of $t_k$}
	\end{align*}
	and thus the average load $\mu' 
		\coloneq \frac{1}{k}\sum_{i=1}^{q+1} \Sigma(G_i)$ of all constructed integers is
	\begin{align*}
		\mu' 
		&= \frac{n}{k}Wn^{20}  + (\mu - sq W - sq)n^3 + \frac{1}{k}\Sigma(G_{q+1}') \\
		&= \frac{m}{k}Wn^{20}  + \left(\mu - sq W - sq - nW - (s-1)\right)n^3 + (s-1)n + 1.
	\end{align*}

	We show equivalence between the two instances. Assume that $(G_1, \dots, G_q, t_1, \dots, t_{k-1})$ is in the \ref{enum:pb2targets_yes} case, i.e.~there exist disjoint subsets $S_{i, 1}, \dots, S_{i, k-1} \subset G_i$ for every $i \in [q]$, such that the following
	\cref{enum:pb2targets_1a_61,enum:pb2targets_1b_61} hold:
	\begin{yesenum}
		\item
	\begin{yesenum}
		\item \label{enum:pb2targets_1a_61} $|S_{i, \ell}| = s$ for all $i \in [q]$ and $\ell \in [k-1]$,
		\item \label{enum:pb2targets_1b_61} $\sum_{i \in [q]} \Sigma(S_{i, \ell}) = t_\ell$ for all $\ell \in [k-1]$.
	\end{yesenum}
	\end{yesenum}
	We want to show that in that case $(G_1', \dots, G_{q+1}')$ is in the \ref{enum:pb2_yes} case as well, i.e.~there exists disjoint subsets $S_{i, 1}', \dots, S_{i, k-1}' \subset G_i$ for every $i \in [q+1]$ satisfying the following \cref{enump:pb2_1a_61,enump:pb2_1b_61}:
	\begin{yesenump}
		\item
	\begin{yesenump}
		\item \label{enump:pb2_1a_61} $|S_{i, \ell}'| = s$ for all $i \in [q+1]$ and $\ell \in [k-1]$,
		\item \label{enump:pb2_1b_61} $\sum_{i \in [q+1]} \Sigma(S_{i, \ell}') = \mu'$ for all $\ell \in [k-1]$.
	\end{yesenump}
	\end{yesenump}
	For $i \in [q]$ and $\ell \in [k-1]$, let $S_{i, \ell}' \subset G_i'$ be the corresponding subset of items in the constructed instance $(G_1', \dots, G_{q+1}')$. For every $\ell \in [k-1]$, define $S_{q+1, \ell}'$ as the subset of $G_{q+1}'$ containing the filling item $f_\ell$ and $s-1$ copies of the dummy item $d$. Note that there are enough copies of $d$ such that $S_{q+1,1}', \dots, S_{q+1, k-1}'$ are disjoint. Then 
	\cref{enum:pb2targets_1a_61} implies \cref{enump:pb2_1a_61}. Furthermore, for every $\ell \in [k-1]$ we have
	\begin{align*}
		\sum_{i = 1}^{q+1} \Sigma(S_{i, \ell}') 
		&= sq Wn^{20} + \left( \sum_{i=1}^q \Sigma(S_{i, \ell}) - sqW - sq\right)n^3 + f_\ell + d \cdot (s-1) \\ 
		&= sq Wn^{20} + \left( t_\ell - sqW - sq\right)n^3 + f_\ell + d \cdot (s-1) \tag{by \cref{enum:pb2targets_1b_61}} \\ 
		&= s(q+1)Wn^{20} + \left(\mu - sqW - sq -nW - (s-1) \right)n^3 + (s-1)n + 1  = \mu'.
	\end{align*}
	So both \cref{enump:pb2_1a_61,enump:pb2_1b_61} are indeed satisfied.

	Conversely, we show by contraposition that if $(G_1, \dots, G_q, t_1, \dots, t_{k-1})$ is in the \ref{enum:pb2targets_no} case then $(G_1', \dots, G_{q+1}')$ is in the \ref{enum:pb2_no} case as well. To this end, suppose that 
	$(G_1', \dots, G_{q+1}')$ is not in the \ref{enum:pb2_no} case of \pbtwoMultisets, i.e.~there exist disjoint subsets $S_{i, 1}', \dots, S_{i, k-1}' \subset G_i'$ for every $i \in [q + 1]$ such that the following \cref{enump:pb2_2a_61,enump:pb2_2b_61,enump:pb2_2c_61} hold:
	\begin{noenump}
		\item 
		\begin{noenump}
			\item\label{enump:pb2_2a_61} $|S_{1, \ell}'| + \dots + |S_{i, \ell}'| \leq i \cdot s$ for all $i \in [q+1]$ and $\ell \in [k-1]$,
			\item\label{enump:pb2_2b_61} $\sum_{i\in [q+1], \ell \in [k-1]} |S_{i, \ell}'| \geq (k-1) (q + 1) s$,
			\item\label{enump:pb2_2c_61} $\sum_{i \in [q+1]} \Sigma(S_{i, \ell}') =  \mu'$ for all $\ell \in [k-1]$.
		\end{noenump}
	\end{noenump}
	We show that in that case $(G_1, \dots, G_q, t_1, \dots, t_{k-1})$ is not in the \ref{enum:pb2targets_no} case of \pbtwoTargetsMultisets, i.e.~we construct disjoint subsets $S_{i, 1}, \dots, S_{i, k-1} \subset G_i$ for every $i \in [q]$ and verify that the following \cref{enum:pb2targets_2a_61,enum:pb2targets_2b_61,enum:pb2targets_2c_61} hold. 
	\begin{noenum}
		\item 
		\begin{noenum}
			\item\label{enum:pb2targets_2a_61} $|S_{1, \ell}| + \dots + |S_{i, \ell}| \leq i \cdot s$ for all $i \in [q]$ and $\ell \in [k-1]$,
			\item\label{enum:pb2targets_2b_61} $\sum_{i\in [q], \ell \in [k-1]} |S_{i, \ell}| \geq (k-1) qs$,
			\item\label{enum:pb2targets_2c_61} $\sum_{i \in [q]} \Sigma(S_{i, \ell}) = t_\ell$ for all $\ell \in [k-1]$.
		\end{noenum}
	\end{noenum}
	First, observe that $\mu' \equiv 1 \Mod{n}$ and $\mu' \equiv  (s-1)n + 1 \Mod{n^3}$, and the only items in $G_1', \dots, G_{q+1}'$ that are non-zero modulo $n^3$ are the filling and dummy items in $G_{q+1}'$. Since filling items modulo $n$ have value $1$ and dummy items modulo $n^3$ have value $n$, we deduce from the above \cref{enump:pb2_2c_61} that $S_{q+1, \ell}'$ consists of exactly one filling item and $s-1$ dummy items, for each $\ell \in [k-1]$. 
	By reordering the bins, we can assume without loss of generality that $S_{q+1, \ell}'$ contains the filling item $f_\ell$.
	Now, for every $i \in [q]$, let $S_{i, 1},  \dots, S_{i, k-1} \subset G_i$ be the subset of items in the given instance corresponding to $S_{i, 1}',  \dots, S_{i, k-1}' \subset G_i'$.
	Then \cref{enump:pb2_2a_61} directly implies \cref{enum:pb2targets_2a_61}.
	Since we argued above that $|S_{q+1, \ell}'| = s$ for every $\ell \in [k-1]$, \cref{enump:pb2_2b_61} implies that 
	\begin{align*}
		\sum_{\ell=1}^{k-1} \sum_{i=1}^q |S_{i, \ell}| = \sum_{\ell=1}^{k-1} \sum_{i=1}^{q+1} |S_{i, \ell}'| - \sum_{\ell=1}^{k-1} |S_{q+1, \ell}'| \geq (k-1)(q+1)s - (k-1)s = (k-1)qs,
	\end{align*}
	i.e.~\cref{enum:pb2targets_2b_61} holds as well.
	Furthermore, we can compute 
	$\Sigma(S_{q+1, \ell}') = sWn^{20} + (\mu -t_\ell - nW - (s-1))n^3 + (s-1)n + 1$
	and thus 
	\begin{align*}
		\mu' - \Sigma(S_{q+1, \ell}') &= sqWn^{20} + (t_\ell - sqW - sq)n^3. \\
		\intertext{By \cref{enump:pb2_2c_61}, this is equal to }
		\sum_{i=1}^{q} \Sigma(S_{i, \ell}')  
		&= \sum_{i=1}^q |S_{i, \ell}'| W n^{20} + \left(\sum_{i=1}^q \Sigma(S_{i, \ell}) - \sum_{i=1}^q|S_{i, \ell}|(W+1)\right)n^3.
	\end{align*}
	Note that $|(t_\ell - sqW - sq)n^3| \leq 2sqWn^3 < Wn^{20}/2$, and, since $0 \leq \sum_{i=1}^q \Sigma(S_{i, \ell}) \leq nW$ and $0 \leq \sum_{i=1}^q |S_{i, \ell}| \leq n$, we can also bound $ |\left(\sum_{i=1}^q \Sigma(S_{i, \ell}) - \sum_{i=1}^q|S_{i, \ell}|(W+1)\right)n^3| \leq 2n^4W < W n^{20} / 2$. Thus, by taking the equation $\sum_{i=1}^{q} \Sigma(S_{i, \ell}')  = \mu' - \Sigma(S_{q+1, \ell}')$ modulo $Wn^{20}$, we obtain 
	\begin{align*}
		\sum_{i=1}^q \Sigma(S_{i, \ell}) - \sum_{i=1}^q|S_{i, \ell}|(W+1) = t_\ell - sqW - sq.
	\end{align*}
	Finally, notice that by combining \cref{enum:pb2targets_2b_61,enum:pb2targets_2a_61}, we can infer that $\sum_{i=1}^q \sum_{\ell = 1}^{k-1} |S_{i, \ell}| = sq$. Thus, the above equation implies that 
	$\sum_{i=1}^q \Sigma(S_{i, \ell}) = t_\ell$ for every $\ell \in [k-1]$, i.e.~\cref{enum:pb2targets_2c_61} holds as well. 
\end{proof}

\begin{lemma}\label{lem:pbtwomulti_to_pbtwo}
	\pbtwoMultisets and \pbtwo are equivalent under parameter-preserving  reductions. 
\end{lemma}
\begin{proof}
	Clearly, \pbtwo is a special case of \pbtwoMultisets, so we show the other direction.
	Consider multisets $G_1, \dots, G_q \subset \mathbb N \cap [W(1 - 1/n^{10}), W]$ for some integers $s, q, W \in \mathbb N$, 
	group size $|G_i| = sk$ and total number of items $n \coloneq ksq$. 
	Without loss of generality, we may assume that $n$ is large enough, and $W \geq n^{10}$ as otherwise the instance can be solved in polynomial time.
	%
	For every $i \in [q]$, fix an arbitrary order (with multiplicities) of the elements of $G_i = \{x_1^i, \dots, x_{sk}^i\}$, and define the sets $G_i' = \{y_1^i, \dots, y_{sk}^i\}$ and $G_{q+i}' = \{z_1^{i}, \dots, z_{sk}^i\}$ where for each $j \in [sk]$
	\begin{align*}
		y_j^i &\coloneq W n^{20} + (x_j - W)n^7 - n^5 + j \\
		z_j^i &\coloneq W n^{20} - j.
	\end{align*}
	Let $m \coloneq 2n$ be the total number of constructed elements and let $W'\coloneq Wn^{20}$.
	Note that $G_{q+i}' \subset [W'(1-1/m^{10}), W']$ for every $i \in [q]$. Additionally, for any $j \in [sk]$, $x_j^i \leq W$ so $y_j^i \leq Wn^{20} - n^5 + n \leq Wn^{20}$. On the other hand, since $x_j^i \geq W(1-1/n^{10})$ we have
	\begin{align*}
		y_j^i &\geq Wn^{20} + \left(W\left(1-1/n^{10}\right) - W\right)n^{7} - n^5 + 1 \\
		&\geq Wn^{20} - \frac{W n^7}{n^{10}} - n^5 \\
		&=Wn^{20}\left(1 - \frac{n^7}{n^{10}n^{20}} - \frac{n^5}{W n^{20}} \right) \\ 
		&\geq Wn^{20}\left(1 - \frac{n^7}{n^{30}} - \frac{n^5}{n^{30}} \right) \tag{since $W \geq n^{10}$} \\ 
		&\geq Wn^{20}\left(1 - \frac{1}{(2n)^{10}} \right). \tag{since $n \geq 2$}
	\end{align*}
	Hence, $G_i' \subset [W'(1-1/m^{10}), W']$ for every $i \in [2q]$ as well, and so $(G_1', \dots, G_{2q}')$ is an instance of \pbtwo with
	total number of items $m = \Oh(n)$. Since the construction takes time $\Oh(n)$ and $W' = W n^{\Oh(1)}$ it is parameter-preserving. Finally, note that the average load $\mu' \coloneq \frac{1}{k}\sum_{i=1}^{2q} \Sigma(G_i')$ of the constructed items is 
	\begin{align*}
		\mu' &= \frac{1}{k} \sum_{i=1}^q \sum_{j=1}^{sk} (y_j^i + z_j^i) = \frac{1}{k} \sum_{i=1}^q \sum_{j=1}^{sk} \left( 2 W n^{20} + (x_j^i - W)n^7 - n^5\right) \\
		&= 2sq W n^{20} + \mu n^7 - sq W n^7 - sqn^5,
	\end{align*}
	where $\mu \coloneq \frac{1}{k} \sum_{i=1}^q \Sigma(G_i)$ is the average load of the given items.

	Suppose that $(G_1, \dots, G_q)$ is in the \ref{enum:pb2_yes} case of \pbtwoMultisets, i.e.~there exist disjoint subsets $S_{i, 1}, \dots, S_{i, k-1} \subset G_i$ such that the following \cref{enum:pb2_1a_62,enum:pb2_1b_62} hold:
	\begin{yesenum}
		\item
		\begin{yesenum}
			\item\label{enum:pb2_1a_62} $|S_{i, \ell}| = s$ for all $i \in [q]$ and $\ell \in [k-1]$,
			\item\label{enum:pb2_1b_62} $\sum_{i \in [q]} \Sigma(S_{i, \ell}) = \mu$ for all $\ell \in [k-1]$.
		\end{yesenum}
	\end{yesenum}
	Let $S_{i, \ell}' \coloneq \{y_j^i \ : \ x_j^i \in S_{i, \ell} \} \subset G_i'$ and $S_{q+i, \ell}' \coloneq \{z_j^i \ : \ x_j^i \in S_{i, \ell}\} \subset G_{q_i}'$ for every $i \in [q]$ and $\ell \in [k-1]$. 
	Then, by \cref{enum:pb2_1a_62}, $|S_{i, \ell}'| = |S_{q+i, \ell}'| = s$ for every $i \in [q]$ and $\ell \in [k-1]$ and thus, by \cref{enum:pb2_1b_62}, we have for every $\ell \in [k-1]$
	\begin{align*}
		\sum_{i=1}^{2q} \Sigma(S_{i, \ell}') 
		&= \sum_{i=1}^q \sum_{x_j^i \in S_{i, \ell}} \left(Wn^{20} + (x_j^i - W)n^{7} - n^5 + j\right) + \sum_{i=q+1}^{2q} \sum_{x_j^i \in S_{i, \ell}} \left(Wn^{20} - j \right) \\
		&= 2|S_{i, \ell}| q W n^{20} + \sum_{i=1}^q \Sigma(S_{i, \ell})n^7 - |S_{i, \ell}|qW n^7 - |S_{i, \ell}|qn^5 \\
		&=2sq Wn^{20} + \mu n^7 - sqWn^7 - sqn^5 = \mu'.
	\end{align*}
	This shows that $(G_1', \dots, G_{2q}')$ is in the \ref{enum:pb2_yes} case of \pbtwo.

	Conversely, we show by contraposition that if $(G_1, \dots, G_q)$ is in the \ref{enum:pb2_no} case of \pbtwoMultisets then $(G_1', \dots, G_{2q}')$ is in the \ref{enum:pb2_no} case of \pbtwo. 
	To this end, suppose that $(G_1', \dots, G_{2q}')$ is not in the \ref{enum:pb2_no} case, i.e.~there exist disjoint subsets $S_{i, 1}', \dots, S_{i, k-1}' \subset G_i'$ for every $i \in [2q]$ such that \cref{enump:pb2_2a_62,enump:pb2_2b_62,enump:pb2_2c_62} hold. 
	\begin{noenump}
		\item 
		\begin{noenump}
			\item\label{enump:pb2_2a_62} $|S_{1, \ell}'| + \dots + |S_{i, \ell}'| \leq i \cdot s$ for all $i \in [2q]$ and $\ell \in [k-1]$,
			\item\label{enump:pb2_2b_62} $\sum_{i\in [2q], \ell \in [k-1]} |S_{i, \ell}'| \geq (k-1) 2q s$,
			\item\label{enump:pb2_2c_62} $\sum_{i \in [2q]} \Sigma(S_{i, \ell}') = \mu'$ for all $\ell \in [k-1]$.
		\end{noenump}
	\end{noenump}
	For every $i \in [q]$ and $\ell \in [k-1]$, let $S_{i, \ell} \coloneq \{x_j^i \ : \ y_j^i \in S_{i, \ell}' \} \subset G_i$. 
	We verify that the following \cref{enum:pb2_2a_62,enum:pb2_2b_62,enum:pb2_2c_62} hold. 
	This will show that $(G_1, \dots, G_q)$ is not in the \ref{enum:pb2_no} case, as desired.
	\begin{noenum}
		\item 
		\begin{noenum}
			\item\label{enum:pb2_2a_62} $|S_{1, \ell}| + \dots + |S_{i, \ell}| \leq i \cdot s$ for all $i \in [q]$ and $\ell \in [k-1]$,
			\item\label{enum:pb2_2b_62} $\sum_{i\in [q], \ell \in [k-1]} |S_{i, \ell}| \geq (k-1) q s$,
			\item\label{enum:pb2_2c_62} $\sum_{i \in [q]} \Sigma(S_{i, \ell}) = \mu$ for all $\ell \in [k-1]$.
		\end{noenum}
	\end{noenum}
	Observe that $|S_{i, \ell}| = |S_{i, \ell}'|$ for every $i \in [q]$ and $\ell \in [k-1]$, so in particular \cref{enump:pb2_2a_62} directly implies \cref{enum:pb2_2a_62}. 
	Furthermore, for every $\ell \in [k-1]$, we have 
	\begin{align*}
		\sum_{i=1}^{2q} \Sigma(S_{i, \ell}') 
		&= \sum_{i=1}^q \sum_{y_j^i \in S_{i, \ell}'} \left( Wn^{20} + (x_j^i - W)n^7 - n^5 +j   \right) + \sum_{i=1}^{q} \sum_{z_j^i \in S_{q+i, \ell}'} \left( Wn^{20} - j \right) \\
		&= \sum_{i=1}^{2q} |S_{i, \ell}'| W n^{20} + \sum_{i=1}^q \Sigma(S_{i, \ell}) n^7 - \sum_{i=1}^q |S_{i, \ell}| \left( W n^7 + n^5 \right) + \sum_{i=1}^{q} \left(\sum_{y_j \in S_{i, \ell}'} j - \sum_{z_j \in S_{q+i, \ell}'} j \right) \\
		&= 2sq W n^{20} + \sum_{i=1}^q \Sigma(S_{i, \ell}) n^7 - \sum_{i=1}^q |S_{i, \ell}| \left( W n^7 + n^5 \right) + \sum_{i=1}^{q} \left(\sum_{y_j \in S_{i, \ell}'} j - \sum_{z_j \in S_{q+i, \ell}'} j \right).
	\end{align*}
	The last equality uses the observation that by  combining \cref{enump:pb2_2a_62,enump:pb2_2b_62} we obtain $2sq = \sum_{i=1}^{2q} |S_{i, \ell}'|$. 
	By \cref{enump:pb2_2c_62}, the above is equal to $\mu'$, which yields the following equation.
	\begin{equation}
		\mu n^7 - sqWn^7 - sqn^5 = \sum_{i=1}^q \Sigma(S_{i, \ell}) n^7 - \sum_{i=1}^q |S_{i, \ell}| \left( W n^7 + n^5 \right) + \sum_{i=1}^{q} \left(\sum_{y_j^i \in S_{i, \ell}'} j - \sum_{z_j^i \in S_{q+i, \ell}'} j \right) \label{eq:pb2_multiset_No3}
	\end{equation}
	The last term can be bounded as $\left| \sum_{i=1}^q \left(\sum_{y_j^i \in S_{i, \ell}'}j - \sum_{z_j^i \in S_{q+i, \ell}'}j \right) \right| < n^3/2$, so by taking \cref{eq:pb2_multiset_No3} modulo $n^3$ we deduce that $0 = \sum_{i =1}^q \left(\sum_{y_j^i \in S_{i, \ell}'}j - \sum_{z_j^i \in S_{q+i, \ell}'}j \right)$. 
	Furthermore, since $0 \leq \sum_{i=1}^q |S_{i, \ell}| n^5 < n^7$, by taking \cref{eq:pb2_multiset_No3} modulo $n^7$ we get  
	$sq= \sum_{i=1}^q |S_{i, \ell}|$, which in turn implies by \cref{eq:pb2_multiset_No3}, that $\mu = \sum_{i=1}^q \Sigma(S_{i, \ell})$. So \cref{enum:pb2_2c_62,enum:pb2_2b_62} also hold.
\end{proof}

\begin{fact}\label{lem:pbtwo_to_pbtwoYES}
	\pbtwoYes is a special case of \pbtwo.
\end{fact}

The following corollaries are immediate consequences of \cref{lem:parameter-preserving-reduction,lem:SETH_to_GroupedkBP,lem:pbtwotargets_to_pbtwomulti,lem:pbtwomulti_to_pbtwo,lem:pbtwo_to_pbtwoYES}.

\begin{corollary}\label{cor:SETH_pbtwo}
Assuming SETH, for every $\eps > 0$ and $k \ge 2$, there exists a $\delta >0$ such that
\pbtwo cannot be solved in $\Oh(2^{\delta n}W^{k-1-\epsilon})$ time, where $n$ is the total number of items and $W$ is the bound on the items.
\end{corollary}

\begin{corollary}\label{cor:SETH_pbtwoYES}
Assuming SETH, for every $\eps > 0$ and $k \ge 2$, there exists a $\delta > 0$ such that
\pbtwoYes cannot be solved in $\Oh(2^{\delta n}W^{k-1-\epsilon})$ time, where $n$ is the total number of items and $W$ is the bound on the items.
\end{corollary}
\noindent
\cref{cor:SETH_pbtwoYES} establishes~\cref{thm:SETH_to_Grouped_kWay_Partition}, while
\cref{cor:SETH_pbtwo} establishes~\cref{thm:SETH_to_Weak_Grouped_kWay_Partition}.

\section{SETH-hardness of Scheduling Problems}
\label{sec:scheduling-problems-seth}


In this section, we present parameter-preserving reductions from \textsc{(Weak)} \pbtwoYes to the problems \PrjCmax, \PSUj, \PSwjCj and \PSpjUj. This will imply \cref{thm:SETH_to_PSpjUj,thm:SETH_to_PSwjCj,thm:SETH_to_PrjCmax,cor:SETH_to_PSUj}. 
See \cref{fig:graph_reduction} for the full graph of proven reductions, and \cref{app:problem-definitions} for the detailed problem definitions.

Observe that the mentioned scheduling problems are defined as optimisation problems where the task is to minimise the respective objective function. In this section, we consider their decision variants, in which the task is to decide if the objective value is at most $\Lambda$ for some given target objective $\Lambda \in \N$.


\subsection{Makespan with Release Dates: Proof of Theorem~\ref{thm:SETH_to_PrjCmax} and Corollary~\ref{cor:SETH_to_PSUj}.}

We reduce \pbtwoYes to a special case of \PSUj in \cref{lem:GroupedkPart_to_special_schedule}, which we prove in \cref{lem:GroupedkPart_to_PrjCmax} to be equivalent to \PrjCmax.

\begin{lemma}\label{lem:GroupedkPart_to_special_schedule}
	There is a parameter-preserving reduction from \pbtwoYes to the variant of \PSUj with $\sum_j U_j = 0$.
\end{lemma}
\begin{proof}
	Let $G_1, \dots, G_q \subset \mathbb N \cap [W(1 - 1/n^{10}), W]$ be an instance of \pbtwoYes for some integers $s, q, W \in \mathbb N$ where the groups have size $|G_i| = sk$ and the total number of items is $n \coloneq ksq$. 
	Without loss of generality, we may assume that $n$ is large enough and $W \geq n^{10}$ as otherwise the instance can be solved in polynomial time.
	Let $\mu \coloneq \frac{1}{k} \sum_{i = 1}^q \Sigma(G_i)$ and denote the elements in the $i$-th group by $G_i = \{p_{(i-1)sk + 1}, \dots, p_{isk}\}$, i.e.~$G_1 \cup \dots \cup G_q = \{p_1, \dots, p_n\}$. 
	Build an instance of \PSUj with target objective $\sum_j U_j = 0$ by constructing for every item $p_j \in G_i$ for $j \in [(i-1)sk +1, isk]$ a job with processing time $p_{j}$ and due date $d_{j} \coloneq \min \{\mu, isW\}$ for each $i \in [q]$. 
	The construction takes $\Oh(n)$ time and produces an instance of \PSUj of size $n$ where the sum of processing times is $T = \Theta(nW)$ so it is parameter-preserving.
	We prove that the two instances are equivalent. Note that the objective of \PSUj is met if and only if all jobs are non-tardy, i.e.~they complete before their due dates.

	\subparagraph{Yes of \pbtwoYes $\to$ $\sum_j U_j = 0$.}
	Suppose that for each $i \in [q]$ there exists a partition of $G_i$ into subsets $S_{i, 1}, \dots, S_{i, k}$ satisfying the following \cref{enum:pb2yes_1a_71,enum:pb2yes_1b_71} of \pbtwoYes :
	\begin{enumerate}
		\item\label{enum:pb2yes_1a_71} $|S_{i, \ell}| = s$ for all $i \in [q]$ and $\ell \in [k]$,
		\item\label{enum:pb2yes_1b_71} $\sum_{i \in [q]} \Sigma(S_{i, \ell}) = \mu$ for all $\ell \in [k]$.
	\end{enumerate}
	Then we schedule the jobs corresponding to items in $S_{1, \ell} \cup \dots \cup S_{q, \ell}$ on the $\ell$-th machine in that order, for each $\ell \in [k]$ (choose an arbitrary order within $S_{i, \ell}$). 
	Note that for $i < q$, the jobs from group $G_i$ have due date $i s W$, and the jobs from group $G_q$ have due date $\mu$. 
	By \cref{enum:pb2yes_1a_71}, since processing times are at most $W$, the jobs of group $G_i$ are completed on the $\ell$-th machine by time $\Sigma(S_{1, \ell}) + \dots + \Sigma(S_{i, \ell}) \leq i s W$. So for each $i < q$ the jobs are completed before their due date.  On the other hand, by \cref{enum:pb2yes_1b_71}, the completion time of the jobs of the group $G_q$ on the $\ell$-th machine is at most $\Sigma(S_{1, \ell}) + \dots + \Sigma(S_{q, \ell}) = \mu$. Hence, all the jobs are completed before their due date.
	
	\subparagraph{$\sum_j U_j = 0$ $\to$ Yes of \pbtwoYes.}
	Conversely, assume that there exists a schedule of all $n$ jobs on the $k$ machines such that every job completes before its due date. Let $S_{i, \ell}$ be the processing times of the jobs scheduled on the $\ell$-th machine corresponding to group $G_i$. 
	Then $S_{i, 1}, \dots, S_{i, k}$ is a partition of $G_i$. We show that the above \cref{enum:pb2yes_1a_71,enum:pb2yes_1b_71} are satisfied.
	The due dates are non-decreasing with $i$, so we can assume without loss of generality that jobs are processed in increasing order of $i$. 
	Then $\Sigma(S_{1, \ell}) + \dots + \Sigma(S_{i, \ell})$ is the maximum completion time of the jobs corresponding to group $G_i$ on the $\ell$-th machine. 
	By the range of the processing times, we can lower bound it by 
	\begin{align*}
		\Sigma(S_{1, \ell}) + \dots + \Sigma(S_{i, \ell}) &\geq W (1 - 1/n^{10}) \cdot (|S_{1, \ell}| + \dots + |S_{i, \ell}|).
	\intertext{On the other hand, since all jobs complete before their due date, which is at most $isW$, we deduce the following upper bound}
		\Sigma(S_{1, \ell}) + \dots + \Sigma(S_{i, \ell}) &\leq isW.
	\end{align*}
	Since $is \leq n$ and the size of $S_{i, \ell}$ is integral, we infer  
	$$|S_{1, \ell}| + \dots  + |S_{i, \ell}| \leq \left\lfloor \frac{i s}{(1- 1/n^{10})} \right\rfloor \leq \lfloor i s \cdot(1 + 2/n^{10}) \rfloor  = is$$ for all $i \in [q]$ and all $\ell \in [k]$, where the second inequality uses
	$\frac{1}{(1-x)} \leq 1 + 2x$ for any $x \in [0, 1/2]$. 
	We prove \cref{enum:pb2yes_1a_71} by induction on $i$. For $i = 1$, the above equation implies that $|S_{1, \ell}| \leq s$ for all $\ell \in [k]$. But since $|G_1| = |S_{1, 1}| + \dots + |S_{1, k}| = sk$, this means that $|S_{1, \ell}| = s$ for all $\ell \in [k]$.
	Consider $i > 1$ and suppose that $|S_{i', \ell}| = s$ for all $i' < i$. Then the above equation implies that $(i-1)s + |S_{i, \ell}| \leq is$, which implies that $|S_{i, \ell}| \leq s$ for all $\ell \in [k]$. Again, since $|G_i| = |S_{i, 1}| + \dots + |S_{i, k}| = sk$, we deduce that $|S_{i, \ell}| = s$ for all $\ell \in [k]$, i.e.~\cref{enum:pb2yes_1a_71} holds for all $i \in [q]$. Furthermore, since the due dates are at most $\mu$ and every job is completed before its due date, 
	we can bound the makespan of the $\ell$-th machine by $\sum_{i=1}^q \Sigma(S_{i, \ell}) \leq \mu$. As this holds for every $\ell \in [k]$ and $\sum_{i=1}^q \Sigma(G_i) = \sum_{\ell = 1}^k \sum_{i=1}^q \Sigma(S_{i, \ell}) = k \mu$, we have $\sum_{i=1}^q \Sigma(S_{i, \ell}) = \mu$ for every $\ell \in [k]$, i.e.~\cref{enum:pb2yes_1b_71} holds as well. 
\end{proof}


\begin{lemma}\label{lem:GroupedkPart_to_PrjCmax}
	\PrjCmax and the special case of \PSUj with target objective $\sum_j U_j = 0$ are equivalent with respect to parameter-preserving reductions. 
\end{lemma}
\begin{proof}
	\begin{figure}[!t]
    \centering
    \resizebox{\textwidth}{!}{%
        \begin{tikzpicture}[
        node distance=0pt,
        box/.style={rectangle,draw,minimum height=.8cm},
        padding/.style={rectangle,draw,minimum width=25pt, minimum height=.8cm, fill=gray!80}]	

    \begin{scope}
    \draw[line width=4pt, cb_orange] (0.75, 0) -- (9, 0) ;
    \draw[line width=4pt, OI_yellow] (2.25, -.5) -- (9, -.5) ;
    \draw[line width=4pt, OI_lightblue] (3.75, -1) -- (9, -1) ;

    \node[anchor=east,align=right] at (-.5, -.5) {Availability \\ of jobs};
    \end{scope}

    \begin{scope}[yshift=-2.5cm]
    \node[box,fill=cb_orange,from={.75,0 to 3,.8}] (A) {A};
    \node[box,fill=OI_lightblue,from={3.75,0 to 6,.8}] (C) {C};
    \node[box,fill=OI_yellow,from={6,0 to 8.25,.8}] (B) {B};    
    \node[anchor=east,align=right] at (-.5, .3) {Schedule};
    \end{scope}
    
    \begin{scope}[yshift=-3cm]
        \draw[line width=1pt, ->, >={Stealth}] (-.5,0)--(10,0);

        \draw[line width=1pt]  (0, 3.5)--(0, -.1) node[below] {$0$};
        \draw[dashed] (.75, 3.2) -- (.75,-.3) node[below,align=center,rotate=60,anchor=north east] {$t_A = 1$\\$r_A = 1$};
        \draw[dashed] (2.25, 2.7) -- (2.25, -.3) node[below,align=center,rotate=60,anchor=north east] {$r_B = 3$};
        \draw[dashed] (3.75, 2.2) -- (3.75,-.3) node[below,align=center,rotate=60,anchor=north east] {$t_C = 5$\\$r_C =5$};
        \draw[dashed] (6, 1.5) -- (6, -.3) node[below,rotate=60,anchor=north east] {$t_B = 8$};

        \draw[dashed] (9, 3.2) -- (9, -.3) node[below,align=center,rotate=60,anchor=north east] {$M= 12$};

        \node at (10, -.4) {time};
    \end{scope}

    \begin{scope}[xshift=12cm]

    \begin{scope}
    \draw[line width=4pt, cb_orange] (0, 0) -- (8.25, 0);
    \draw[line width=4pt, OI_yellow] (0, -.5) -- (6.75, -.5);
    \draw[line width=4pt, OI_lightblue] (0, -1) -- (5.25, -1);


    \end{scope}

    \begin{scope}[yshift=-2.5cm]
    \node[box,fill=cb_orange,from={6,0 to 8.25,.8}] (A) {A};
    \node[box,fill=OI_lightblue,from={3,0 to 5.25,.8}] (C) {C};
    \node[box,fill=OI_yellow,from={.75,0 to 3,.8}] (B) {B};    
    \end{scope}

    \begin{scope}[yshift=-3cm]
        \draw[line width=1pt, ->, >={Stealth}] (-.5,0)--(9.5,0);

        \draw[line width=1pt]  (0, 3.5)--(0, -.3) node[below] {$0$};
        \draw[dashed] (6, 1.5) -- (6,-.3) node[below,align=center,rotate=60,anchor=north east] {$t_A = 8$};
        \draw[dashed] (6.75, 2.7) -- (6.75, -.3) node[below,align=center,rotate=60,anchor=north east] {$d_B = 9$};
        \draw[dashed] (8.25, 3.2) -- (8.25, -.3) node[below,align=center,rotate=60,anchor=north east] {$d_A = 11$};
        \draw[dashed] (5.25, 2.2) -- (5.25, -.3) node[below,align=center,rotate=60,anchor=north east] {$d_C = 7$};
        \draw[dashed] (3, 1.5) -- (3, -.3) node[below,align=center,rotate=60,anchor=north east] {$t_C = 4$};
        \draw[dashed] (.75, 1.5) -- (.75,-.3) node[below,rotate=60,anchor=north east] {$t_B = 1$};

        \node at (9.5, -.4) {time};
    \end{scope}

    \end{scope}
    \end{tikzpicture}
    }%
\caption{An instance of \PrjCmax with a valid schedule $S$ (left) and the corresponding instance of \PSUj with target objective $\sum_j U_j=0$ with the valid schedule $S'$ (right) in the proof of \cref{lem:GroupedkPart_to_PrjCmax}.}
\label{fig:GroupedkPart_to_PrjCmax}
\end{figure}
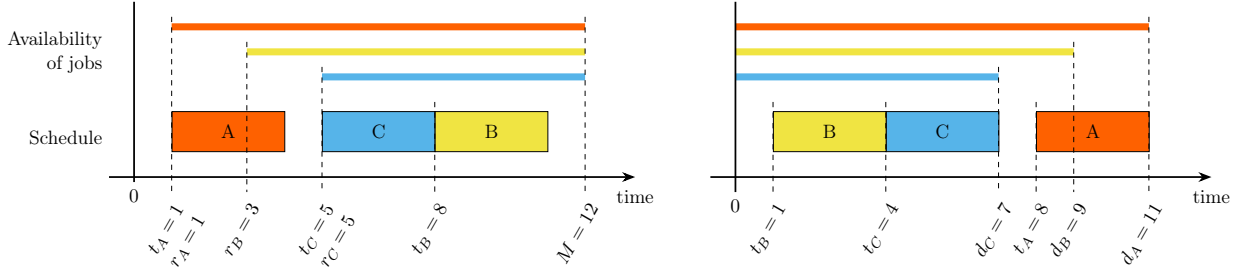
	The idea behind the equivalence is to ``reverse the course of time'' (see \cref{fig:GroupedkPart_to_PrjCmax} for an illustration).
	Indeed, for an instance of \PrjCmax with processing times $p_1, \dots, p_n \in \mathbb N$, release dates $r_1, \dots, r_n \in \mathbb N$ and target objective $M \in \mathbb N$, we create jobs with the same processing time $p_j' \coloneq p_j$ and due date $d_j \coloneq M - r_j$ for all $j \in [n]$. Since we can assume without loss of generality that $r_j \leq M$, the due dates are non-negative integers $d_j \in \mathbb N$. 
	For any schedule $S$ of the jobs $p_1, \dots, p_n$ with release dates $r_1, \dots, r_n$ on $k$ identical machines, construct the schedule $S'$ of the jobs $p_1', \dots, p_n'$ with due dates $d_1, \dots, d_n$ on $k$ identical machines as follows. If job $j$ is scheduled in $S$ from time $t_j$ to $t_j + p_j$, then schedule the job $j$ in $S'$ on the same machine from time $t_j'$ to $t_j' + p_j'$ for $t_j' \coloneq M - t_j - p_j'$. 
	We say that $S$ is \emph{valid} if every job is scheduled after its release date and the makespan of every machine is at most $M$. Similarly, $S'$ is \emph{valid} if every job terminates before its due date, and every job is scheduled after the start time $0$. Note that the job $j$ is scheduled in $S$ after its release date if and only if $t_j \geq r_j \Leftrightarrow t'_j + p_j' = M - t_j \leq M - r_j = d_j$, i.e.~the corresponding job $j$ terminates before its due date in the schedule $S'$. Furthermore, the makespan of every machine in $S$ is at most $M$ if and only if for every job $t_j + p_j \leq M \Leftrightarrow 0 \leq  M - t_j - p_j = t_j'$, i.e.~every job is scheduled in $S'$ after the start time $0$. Hence, there is a one-to-one mapping between valid schedules of the \PrjCmax instance and valid schedules of the constructed \PSUj instance with target objective $\sum_j U_j = 0$.
	

	Conversely, for an instance of \PSUj with processing times $p_1, \dots, p_n \in \mathbb N$, due dates $d_1, \dots, d_n \in \mathbb N$ and target objective $\sum_j U_j = 0$, let $M \coloneq \max_{j} d_j$ be a makespan objective and create jobs with the same processing times $p_j$ and release dates $r_j \coloneq M - d_j$. 
	Then a similar argument as above gives a one-to-one mapping between valid schedules of the \PrjCmax instance and valid schedules of the \PSUj instance with target objective $\sum_j U_j = 0$.
\end{proof}

Combining \cref{lem:parameter-preserving-reduction,cor:SETH_pbtwoYES,lem:GroupedkPart_to_special_schedule,lem:GroupedkPart_to_PrjCmax} proves the following, which in turn proves~\cref{thm:SETH_to_PrjCmax}.
\begin{corollary}
	Assuming SETH, for every $\eps > 0$ and $k \ge 2$, there exists $\delta > 0$ such that \PrjCmax cannot be solved in time $\Oh(2^{\delta n} T^{k-1-\epsilon})$, where $n$ is the number of jobs and $T$ is the total processing time.
\end{corollary}

Finally, observe that in \cref{lem:GroupedkPart_to_special_schedule} we studied the special case of \PSUj with target objective $\sum_j U_j = 0$. This problem trivially generalises to \PSUj (with any target objective). 
\begin{fact}\label{lem:GroupedkPart_to_PSUj}
	\PSUj with $\sum_j U_j = 0$ is a special case of \PSUj.
\end{fact}
\noindent Therefore, combining \cref{lem:parameter-preserving-reduction,lem:GroupedkPart_to_special_schedule,cor:SETH_pbtwoYES} with \cref{lem:GroupedkPart_to_PSUj} proves the following, which establishes~\cref{cor:SETH_to_PSUj}.
\begin{corollary}
	Assuming SETH, for every $\eps > 0$ and $k \ge 2$ there exists $\delta > 0$ such that \PSUj cannot be solved in time 
	$\Oh(2^{\delta n} T^{k-1-\epsilon})$, where $n$ is the number of jobs and $T$ is the total processing time.
\end{corollary}

\subsection{Weighted Sum of Completion Times: Proof of Theorem~\ref{thm:SETH_to_PSwjCj}}

\begin{lemma}\label{lem:GroupedkPart_to_PSwjCj}
	There is a parameter-preserving reduction from \pbtwoYes to \PSwjCj.
\end{lemma}
\begin{proof}
	Let $G_1, \dots, G_q \subset \mathbb N \cap [W(1 - 1/n^{10}), W]$ be an instance of \pbtwoYes for some integers $s, q, W \in \mathbb N$ where the groups have size $|G_i| = sk$ and the total number of items is $n \coloneq ksq$. 
	Without loss of generality, we may assume that $W \geq n^{10}$ as otherwise the instance can be solved in polynomial time.
	Let $\mu \coloneq \frac{1}{k} \sum_{i = 1}^q \Sigma(G_i)$ and denote the elements in the $i$-th group by $G_i = \{p_{(i-1)sk + 1}, \dots, p_{isk}\}$, i.e.~$G_1 \cup \dots \cup G_q = \{p_1, \dots, p_n\}$. Let $\Tilde W \coloneq 2 W n^{10}$. 
	Build an instance of \PSwjCj by constructing for every item $p_j \in G_i$ for $j \in [(i-1)sk +1, isk]$ a job with processing time $p_{j}$ and weight $w_j \coloneq p_j \Tilde W + (q-i)$, and set the target objective $\Lambda$ of weighted completion times to 
	$$
	\Lambda \coloneq \frac{\Tilde W}{2} k \mu^2 + \frac{\Tilde W}{2} \sum_{j = 1}^n p_j^2 + W k \sum_{i=1}^q (q-i) \sum_{r = 1}^s ((i-1)s + r).
	$$
	This constructs an instance of \PSwjCj in time $\Oh(n)$ of size $n$ where the sum of processing times is $\sum_{i =1}^q \Sigma(G_i) = \Theta(nW)$. Hence, it remains to show equivalence of the constructed and original instance.
	 
	We start with the following observations.
	Consider a schedule of all $n$ jobs on the $k$ machines and let $S_{i, \ell} \subset G_i$ be the set of processing times of the jobs of group $G_i$ scheduled on the $\ell$-th machine. Let $S_{\ell} \coloneq \bigcup_{i=1}^q S_{i, \ell}$.
	For any $j \in [n]$, denote by $C_j$ the completion time of job $j$ and let $B_j \subset [n]$ be the set of jobs that are scheduled on the same machine as $j$ and not processed after $j$ (so $j \in B_j$). Then $C_j = \sum_{j' \in B_j} p_{j'}$ and for each $\ell \in [k]$ we can express $\sum_{p_j \in S_\ell} p_j C_j$ as follows:
	\begin{align*}
		\sum_{p_j \in S_\ell} p_j C_j 
		= \sum_{p_j \in S_\ell} p_j \sum_{j' \in B_j} p_{j'} 
		= \frac{1}{2} \left( \left(\sum_{p_j \in S_\ell} p_j\right)^2 + \sum_{p_j \in S_\ell} p_{j}^2 \right)
		= \frac{1}{2} \Sigma(S_{\ell})^2 + \frac{1}{2}\sum_{p_j \in S_\ell} p_{j}^2.
	\end{align*}
	So the sum of weighted completion times is:
	\begin{align*}
		\sum_{j \in [n]} w_j C_j 
		&= \sum_{\ell =1}^{k} \sum_{p_j \in S_\ell} w_j C_j \\
		&= \sum_{\ell = 1}^{k} \left( \sum_{p_j \in S_\ell} p_j \Tilde W \cdot C_j + \sum_{i=1}^q \sum_{p_j \in S_{i, \ell}} (q-i) \cdot C_j \right) \\
		&= \frac{\Tilde W}{2} \sum_{\ell =1}^{k} \Sigma(S_{\ell})^2 + \frac{\Tilde W}{2} \sum_{\ell = 1}^k \sum_{p_j \in S_\ell} p_j^2 + \sum_{\ell =1}^{k} \sum_{i=1}^q \sum_{p_j \in S_{i, \ell}} (q-i)C_j \\
		&= \frac{\Tilde W}{2} \sum_{\ell =1}^{k} \Sigma(S_{\ell})^2 + \frac{\Tilde W}{2} \sum_{j =1}^n p_j^2 + \sum_{\ell =1}^{k} \sum_{i=1}^q \sum_{p_j \in S_{i, \ell}} (q-i)C_j
	\end{align*}
	and thus $\sum_{j \in [n]} w_j C_j \leq \Lambda$ if and only if 
	$$
	\frac{\Tilde W}{2} \sum_{\ell =1}^{k} \Sigma(S_{\ell})^2  + \sum_{\ell =1}^{k} \sum_{i=1}^q \sum_{p_j \in S_{i, \ell}} (q-i)C_j  \leq \frac{\Tilde W}{2} k \mu^2 + W k \sum_{i=1}^q (q-i) \sum_{r = 1}^s ((i-1)s + r).
	$$
	Note that on the left-hand side, the first term is at least $\Tilde W /2$ and the second term is less than $\Tilde W /2$. On the right-hand side, the first term is also at least $\Tilde W/2$ while the second term is less than $\Tilde W /2$.
	So in the above equation, there is no overflow between terms below and above $\Tilde W / 2$, and therefore $\sum_{j \in [n]} w_j C_j \leq \Lambda$ holds if and only if either \cref{eq:PSwjCj_makespan_str} holds  
	\begin{equation}\label{eq:PSwjCj_makespan_str}
		\sum_{\ell =1}^{k} \Sigma(S_{\ell})^2 < k \mu^2
	\end{equation}
	or \cref{eq:PSwjCj_makespan_tight,eq:PSwjCj_order_Cj} hold:
	\begin{equation}\label{eq:PSwjCj_makespan_tight}
		\sum_{\ell =1}^{k} \Sigma(S_{\ell})^2 =  k \mu^2
	\end{equation}
	\begin{equation}\label{eq:PSwjCj_order_Cj}
		\sum_{\ell =1}^{k} \sum_{i=1}^q \sum_{p_j \in S_{i, \ell}} (q-i)C_j \leq Wk\sum_{i=1}^q (q-i) \sum_{r= 1}^s ((i-1)s + r)
	\end{equation}
	However, note that by the inequality between quadratic and arithmetic means, we have 
	$$\sum_{\ell=1}^k \Sigma(S_{\ell})^2 \geq k \left(\frac{1}{k} \sum_{\ell=1}^k \Sigma(S_{\ell}) \right)^2 = k\mu^2,$$
	and equality holds if and only if $\Sigma(S_1) = \dots = \Sigma(S_k)$.
	Hence, \cref{eq:PSwjCj_makespan_str} never holds and we can replace \cref{eq:PSwjCj_makespan_tight} by 
	\begin{equation}\label{eq:PSwjCj_makespan}
		\Sigma(S_{\ell}) =  \mu \quad \forall \ell \in [k].
	\end{equation}
	On the other hand, since all processing times are at least $W (1 - 1/n^{10})$, we can bound the completion time of job $j$ by $C_j \geq |B_j| \cdot W (1 - 1/n^{10})$. Since all $|B_j|$ are integral and the right-hand side of \cref{eq:PSwjCj_order_Cj} is less than $W (n^{10} -1)$, \cref{eq:PSwjCj_order_Cj} is equivalent to \cref{eq:PSwjCj_order_Bj}:
	\begin{equation}\label{eq:PSwjCj_order_Bj}
		\sum_{\ell =1}^{k} \sum_{i=1}^q \sum_{p_j \in S_{i, \ell}} (q-i)|B_j| \leq k\sum_{i=1}^q (q-i) \sum_{r= 1}^s ((i-1)s + r).
	\end{equation}
	Therefore, $\sum_{j \in [n]} w_j C_j \leq \Lambda$ holds if and only if \cref{eq:PSwjCj_makespan,eq:PSwjCj_order_Bj} hold.
	We now show equivalence between the given instance of \pbtwoYes and the constructed instance of \PSwjCj.

	\subparagraph*{Yes of \pbtwoYes $\to$ $\sum_{j} w_j C_j \leq \Lambda$.}
	Suppose that for all $i \in [q]$ there exists a partition of $G_i$ into subsets $S_{i, 1}, \dots, S_{i, k}$ satisfying the following \cref{enum:pb2yes_1a_76,enum:pb2yes_1b_76} of~\pbtwoYes. 
	\begin{enumerate}
		\item\label{enum:pb2yes_1a_76} $|S_{i, \ell}| = s$ for all $i \in [q]$ and $\ell \in [k]$,
		\item\label{enum:pb2yes_1b_76} $\sum_{i \in [q]} \Sigma(S_{i, \ell}) = \mu$ for all $\ell \in [k]$.
	\end{enumerate}
	For each $\ell \in [k]$, schedule the jobs corresponding to $S_{1, \ell}, \dots, S_{q, \ell}$ on the $\ell$-th machine in that order (choose an arbitrary order within $S_{i, \ell}$). Then by \cref{enum:pb2yes_1b_76}, for each $\ell \in [k]$ we have $\sum_{i = 1}^q \Sigma(S_{i, \ell}) = \mu$. Therefore, \cref{eq:PSwjCj_makespan} holds. 
	Furthermore, by \cref{enum:pb2yes_1a_76}, the $r$-th job $j$ in $S_{i, \ell}$ has $|B_j| = (i-1)s + r$ jobs scheduled on the $\ell$-th machine not after $j$. So \cref{eq:PSwjCj_order_Bj} also holds. By the above discussion, this means that the schedule satisfies $\sum_{j\in[n]} w_j C_j \leq \Lambda$.

	\subparagraph*{$\sum_{j} w_j C_j \leq \Lambda$ $\to$ Yes of \pbtwoYes.}

	Conversely, suppose that there exists a schedule $\sigma$ with $\sum_{j\in[n]} w_j C_j \leq \Lambda$.
	Let $S_{i, \ell}$ be the processing times of jobs scheduled on the $\ell$-th machine corresponding to items in group $G_i$ and $S_\ell \coloneq \bigcup_{i = 1}^q S_{i, \ell}$. Then $S_{i, 1}, \dots, S_{i, k}$ is a partition of $G_i$. We show that the above \cref{enum:pb2yes_1a_76,enum:pb2yes_1b_76} are satisfied.
	Let $B_j \subset [n]$ be the set of jobs processed on the same machine as~$j$ not after~$j$, as introduced above. 
	Then, by the above discussion, both \cref{eq:PSwjCj_order_Bj,eq:PSwjCj_makespan} hold.
	Observe that we can assume without loss of generality that the jobs are scheduled with increasing order of $i$ on each machine. Indeed, this rearranging does not affect \cref{eq:PSwjCj_makespan}, and can only decrease the left-hand side of \cref{eq:PSwjCj_order_Bj}.
	Since \cref{enum:pb2yes_1b_76} directly follows from \cref{eq:PSwjCj_makespan}, we focus on showing \cref{enum:pb2yes_1a_76}.
	We will prove that the left-hand side of \cref{eq:PSwjCj_order_Bj} is minimised whenever \cref{enum:pb2yes_1a_76} holds. Since \cref{enum:pb2yes_1a_76} implies that \cref{eq:PSwjCj_order_Bj} is tight, and $\sigma$ satisfies \cref{eq:PSwjCj_order_Bj}, this will imply that $\sigma$ satisfies \cref{enum:pb2yes_1a_76} as well.


	Suppose that $\sigma$ does not satisfy \cref{enum:pb2yes_1a_76}, i.e.~there exists $i \in [q]$ and $\ell \in [k]$ with $|S_{i, \ell}| \neq s$. Let $g$ be the smallest such $i \in [q]$, i.e.~for each $i < g$ we have $|S_{i, \ell}| = s$ for all $\ell \in [k]$.
	Since $|S_{g, 1}| + \dots + |S_{g, k}| = |G_g| = sk$, this means that there exist $\alpha, \beta \in [k]$ such that $|S_{g, \beta}| < s < |S_{g, \alpha}|$. 
	Let $p_x \in S_{g, \alpha}$ be the last job from $S_{g, \alpha}$ scheduled on the $\alpha$-th machine and let $p_y \in S_{f, \beta}$ be the first job scheduled on the $\beta$-th machine right after the jobs of $S_{g, \beta}$ (see \cref{fig:PSwjCj_swap} for an illustration). In particular, $f > g$. 
	%
	Consider the schedule $\sigma'$ where we swap the jobs $p_x$ and $p_y$, i.e.~$p_y$ is now scheduled on the $\alpha$-th machine after the jobs of $S_{g, \alpha} \setminus \{p_x\}$ and $p_x$ is scheduled on the $\beta$-th machine after all the jobs of $S_{g, \beta}$ (see \cref{fig:PSwjCj_swap}). We use the prime notation to denote quantities corresponding to the schedule $\sigma'$. 
	Then $|B_x'| = |B_y| = 1 + \sum_{i = 1}^g |S_{i, \beta}|$, $|B_y'| = |B_x| = \sum_{i= 1}^g |S_{i, \alpha}|$ and $|B_j'| = |B_j|$ for any job $j \notin \{x, y\}$. 
	Therefore, the left-hand side of \cref{eq:PSwjCj_order_Bj} decreases by
	\begin{align*}
		&\phantom{=} 
		\sum_{\ell =1}^{k} \sum_{i=1}^q \sum_{p_j \in S_{i, \ell}} (q-i)|B_j| - \sum_{\ell =1}^{k} \sum_{i=1}^q \sum_{p_j \in S_{i, \ell}'} (q-i)|B_j'|\\
		&=
		(q - g)|B_x| + (q - f)|B_y| - (q-g)|B_x'| - (q-f) |B_y'| \\
		&= (|B_x| - |B_x'|)( (q-g) - (q - f)) \\
		&\geq |B_x| - |B_x'| \tag{because $f> g$} \\
		&\geq 1. \tag{because $|S_{g, \beta}| < s < |S_{g, \alpha}|$}
	\end{align*}
	In particular, the left-hand side of \cref{eq:PSwjCj_order_Bj} is strictly smaller for the schedule $\sigma'$ than for the schedule $\sigma$. This means that the left-hand side of \cref{eq:PSwjCj_order_Bj} is minimised when \cref{enum:pb2yes_1a_76} holds, as otherwise we can apply the above transformation. However, if \cref{enum:pb2yes_1a_76} holds, i.e.~$|S_{i, \ell}| = s$ for all $i \in [q]$ and $\ell \in [k]$, then the $r$-th job $j$ in $S_{i, \ell}$ has $|B_j| = (i-1)s + r$ jobs scheduled on the $\ell$-th machine not after $j$, and so 
	$$
	\sum_{\ell =1}^{k} \sum_{i=1}^q \sum_{p_j \in S_{i, \ell}'} (q-i)|B_j'| = k \sum_{i=1}^q (q-i) \sum_{r=1}^s ((i-1)s + r).
	$$
	This means that \cref{eq:PSwjCj_order_Bj} can only be satisfied if the left-hand side is minimised, i.e.~\cref{enum:pb2yes_1a_76} holds. Since $\sigma$ satisfies \cref{eq:PSwjCj_order_Bj}, we conclude that $\sigma$ satisfies \cref{enum:pb2yes_1a_76}, and thus that the given \pbtwoYes instance is in the \textsc{Yes} case.
	\begin{figure}[!t]
    \centering
    \resizebox{.8\textwidth}{!}{%
        \begin{tikzpicture}[
        node distance=0pt,
        box/.style={rectangle,draw,minimum height=.8cm},
        padding/.style={rectangle,draw,minimum width=25pt, minimum height=.8cm, fill=gray!80}]

    \begin{scope}
    \node[box,fill=cb_orange,minimum width = 70pt,anchor=west] (j1) at (0, 0) { };
    \node[box,fill=cb_orange,minimum width = 20pt] (j2) [right=of j1] { };
    \node[box,fill=cb_orange,minimum width = 20pt] (C) [right=of j2] { };    

    \node [left=.5cm of j1] (T) {Machine $\alpha$};
    \node [above=.5cm of T,xshift=-1cm] {\textbf{Schedule $\sigma$}};

    \node[box,fill=OI_lightblue,minimum width = 20pt,anchor=west] (j1) [right=of C] { };
    \node[box,fill=OI_lightblue,minimum width = 45pt] (j2) [right=of j1] { };
    \node[box,fill=OI_lightblue,minimum width = 25pt] (I) [right=of j2] { };  
    \node[box,fill=OI_lightblue,minimum width = 55pt] (J) [right=of I] {$x$};

    \node[box,fill=OI_yellow,minimum width = 20pt,anchor=west] (j1) [right=of J] { };
    \node[box,fill=OI_yellow,minimum width = 50pt] (j2) [right=of j1] { };

    \end{scope}

    \begin{scope}[yshift=-1cm]
    \node[box,fill=cb_orange,minimum width = 10pt,anchor=west] (A) at (0, 0) { };
    \node[box,fill=cb_orange,minimum width = 40pt] (j2) [right=of A] { };
    \node[box,fill=cb_orange,minimum width = 50pt] (C) [right=of j2] { };    

    \node [left=.5cm of A] {Machine $\beta$};

    \node[box,fill=OI_lightblue,minimum width = 60pt,anchor=west] (j1) [right=of C] { };
    \node[box,fill=OI_lightblue,minimum width = 30pt] (j2) [right=of j1] { };

    \node[box,fill=OI_green,minimum width = 20pt,anchor=west] (j1) [right=of j2] {$y$};
    \node[box,fill=OI_green,minimum width = 45pt] (j2) [right=of j1] { };
    \node[box,fill=OI_green,minimum width = 55pt] (I) [right=of j2] { };  
    \node[box,fill=OI_green,minimum width = 15pt] (Z) [right=of I] { };  

    \end{scope}

    \begin{scope}[yshift=-2cm]
        \draw[line width=1pt]  ([yshift=(1.5cm)]A.north west)--(0, -.1) node[below] {$0$};
    \end{scope}

    
    \begin{scope}[yshift=-4.5cm]

    \begin{scope}
    \node[box,fill=cb_orange,minimum width = 70pt,anchor=west] (j1) at (0, 0) { };
    \node[box,fill=cb_orange,minimum width = 20pt] (j2) [right=of j1] { };
    \node[box,fill=cb_orange,minimum width = 20pt] (C) [right=of j2] { };    

    \node [left=.5cm of j1] (T) {Machine $\alpha$};
    \node [above=.5cm of T,xshift=-1cm] {\textbf{Schedule $\sigma'$}};

    \node[box,fill=OI_lightblue,minimum width = 20pt,anchor=west] (j1) [right=of C] { };
    \node[box,fill=OI_lightblue,minimum width = 45pt] (j2) [right=of j1] { };
    \node[box,fill=OI_lightblue,minimum width = 25pt] (I) [right=of j2] { };  
    \node[box,fill=OI_green,minimum width = 20pt] (J) [right=of I] {$y$};

    \node[box,fill=OI_yellow,minimum width = 20pt,anchor=west] (j1) [right=of J] { };
    \node[box,fill=OI_yellow,minimum width = 50pt] (j2) [right=of j1] { };

    \end{scope}

    \begin{scope}[yshift=-1cm]
    \node[box,fill=cb_orange,minimum width = 10pt,anchor=west] (A) at (0, 0) { };
    \node[box,fill=cb_orange,minimum width = 40pt] (j2) [right=of A] { };
    \node[box,fill=cb_orange,minimum width = 50pt] (C) [right=of j2] { };    

    \node [left=.5cm of A] {Machine $\beta$};

    \node[box,fill=OI_lightblue,minimum width = 60pt,anchor=west] (j1) [right=of C] { };
    \node[box,fill=OI_lightblue,minimum width = 30pt] (j2) [right=of j1] { };

    \node[box,fill=OI_lightblue,minimum width = 55pt,anchor=west] (j1) [right=of j2] {$x$};
    \node[box,fill=OI_green,minimum width = 45pt] (j2) [right=of j1] { };
    \node[box,fill=OI_green,minimum width = 55pt] (I) [right=of j2] { };  
    \node[box,fill=OI_green,minimum width = 15pt] (Z) [right=of I] { };  

    \end{scope}

    \begin{scope}[yshift=-2cm]
        \draw[line width=1pt, ->, >={Stealth}] (-.5,0)--([yshift=(-1cm),xshift=(1cm)]Z.east);

        \draw[line width=1pt]  ([yshift=(1.5cm)]A.north west)--(0, -.1) node[below] {$0$};
        \node at ([yshift=(-1.4cm),xshift=(1cm)]Z.east) {time};
    \end{scope}

    \end{scope}

    \begin{scope}[yshift=-2cm]
        \draw[line width=1pt, ->, >={Stealth}] (-.5,0)--([yshift=(+3.5cm),xshift=(1cm)]Z.east);

        \node at ([yshift=(+3.1cm),xshift=(1cm)]Z.east) {time};
    \end{scope}

    \end{tikzpicture}
    }%
\caption{In the proof of \cref{lem:GroupedkPart_to_PSwjCj} we transform the schedule $\sigma$ into the schedule $\sigma'$ by swapping the jobs $x$ and $y$. Jobs belonging to different groups are represented with different colours.}
\label{fig:PSwjCj_swap}
\end{figure}
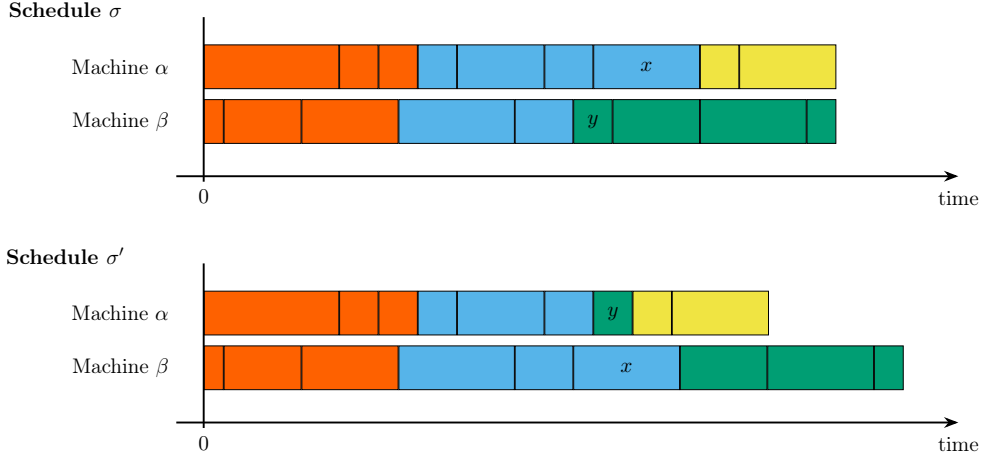
\end{proof}

Combining \cref{lem:parameter-preserving-reduction,cor:SETH_pbtwoYES,lem:GroupedkPart_to_PSwjCj}
proves the following.
\begin{corollary}
	Assuming SETH, for every $\eps > 0$ and $k \ge 2$, there exists $\delta > 0$ such that \PSwjCj cannot be solved in time $\Oh(2^{\delta n} T^{k-1-\epsilon})$, where $n$ is the number of jobs and $T$ is the total processing time.
\end{corollary}
\noindent
This establishes~\cref{thm:SETH_to_PSwjCj}.

\subsection{Total Processing Time of Tardy Jobs: Proof of Theorem~\ref{thm:SETH_to_PSpjUj}}

Finally, we reduce \pbtwo to \PSpjUj[k-1]. 
We emphasise that the reduction is from $k$ to $k-1$, so that we get the desired tight lower bound for \PSpjUj.

\begin{lemma}\label{lem:WeakGroupedkPart_to_PSpjUj}
	There is a parameter-preserving reduction from \pbtwo to \PSpjUj[k-1].
\end{lemma}

\begin{proof}
	Let $G_1, \dots, G_q \subset \mathbb N \cap [W(1 - 1/n^{10}), W]$ be an instance of \pbtwo for some integers $s, q, W \in \mathbb N$ where the groups have size $|G_i| = sk$ and the total number of items is $n \coloneq ksq$. 
	Without loss of generality, we may assume that $W \geq n^{10}$ as otherwise the instance can be solved in polynomial time.
	Let $\mu \coloneq \frac{1}{k} \sum_{i = 1}^q \Sigma(G_i)$ and denote the elements in the $i$-th group by $G_i = \{p_{(i-1)sk + 1}, \dots, p_{isk}\}$, i.e.~$G_1 \cup \dots \cup G_q = \{p_1, \dots, p_n\}$. 
	Build an instance of \PSpjUj[k-1] with target objective $\mu$ by constructing for every $i \in [q]$ and for every item $p_j \in G_i$ with $j \in [(i-1)sk +1, isk]$ a job with processing time $p_{j}$ and due date $d_{j} \coloneq \min \{\mu, isW\}$.
	This constructs an instance of \PSpjUj[k-1] in time $\Oh(n)$ of size $n$ where the sum of processing times is $\sum_{i =1}^q \Sigma(G_i) = \Theta(nW)$. Hence, it remains to show equivalence of the constructed and original instance.
	To simplify the presentation, we view a schedule of all jobs as a partial schedule where only non-tardy jobs are actually scheduled, while tardy jobs are viewed as unscheduled. 

	\subparagraph*{\ref{enum:pb2_yes} of \pbtwo $\to$ Yes of \PSpjUj[k-1].}

	Suppose that the instance $(G_1, \dots, G_q)$ is in the \ref{enum:pb2_yes} case of \pbtwo, i.e.~the following \cref{enum:pb2_1a_78,enum:pb2_1b_78} are satisfied by disjoint subsets $S_{i, 1}, \dots, S_{i, k-1} \subset G_i$ for every $i \in [q]$:
	\begin{yesenum}
		\item 
		\begin{yesenum}
			\item\label{enum:pb2_1a_78} $|S_{i, \ell}| = s$ for all $i \in [q]$ and $\ell \in [k-1]$,
			\item\label{enum:pb2_1b_78} $\sum_{i \in [q]} \Sigma(S_{i, \ell}) = \mu$ for all $\ell \in [k-1]$.
		\end{yesenum}
	\end{yesenum}
	Then, for each $\ell \in [k-1]$, we schedule the jobs corresponding to items in $S_{1, \ell} \cup \dots \cup S_{q, \ell}$ on the $\ell$-th machine in that order (choose an arbitrary order within $S_{i, \ell}$).
	We argue that those jobs are non-tardy.
	Note that for $i < q$, the jobs from the group $G_i$ have due date $i s W$, and for $i = q$, the jobs from the group $G_q$ have due date $\mu$. 
	By \cref{enum:pb2_1a_78}, since processing times are at most $W$, the jobs of group $G_i$ are completed on the $\ell$-th machine at time $\Sigma(S_{1, \ell}) + \dots + \Sigma(S_{i, \ell}) \leq i s W$, so for $i < q$ all jobs are completed before their due date. 
	On the other hand, by \cref{enum:pb2_1b_78}, the completion time of the jobs of the group $G_q$ on the $\ell$-th machine is $\Sigma(S_{1, \ell}) + \dots + \Sigma(S_{q, \ell}) = \mu$. Hence, all the scheduled jobs are completed before their due date, i.e.~they are non-tardy. Furthermore, the sum of the processing times of all remaining jobs is
	\begin{align*}
		\sum_{j \in [n]} p_j U_j = \sum_{i = 1}^q \Sigma(G_i) - \sum_{\ell \in [k-1] }\sum_{i \in [q]} \Sigma(S_{i, \ell}) 
		&= \sum_{i = 1}^q \Sigma(G_i) - (k-1) \cdot \mu  = \mu
	\end{align*}
	Since the remaining jobs include all tardy jobs, this shows that there is a schedule that meets the objective $\sum_j p_j U_j \leq \mu$. 

	\subparagraph*{\ref{enum:pb2_no} of \pbtwo $\to$ No of \PSpjUj[k-1].}
	This case is analogous to the respective case in the proof of~\cref{lem:GroupedkPart_to_special_schedule}. We include it here for completeness.
	
	By contraposition, assume that there exists a schedule with $\sum_j p_j U_j \leq \mu$. 
	Let $S_{i, \ell}$ be the set of processing times of the non-tardy jobs scheduled on the $\ell$-th machine corresponding to items of group $G_i$. Then $S_{i, 1}, \dots, S_{i, k-1}$ are disjoint subsets of $G_i$, for every $i \in [q]$. 
	We now show that the following \cref{enum:pb2_2a_78,enum:pb2_2b_78,enum:pb2_2c_78} are satisfied. This will imply that $(G_1, \dots, G_q)$ is not in the \ref{enum:pb2_no} case of \pbtwo, as desired.
	\begin{noenum}
		\item 
		\begin{noenum}
			\item\label{enum:pb2_2a_78} $|S_{1, \ell}| + \dots + |S_{i, \ell}| \leq i \cdot s$ for all $i \in [q]$ and $\ell \in [k-1]$,
			\item\label{enum:pb2_2b_78} $\sum_{i\in [q], \ell \in [k-1]} |S_{i, \ell}| \geq (k-1) q s$,
			\item\label{enum:pb2_2c_78} $\sum_{i \in [q]} \Sigma(S_{i, \ell}) = \mu$ for all $\ell \in [k-1]$.
		\end{noenum}
	\end{noenum}
	We can assume without loss of generality that all non-tardy jobs are scheduled before tardy jobs, and, since due dates are non-decreasing with $i$, that non-tardy jobs are processed in increasing order of $i$. 
	Then $\Sigma(S_{1, \ell}) + \dots + \Sigma(S_{i, \ell})$ is the maximum completion time of the non-tardy jobs corresponding to group $G_i$ on the $\ell$-th machine. 
	By the range of the processing times, we can lower bound it by 
	\begin{align*}
		\Sigma(S_{1, \ell}) + \dots + \Sigma(S_{i, \ell}) &\geq W (1 - 1/n^{10}) \cdot (|S_{1, \ell}| + \dots + |S_{i, \ell}|).
	\intertext{On the other hand, since all non-tardy jobs complete before their due date, which is at most $isW$, we deduce the following upper bound}
		\Sigma(S_{1, \ell}) + \dots + \Sigma(S_{i, \ell}) &\leq isW.
	\end{align*}
	Since $is \leq n$ and the size of $S_{i, \ell}$ are integral, we infer 
	$$|S_{1, \ell}| + \dots  + |S_{i, \ell}| \leq \left\lfloor \frac{i s}{(1- 1/n^{10})} \right\rfloor \leq \lfloor i s \cdot(1 + 2/n^{10}) \rfloor  = is$$ for all $i \in [q]$ and all $\ell \in [k-1]$ where the second inequality uses
	$\frac{1}{(1-x)} \leq 1 + 2x$ for any $x \in [0, 1/2]$. 
	In other words, the subsets $S_{i, \ell}$ satisfy \cref{enum:pb2_2a_78}. 
	We can also bound the sum of processing times of tardy jobs $\sum_{j} p_j U_j$ by
	$$
		W \left(1 - \frac{1}{n^{10}}\right) \cdot \left(n - \sum_{i \in [q], \ell \in [k-1]} |S_{i, \ell}|\right) \leq \sum_{j \in [n]} p_j U_j \leq \mu \leq \frac{1}{k} n W = qs W
	$$
	which implies that
	$$	
		\sum_{i \in [q], \ell \in [k-1]} |S_{i, \ell}| 
		\geq \sum_{i \in [q], \ell \in [k-1]} |S_{i, \ell}| \cdot \left(1 - \frac{1}{n^{10}}\right) \geq \left(1 - \frac{1}{n^{10}}\right) n - qs = (k-1)qs - \frac{1}{n^{9}}.
	$$
	Since $\sum_{i \in [q], \ell \in [k-1]} |S_{i, \ell}|$ is integral, we conclude that $\sum_{i \in [q], \ell \in [k-1]} |S_{i, \ell}| \geq (k-1) qs$, i.e.~\cref{enum:pb2_2b_78} holds as well. 
	Finally, since due dates are at most $\mu$, the completion times of non-tardy jobs are at most $\mu$, i.e.~$\sum_{i\in [q]}\Sigma(S_{i, \ell}) \leq \mu$ for each $\ell \in [k-1]$. Hence, the assumption that $\sum_{j \in [n]} p_j U_j \leq \mu$ implies that $\sum_{j \in [n]} p_j U_j = \mu$ and $\sum_{i\in [q]}\Sigma(S_{i, \ell}) = \mu$ for each $\ell \in [k-1]$. This proves that all \cref{enum:pb2_2a_78,enum:pb2_2b_78,enum:pb2_2c_78} are indeed satisfied.
\end{proof}
Note that in \cref{lem:WeakGroupedkPart_to_PSpjUj} we show a
parameter-preserving reduction from \pbtwo to \PSpjUj[k-1].  By
combining
\cref{cor:SETH_pbtwo,lem:WeakGroupedkPart_to_PSpjUj,lem:parameter-preserving-reduction}
we obtain the following.
\begin{corollary}
	Assuming SETH, for every $\eps > 0$ and $k \ge 1$, there exists $\delta > 0$ such that 
\PSpjUj cannot be solved in time
$\Oh(2^{\delta n} T^{k-\epsilon})$, where $n$ is the number of jobs and $T$ is the total processing time.
\end{corollary}
\noindent
This concludes the proof of~\cref{thm:SETH_to_PSpjUj}.

\section{Equivalences of Bin Packing}
\label{sec:binpacking_equivalence}

In this section, we show equivalence of \binpacking, \kpartition, \PCmax and various problem variants under parameter-preserving reductions.
Recall that in the ``\textsc{with targets}'' variants of \kpartition, every bin has a different capacity $t_1, \dots, t_k \in \mathbb N$ given with the input. In the ``\textsc{Bounded}'' variants of \kpartition, we are guaranteed that the $n$  input integers lie in the range $[W(1 - 1/n^{10}), W]$ for some integer $W \in \mathbb N$. See \cref{app:problem-definitions} for the detailed problem definitions.

\begin{theorem}\label{thm:equivalence-bin-packing}
	The following problems are equivalent under parameter-preserving reductions:
	\binpacking, \kpartition, \kpartitionTargets, \kpartitionBounded,
		 \binpackingMultisets,  \kpartitionMultisets, \kpartitionTargetsMultisets, \kpartitionBoundedMultisets, 
		 \PCmax and \QCmax.
\end{theorem}

The proof of Theorem~\ref{thm:equivalence-bin-packing} follows from the lemmas and observations established in this section (see the blue-highlighted cycle of reductions in \cref{fig:graph_reduction}).

To organise the section, we group the reduction into the following categories: set vs.\ multiset variants, bin packing vs.\ partition, target and bounded assumptions, and packing vs.\ scheduling problems. We note, that some of the reductions are folklore or can be found in the literature. For example, the relation between \kpartition and \binpacking was independently observed in~\cite{dreier2019complexity,heeger2023single}. We include them for completeness. 

\subsection{Set vs. Multiset Variants}
We first observe the following trivial relations between problems in \cref{thm:equivalence-bin-packing}.

\begin{fact}\label{lem:BoundedPart_to_Part}
	\kpartitionBounded is a special case of \kpartition.
\end{fact}

\begin{fact}\label{lem:Part_to_TargetPart}
	\kpartition is a special case of \kpartitionTargets
\end{fact}

\begin{fact}\label{lem:Part_to_BP}
	There is a parameter-preserving reduction from \kpartition to \binpacking.
\end{fact}
\begin{proof}
	One can decide the \kpartition problem on instance $X$ by calling an algorithm for \binpacking with the set of integers $X$ and bin capacity $T = \Sigma(X) / k$. 
	Indeed, if a partition $X_1, \dots, X_k$ of $X$ satisfies $\Sigma(X_i) \leq \Sigma(X)/k$ for all $i \in [k]$, then it necessarily also satisfies $\Sigma(X_i) = \Sigma(X)/k$ for all $i \in [k]$.
\end{proof}

All the above problems are special cases of the variants where the input integers are given as a multiset. In particular, we can generalise \binpacking and \kpartitionTargets to their multiset variants.

\begin{fact}\label{lem:binpacking_to_binpackingMultisets}
	\binpacking is a special case of \binpackingMultisets.
\end{fact}

\begin{fact}\label{lem:TargetPart_to_TargetPartMultisets}
	\kpartitionTargets is a special case of \kpartitionTargetsMultisets.
\end{fact}

In the following we show that this relation holds both ways for \kpartitionBounded. This justifies the focus on multiset variants to show the converse reductions of the above facts.

\begin{lemma}\label{lem:MultiBoundedPart_equiv_BoundedPart}
	\kpartitionBoundedMultisets and \kpartitionBounded are equivalent under parameter-preserving reductions. 
\end{lemma}
\begin{proof}
	Note that \kpartitionBounded is a special case of \kpartitionBoundedMultisets. We show the converse reduction, which is similar to the  transformation of multisets into sets of \cref{lem:pbtwomulti_to_pbtwo}. 
	Let $X \subset \N \cap [W(1 - 1/n^{10}), W]$ be a multiset of size $n$ for some integer $W \in \mathbb N$. 
	Without loss of generality, we may assume that $n$ is large enough and $W \geq n^{10}$ as otherwise the instance can be solved in polynomial time.
	Consider an arbitrary order on the elements $X = \{x_1, \dots, x_n\}$ (with multiplicities) and define the sets $Y \coloneq \{y_1, \dots, y_n\}$ and $Z \coloneq \{z_1, \dots, z_n\}$ where	for each $j \in [n]$ 
	\begin{align*}
		y_j &\coloneq Wn^{20} + (x_j - W)n^{7} - n^5 + j \\
		z_j &\coloneq Wn^{20} -j.
	\end{align*} 
	Note that for any $j \in [n]$, since $x_j \leq W$, we have $y_j \leq Wn^{20} - n^5 + n < Wn^{20} - n \leq z_j $. So the sets $Y$ and $Z$ are disjoint, and we can define the set $X' \coloneq Y \cup Z$.
	Furthermore, for any $j\in [n]$ we have $z_j \leq Wn^{20} - 1 \leq Wn^{20}$ and
	\begin{align*}
		y_j &\geq Wn^{20} + \left(W\left(1-\frac{1}{n^{10}}\right) - W\right)n^{7} - n^5 + 1 \tag{since $x_j \geq W(1-1/n^{10})$}\\
		&\geq Wn^{20} - \frac{W n^7}{n^{10}} - n^5 \\
		&=Wn^{20}\left(1 - \frac{n^7}{n^{10}n^{20}} - \frac{n^5}{W n^{20}} \right) \\ 
		&\geq Wn^{20}\left(1 - \frac{n^7}{n^{30}} - \frac{n^5}{n^{30}} \right) \tag{since $W \geq n^{10}$} \\ 
		&\geq Wn^{20}\left(1 - \frac{1}{(2n)^{10}} \right). \tag{since $n \geq 2$}
	\end{align*}
	Thus, the set $X'$ is in the range $X' \subset [W'(1-1/m^{10}), W']$ for $W' \coloneq Wn^{20}$ and $m \coloneq |X'| = 2n$. So $X'$ is an instance of \kpartitionBounded.  Since this construction takes $\Oh(n)$, it remains to show equivalence to prove that the reduction is parameter-preserving. 
	We first note that the average load $\mu' \coloneq \frac{1}{k} \Sigma(X')$ of the constructed integers is
	\begin{align*}
		\mu' 
		= \frac{1}{k} \sum_{j = 1}^n (y_j + z_j) =\frac{1}{k} \sum_{j=1}^n \left( 2 W n^{20} + (x_j - W)n^7 - n^5\right) =\frac{2n}{k} W n^{20} + \mu n^7 - \frac{n}{k} W n^7 - \frac{n}{k} n^5
	\end{align*}
	where $\mu \coloneq \frac{1}{k} \Sigma(X)$ is the average load of the given integers.

	Suppose that there is a partition of $X$ into $X_1, \dots, X_k$ such that $\Sigma(X_i) = \mu$ for all $i \in [k]$. Define $Y_i \coloneq \{y_j \ : \ x_j \in X_i\}$ and $Z_i \coloneq \{z_j \ : \ x_j \in X_i\}$ for every $i \in [k]$. Then the sets $X_i' \coloneq Y_i \uplus Z_i$ for $i \in [k]$ partition $X'$. Furthermore, by \cref{obs:bounded_subset_size}, we have $|X_i| = n/k$ for all $i \in [k]$. So
	\begin{align*}
		\Sigma(X_i') 
		&= \sum_{x_j \in X_i} \left(Wn^{20} + (x_j - W)n^{7} - n^5 + j\right) + \sum_{x_j \in X_i} \left(Wn^{20} - j \right)\\
		&= |X_i| 2 W n^{20} + \Sigma(X_i)n^7 - |X_i|Wn^7 - |X_i|n^5 \\
		&= \frac{n}{k} 2 W n^{20} + \mu n^7 - \frac{n}{k} Wn^7 - \frac{n}{k}n^5 = \mu.
	\end{align*} 

	Conversely, suppose that $X'$ can be partitioned into $X_1', \dots, X'_k$ such that $\Sigma(X_i') = \mu'$. 
	For $i \in [k]$, let $Y_i = X_i' \cap Y$, $Z_i = X_i' \cap Z$ and $X_i \coloneq \{x_j \ : \ y_j \in Y_i\}$. Then 
	\begin{align*}
		\Sigma(X_i') 
		&= \sum_{y_j \in Y_i} \left( Wn^{20} + (x_j - W)n^{7} - n^5 + j  \right) + \sum_{z_\ell \in Z_i} \left(Wn^{20} - \ell \right) \\
		&= |X_i'| Wn^{20} + \Sigma(X_i)n^7 - |X_i|Wn^7 - |X_i|n^5 + \sum_{y_j \in Y_i} j - \sum_{z_\ell \in Z_i} \ell. 
	\end{align*}
	Note that 
	$|\sum_{y_j \in Y_i} j - \sum_{z_\ell \in Z_i} \ell| < n^5/2$. So by taking the equation $\Sigma(X_i')= \mu'$ modulo $n^5$ we deduce that $\sum_{y_j \in Y_i} j - \sum_{z_\ell \in Z_i} \ell = 0$. Furthermore, since $0 \leq |X_i|n^5 < n^7$, by taking the equation $\Sigma(X_i')= \mu'$ modulo $n^7$ we get 
	$|X_i| = n/k$. Additionally, by \cref{obs:bounded_subset_size}, $|X_i'| = m / k = 2n/k$. So the equation $\Sigma(X_i')= \mu'$ implies $\Sigma(X_i) = \mu$. As this holds for every $i \in [k]$, we obtain a partition $X = X_1 \uplus \dots \uplus X_k$ with $\Sigma(X_i) = \mu$ for every $i \in [k]$. 
	\end{proof}

\subsection{Bin Packing vs. Partition}

To show equivalence between \binpackingMultisets and \kpartitionMultisets, we need the following tool adapted from \cite{RohwedderW25}.
As this is a generalisation of \cite[Lemma 18]{RohwedderW25}, we include the proof of \cref{lem:genRohwedderW25}
in \cref{app:technical-lemmas} for completeness.

\begin{restatable}{lemma}{genrohwedder}\label{lem:genRohwedderW25}
	For any $\tau, k \in \N$, there exists a multiset of integers $P \subset \N$
	of size $|P| \leq \Oh(k^2 \log \tau)$, and sum $\Sigma(P) = \tau$ such that for every $(t_1, \dots, t_k) \in \N^k$ with $t_1 + \dots + t_k = \tau$ there exists a partition of $P$ into $k$ subsets $P_1, \dots, P_k$ with $\Sigma(P_i) = t_i$ for all $i \in [k]$. 
	Furthermore, such a multiset $P$ can be constructed in time $\Oh(k^2 \log \tau)$.  
\end{restatable}

\begin{lemma}\label{lem:MultiBP_to_MultiPart}
	\binpackingMultisets and \kpartitionMultisets are equivalent under parameter-preserving reductions.
\end{lemma}

\begin{proof}
	The parameter-preserving reduction from \kpartitionMultisets to \binpackingMultisets works verbatim as in \cref{lem:Part_to_BP}.
	To reduce \binpackingMultisets to \kpartitionMultisets, consider an instance $(X, T)$ of \binpackingMultisets.
	In $\Oh(k^2 \log(kT)) = \Oh(\log T)$ time, construct the multiset $P$ from \cref{lem:genRohwedderW25} with parameters $k$ and $\tau \coloneq k T - \Sigma(X)$. 
	Let $X' \coloneq X \cup P$ be the obtained instance of \kpartitionMultisets with parameter $T' \coloneq \Sigma(X') / k$. We have $|X'| = \Oh(|X| + \log T )$ and $\Sigma(X') = \Sigma(X) + \Sigma(P) = \Sigma(X) + \tau = k T $, i.e.~$T = T'$. 

	Suppose that there exists a partition of $X'$ into $k$ subsets $X_1', \dots, X_k'$ with $\Sigma(X_1') = \dots = \Sigma(X_k')$. Then $\Sigma(X_i') = \Sigma(X')/ k = T$ for every $i \in [k]$.
	Partition $X$ into $X_1, \dots, X_k$ where $X_i \coloneq X_i' \cap X$. Then clearly $\Sigma(X_i) \leq \Sigma(X_i') = T$ for every $i \in [k]$.
	Conversely, suppose that there exists a partition of $X$ into $k$ subsets $X_1, \dots, X_k$ with $\Sigma(X_i) \leq T$ for every $i \in [k]$. Let $t_i \coloneq T - \Sigma(X_i) \geq 0$. Then $t_1 + \dots + t_k = \tau$, so by \cref{lem:genRohwedderW25} there exists a partition of $P$ into $k$ subsets $P_1, \dots, P_k$ such that $\Sigma(P_i) = t_i$ for every $i \in [k]$. So $X'$ can be partitioned into $X_1', \dots, X_k'$ where $X_i' \coloneq X_i \cup P_i$ satisfies $\Sigma(X_i') = \Sigma(X_i) + t_i = T$ for all $i \in [k]$.
\end{proof}

\subsection{Target and Bounded assumptions}
As for the set variants, we can naturally reduce \kpartitionBoundedMultisets to \kpartitionMultisets, which in turn can trivially be reduced to \kpartitionTargetsMultisets. We show that the converse reductions also hold.

\begin{lemma}\label{lem:MultiParttargets_equiv_MultiPart}
	\kpartitionMultisets and \kpartitionTargetsMultisets
	are equivalent under parameter-preserving reductions.
\end{lemma}
\begin{proof}
	Note that \kpartitionMultisets is a special case of \kpartitionTargetsMultisets where all targets are equal. For the other direction, consider a \kpartitionTargetsMultisets instance with items $X$ and targets $t_1, \dots, t_k$. For every $i \in [k]$ we add $k$ dummy items $d_i \coloneqq T - t_i$ for all $i \in [k]$, where $T \coloneqq 3 \cdot \max\{t_1, \dots, t_k\},$ and consider the new set of items $X' \coloneq X \cup \{d_1, \dots, d_k\}$ and bin capacity $T$ of \kpartitionMultisets. Clearly, the construction is parameter-preserving. We claim that there exists a partition of $X$ into $X_1, \dots, X_k$ such that $\Sigma(X_i) = t_i$ for all $i \in [k]$ if and only if there exists a partition of $X'$ into $X_1', \dots, X_k'$ such that $\Sigma(X_i') = T$ for all $i \in [k]$. See \cref{fig:MultiBptargets} for an illustration.

	\begin{figure}[t!]
    \centering
    \resizebox{\textwidth}{!}{%
        \begin{tikzpicture}[
        node distance=0pt,
        box/.style={rectangle,draw,minimum height=.8cm},
        padding/.style={rectangle,draw,minimum width=25pt, minimum height=.8cm, fill=gray!80}]

    \begin{scope}
    \node[box,fill=cb_orange,from={1,0 to 2,.8}] (j1) at (0, 0) { };
    \node[box,fill=cb_orange,from={2,0 to 4,.8}] (j2) [right=of j1] { };
    \node[box,fill=cb_orange,from={4,0 to 5,.8}] (C) [right=of j2] { };    
    \node[box,fill=gray,from={5,0 to 21,.8}] (P1) [right=of C] {$d_1$};    
    \node [left=.5cm of j1] {$X_1'$};
    \end{scope}
    
    \begin{scope}[yshift=-1cm]
    \node[box,fill=OI_lightblue,from={1,0 to 3,.8}] (j1) at (0, 0) { };
    \node[box,fill=OI_lightblue,from={3,0 to 6,.8}] (j2) [right=of j1] { };
    \node[box,fill=OI_lightblue,from={6,0 to 7,.8}] (F) [right=of j2] { };   
    \node[box,fill=gray,from={7,0 to 21,.8} ] (P2) [right=of F] {$d_2$};    
    \node [left=.5cm of j1] {$X_2'$};
    \end{scope}

    \begin{scope}[yshift=-2cm]
    \node[box,fill=OI_yellow,from={1,0 to 2,.8}] (j1) at (0, 0) { };
    \node[box,fill=OI_yellow,from={2,0 to 3,.8}] (j2) [right=of j1] { };
    \node[box,fill=OI_yellow,from={3,0 to 4,.8}] (I) [right=of j2] { };    
    \node[box,fill=gray,from={4,0 to 21,.8}] (P3) [right=of I] {$d_3$};    
    \node [left=.5cm of j1] {$X_3'$};
    \end{scope}

    \begin{scope}[yshift=-3cm]
        \draw[line width=1pt, ->, >={Stealth}] ([yshift=(-1cm),xshift=(-.2cm)]j1.west)--([yshift=(-1cm),xshift=(1cm)]P3.east);

        \draw[line width=1pt]  ([yshift=(+2.5cm)]j1.north west)--(0, -.1) node[below] {$0$};
        \draw[dashed] ([yshift=(+.1cm)]C.north east) -- ([yshift=(-3.1cm)]C.east) node[below] {$t_1$};
        \draw[dashed] ([yshift=(+.1cm)]F.north east) -- ([yshift=(-2.1cm)]F.east) node[below] {$t_2$};
        \draw[dashed] ([yshift=(+.1cm)]I.north east) -- ([yshift=(-1.1cm)]I.east) node[below] {$t_3$};
        \draw[dashed] ([yshift=(+.5cm)]P1.north east) -- ([yshift=(-1.5cm)]P3.south east) node[below] {$T = 3 \cdot t_2$};
        \node at ([yshift=(-1.4cm),xshift=(1cm)]P3.east) {$\Sigma(X_i')$};
    \end{scope}
    
    \end{tikzpicture}
    }%
\caption{Adding dummy items (grey) in \cref{lem:MultiParttargets_equiv_MultiPart} to ensure a one-to-one mapping between solutions of \kpartitionTargetsMultisets and solutions of \kpartitionMultisets.}
\label{fig:MultiBptargets}
\end{figure}

	For the forward direction, we can construct $X_i' \coloneq X_i \cup \{d_i\}$ for all $i \in [k]$. Since $\Sigma(X_i) = t_i$, we have $\Sigma(X_i') = t_i + (T - t_i) = T$ for all $i \in [k]$. For the backward direction, consider a partition of $X'$ into $X_1', \dots, X_k'$ such that $\Sigma(X_i') = T$ for all $i \in [k]$. Since $T = 3 \cdot \max\{t_1, \dots, t_k\}$ and every dummy item $d_i$ has value at least $2 \cdot \max\{t_1, \dots, t_k\}$, every bin $X_i'$ contains exactly one dummy item. For each $j \in [k]$, we can construct $X_j \coloneqq X_i' \setminus \{d_j\}$, where $X_i'$ is the bin containing $d_j$. Since $\Sigma(X_i') = T$, we have $\Sigma(X_j) = T - d_j = t_j$ for all $j \in [k]$.
\end{proof}

\begin{lemma}\label{lem:MultiPart_to_MultiBoundedPart}
	\kpartitionMultisets and \kpartitionBoundedMultisets
	are equivalent under parameter-preserving reductions. 
\end{lemma}
\begin{proof}
	Note that \kpartitionBoundedMultisets is a special case of \kpartitionMultisets.
	To reduce \kpartitionMultisets to the \textsc{Bounded} variant, consider an instance $X \subset \N$ of \kpartitionMultisets of size $n$ and parameter $T$. Note that the instance is trivial if $T \neq \frac{1}{k}\Sigma(X)$, so assume $T = \frac{1}{k}\Sigma(X)$.
	Let $m \coloneq n  k$ and $W \coloneq m^{10} \cdot \max{(X)}$. We build the multiset
	$$
	X' \coloneq \{x + W \ : \ x \in X\}
	$$
	and define $X''$ to be the multiset $X'$ with additionally $n (k-1)$ copies of the integer $W$. 
	So $X''$ contains $m$ integers 
	and we can verify that $X'' \subset [W' (1 - 1/m^{10}), W']$ for $W' \coloneq \lfloor W / (1 - 1/m^{10})\rfloor$. Indeed, we clearly have $\min(X'') \geq W \geq W'(1 - 1/m^{10})$. For the upper bound, note that $\max(X'') = \max(X) + W = W (1 + 1/m^{10}) \leq W / (1 - 1/m^{10})$ where the inequality follows from $(1 + 1/m^{10}) (1 - 1/m^{10}) = 1 - 1/m^{2} \leq 1$. 
	Since $\max(X'')$ is integral, we obtain $\max(X'')\leq W'$.
	Therefore, $X''$ is an instance of \kpartitionBoundedMultisets of size $m = \Oh(n)$ and parameter $W' = W n^{\Oh(1)}$. 
	The construction takes time $\Oh(n)$ so it is parameter-preserving. 
	Observe that the average load of the constructed multiset is 
	\begin{align*}
		\frac{1}{k}\Sigma(X'') =  \frac{1}{k} \left(\Sigma(X) + |X|W + n(k-1)W\right) = \frac{1}{k} (\Sigma(X) + n k W) = \frac{1}{k} \Sigma(X) + nW.
	\end{align*}
	It remains to verify that the instances are equivalent.

	Suppose that there exists a partition of $X$ into $k$ subsets $X_1, \dots, X_k$ with $\Sigma(X_i) = \Sigma(X) / k$ for every $i \in [k]$. Let $X_i''$ be the items in $X'$ corresponding to $X_i$ and $n - |X_i|$ copies of $W$. Then $X_1'', \dots, X_k''$ partitions the multiset $X''$ and for every $i \in [k]$ we have
	\begin{align*}
		\Sigma(X_i'') = \Sigma(X_i) + |X_i| W  + (n - |X_i|) W = \Sigma(X)/k + nW = \Sigma(X'')/ k.
	\end{align*}

	Conversely, suppose that there exists a partition of $X''$ into $k$ subsets $X_1'', \dots, X_k''$ with $\Sigma(X_i'') = \Sigma(X'')/k$ for all $i \in [k]$. 
	If there exists $i \in [k]$ such that $|X_i''|  \neq n$, then in particular there exists $i \in [k]$ such that $|X_i''| < n$. But then $\Sigma(X_i'') \leq (n-1) \max(X) + (n-1) W < nW$, which contradicts the assumption that $\Sigma(X_i'') = \Sigma(X'') / k$. Therefore, $|X_i''| = n$ for all $i\in[k]$.
	Let $X_i$ be the set of items in $X$ corresponding to $X_i'' \cap X'$. 
	Then $X_1, \dots, X_k$ partitions $X$ and, since $\Sigma(X'') / k = \Sigma(X_i'') = \Sigma(X_i) + nW$, we have  $\Sigma(X_i) = \Sigma(X) / k$ for every $i \in [k]$. 
\end{proof}

\subsection{Packing vs Scheduling Problems}
We show that \binpackingMultisets is equivalent to \PCmax and \QCmax. 

\begin{fact}\label{lem:PkCmax_equiv_MultiBP}
	\binpackingMultisets and \PCmax are equivalent problems.
\end{fact}
\begin{proof}
	Any instance of \binpackingMultisets with integers $x_1, \dots x_n$ and bin capacity $T$ is equivalent to the instance of \PCmax with processing times $p_1 = x_1, \dots, p_n = x_n$ and makespan objective $T$.
\end{proof}

\begin{lemma}\label{lem:QkCmax_equiv_PkCmax}
	\PCmax and \QCmax are equivalent under parameter-preserving reductions. 
\end{lemma}
\begin{proof}

	Note that \PCmax is the special case of \QCmax where all machines have the
	same speed $s_1 = \dots = s_k = 1$. To reduce \QCmax to \PCmax, consider a
	\QCmax instance with processing times $p_1, \dots, p_n \in \mathbb N$,
	speeds $s_1, \dots, s_k \in (0, 1]$ and objective $M \in \mathbb{R}_{\ge 0}$. Our goal is to find an assignment $\pi : [n] \to [k]$ such that $\sum_{j :\pi(j) = i} p_j \leq \floor{s_i M}$ for every $i \in [k]$, as $p_j$'s are integrals. We use a similar idea as in \cref{lem:MultiParttargets_equiv_MultiPart} to add long jobs to increase the load of each machine.

	Let $M' \coloneqq 3 \cdot \max_{i \in [k]} \{\floor{s_i M}\}$. For every $i \in [k]$ we introduce a dummy job $p_i' \coloneqq M
	- \floor{s_i \cdot M}$ and consider the new set of jobs $P' \coloneqq
	\{p_1, \dots, p_n\} \cup \{p_1', \dots, p_k'\}$ and the makespan objective $M'$
	of \PCmax. Clearly, the construction is parameter-preserving.
	We claim that there exists an assignment $\pi : [n] \to [k]$ such that $\sum_{j :
	\pi(j) = i} p_j \leq \floor{s_i M}$ for every $i \in [k]$ if and only if
	the makespan of $P'$ on $k$ parallel identical machines is at most $M'$.

	For the forward direction, consider a schedule of the $n$ original jobs on the $k$
	uniform machines such that the makespan of the $i$-th machine with speed $s_i$ is at most $M$.
	Then the same schedule on $k$ identical machines (i.e.~all of speed $1$) is
	such that the makespan of the $i$-th machine is at most $\floor{s_i \cdot M}$. By
	additionally scheduling the $i$-th dummy job $p_i'$ on the $i$-th machine, we
	obtain a schedule of all $n+k$ jobs on $k$ identical machines such that the
	total makespan is at most $M'$. 

	Conversely, consider a schedule of all $n + k$ jobs on $k$ identical
	machines such that the total makespan is at most $M'$. Since all dummy jobs
	have processing time at least $2/3 M'$, and since there are as many dummy
	jobs as there are machines, each machine necessarily schedules exactly one
	dummy job. Without loss of generality, we can assume that the $i$-th dummy job
	$p_i'$ is scheduled on the $i$-th machine. Let $X_i$ be the remaining jobs
	scheduled on the $i$-th machine. Then we have $\Sigma(X_i) + p_i' \leq M'$
	and thus $\Sigma(X_i) \leq \floor{s_i M} \le s_i M$ for every $i \in
	[k]$. Hence, the schedule of the $n$ first jobs on the $k$ uniform
	machines with speeds $s_1,\ldots,s_k$ given by $X_1, \dots, X_k$ has a total
	makespan of at most $M'$.
\end{proof}

\subsection{Relation between Bin Packing and \textsc{Weak} Grouped $k$-way Partition}

Finally, we formally establish the relation between \binpacking and \pbtwo. Showing
the same for arbitrary $q$ would yield a SETH-based lower bound for \binpacking.

\begin{lemma}\label{lem:BoundedPart_to_WeakGroupedPart}
	The \kpartitionBounded problem is equivalent to the special case of \pbtwo where the number of groups is restricted to $q=1$.
\end{lemma}
\begin{proof}
	Consider an instance of \pbtwo with $q=1$, i.e.~a set $G \subset \N \cap [W (1 - 1/n^{10}), W]$ of size $n = s k$ for some integer $W \in \mathbb N$.
	Without loss of generality, we may assume that $W \geq n^{10}$ as otherwise the instance can be solved in polynomial time. 
	Since $q=1$, we can rewrite the \ref{enum:pb2_yes} and \ref{enum:pb2_no} cases of \pbtwo as:
	\begin{yesenum} 
		\item There is a partition of $G$ into subsets $S_1, \dots, S_k$ such that
	\begin{yesenum}
		\item\label{enum:pb2q1_1a} $|S_{\ell}| = s$ for all $\ell \in [k]$,
		\item\label{enum:pb2q1_1b} $\Sigma(S_\ell) = \frac{1}{k}\Sigma(G)$ for all $\ell \in [k]$.
	\end{yesenum}
	\end{yesenum}
	\begin{noenum}
		\item  There is no partition of $G$ into subsets $S_1, \dots, S_k$ such that 
		\begin{noenum}
			\item\label{enum:pb2q1_2a} $|S_{\ell}| \leq s$ for all $\ell \in [k]$,
			\item\label{enum:pb2q1_2b} $\sum_{\ell \in [k]} |S_{\ell}| \geq k s$,
			\item\label{enum:pb2q1_2c} $\Sigma(S_{\ell}) = \frac{1}{k}\Sigma(G)$ for all $\ell \in [k]$.
		\end{noenum}
	\end{noenum}
	Note that by \cref{obs:bounded_subset_size}, any partition $S_1 \uplus \ldots \uplus S_k = G$ with $\Sigma(S_i) =  \frac{1}{k} \Sigma(G)$ for every $i \in [k]$ necessarily has $|S_i| = |G|/k = s$. Hence, \cref{enum:pb2q1_1a} is implied by \cref{enum:pb2q1_1b} and \cref{enum:pb2q1_2a} is implied by \cref{enum:pb2q1_2c}. Furthermore, \cref{enum:pb2q1_2b} holds for any partition of $G$. So \pbtwo when $q=1$ really asks whether there is a partition $S_1, \dots, S_k$ of $G$ such that $\Sigma(S_1) = \dots = \Sigma(S_k) = \frac{1}{k}\Sigma(G)$. This is exactly the problem of \kpartitionBounded with input set $G$. 
\end{proof}

\section{Conclusion and Future Work}
\label{sec:conclusion}

In this work, we proved a tight ETH-based lower bound for \binpacking, answering
an open question posed by Jansen, Kratsch, Marx, and
Schlotter~\cite{jansen2013bin}. Our main result
(\Cref{thm:bin-packing-lowerbound}) rules out time $2^{o(n)} T^{o(k)}$ for
\binpacking under ETH, improving upon Jansen et al.'s bound of $(nT)^{o(k/\log k)}$ by
eliminating the logarithmic factor loss in the exponent.

The hardness of \binpacking has long been recognised as a central bottleneck in
parameterised complexity. The $\log k$-factor gap in the lower bound of Jansen
et al.~\cite{jansen2013bin} was inherited by numerous reductions to
other problems, making it a fundamental barrier to establishing tight complexity
bounds.  Our result resolves this issue and yields improved ETH-based lower
bounds for a wide range of problems in parameterised complexity, as detailed in
\cref{app:further-applications}.

While our ETH-based lower bound for \binpacking is tight, we leave as an
intriguing open problem whether a tight SETH-based lower bound can be
established, i.e.~whether \binpacking can be solved in time $2^{o(n)} T^{k-1-\varepsilon}$ for
any $\varepsilon > 0$. 

This may be a challenging task. 
One reason is that \binpacking is closely related to the Set Cover problem~\cite{NederlofPSW23}, but SETH-based lower bounds for Set Cover remain a notoriously open problem~\cite{DBLP:journals/talg/CyganDLMNOPSW16}.
Another reason is that there have been surprising algorithmic breakthroughs for \binpacking in some settings: Nederlof et al.~\cite{NederlofPSW23} showed that for every $k \ge 2$ there exists $\eps_k >0$ such that \binpacking can be solved in time
$\Oh((2-\eps_k)^n)$, thus shattering the hope of a lower bound ruling out time $\Oh((2-\eps)^n)$ for any $\eps>0$ for, say, \binpacking[10]. It is conceivable that a similar algorithmic breakthrough could also be possible for pseudopolynomial algorithms, and \binpacking could be solved in time $2^{o(n)} T^{k-1-\eps}$ after all.



\paragraph{SETH Lower Bounds for Scheduling Problems}

We demonstrate that this barrier can be overcome for several multi-machine
scheduling problems that are harder than \binpacking. Specifically, we
prove tight SETH-based lower bounds for the scheduling problems \PrjCmax, \PSwjCj, \PLmax,
$\PSUj$, and \PSpjUj (see~\cref{tab:summary}). This establishes optimality of
several classical algorithms from the 60s and 70s~\cite{lawler1969functional,rothkopf1966scheduling,lawler}.

After Abboud et al.~\cite{AbboudBHS19} proved an SETH-based lower
bound for \subsetsum, a fruitful research direction has emerged: improving the
dependence on $n$ in the running times of pseudopolynomial-time
algorithms~\cite{DBLP:conf/soda/Bringmann17,bringmann2022faster,fischer2025minimizing,polak2021knapsack,klein2023minimizing,bringmann2024even,bringmann2024knapsack}.
For instance, Fischer and Wennmann~\cite{fischer2025minimizing} improved the running time for \PSpjUj
from $\Oh(nT^k)$ to $\Ot(n+T^k)$.  This raises the natural
question: can similar improvements be achieved for the other problems we study?
Specifically, can \PrjCmax, \PSwjCj, and \PSUj all be solved in
time $\Ot(n + T^{k-1})$? Similarly, to \subsetsum, one can also study the running time dependence 
parameterised by $p_{\max}$, the largest processing time. Perhaps all \PrjCmax, \PSwjCj, and \PSUj can be solved in time $\Ot(n + (p_{\max})^{k-1})$?

Finally, another interesting research direction is to study tight FPTASes for these scheduling problems. Indeed, our (S)ETH-based lower bounds opens the possibility of tight FPTASes. 
For example \PCmax admits an approximation scheme running in time $\Oh(n) + (1/\eps)^{\Oh(k)}$ as shown by Jansen et al.~\cite{jansen2010scheduling}. Prior to our work, Chen et al.~\cite{chen2018optimality} showed that a running time of $n^{\Oh(1)} + (1/\eps)^{\Oh(k^{1-\delta})}$ for $\delta > 0$ would contradict ETH.
Our results yield an improved lower bound, ruling out time $2^{o(n)} \cdot (1/\eps)^{o(k)}$ assuming ETH.\footnote{Indeed, plugging $\eps < 1/T$ into the time bound $2^{o(n)} \cdot (1/\eps)^{o(k)}$ would yield an exact algorithm running in time $2^{o(n)} \cdot T^{o(k)}$, contradicting \cref{thm:bin-packing-lowerbound}.} Hence, the FPTAS by Jansen et al.~\cite{jansen2010scheduling} is tight up to constant factors in the exponent.

\bibliographystyle{abbrv}
\bibliography{lit}

\appendix 

\section{Problem Definitions}
\label{app:problem-definitions}

In the following problems, $k$ is treated as a constant and $n$ denotes the instance size. Every problem is defined with respect to a parameter. 

\hypertarget{prob:binpacking}{
\paragraph*{Bin Packing}}
We consider the following decision version of \binpacking.
The \binpackingMultisets problem is the variant where given integers form a multiset. 
We parameterise the problems by the bin capacity. 


\Problem{\binpacking}
{Set of $n$ integers $X \subset \N$; bin capacity $T \in \N$.}
{Is there a partition of $X$ into $X_1, \dots, X_k$ such that $\Sigma(X_\ell) \leq T$ for every $\ell \in [k]$?}
{$T$}


\hypertarget{prob:kpartition}{
\paragraph*{Multiway Partition}}
In the standard \kpartition problem, the goal is to partition a set of integers into $k$ parts of equal sum.
We extend the problem to the \kpartitionTargets where every part has to sum to a given value called the target. Additionally, we assume in \kpartitionBounded that given integers are within the same range. 
Finally, \kpartitionMultisets, \kpartitionTargetsMultisets and \kpartitionBoundedMultisets are the natural variants where the given integers form a multiset. 
We parameterise the problems by the (maximum) target value. 


\Problem{{\kpartition}}
{Set of $n$ integers $X \subset \mathbb N$.}
{Is there a partition of $X$ into $X_1, \dots, X_k$ such that $\Sigma(X_1) = \ldots = \Sigma(X_k)$?}
{$T \coloneq \Sigma(X) / k$}

\Problem{\hypertarget{prob:kpartitiontargets}{\kpartitionTargets}}
{Set of $n$ integers $X \subset \N$; targets $t_1, \dots, t_k \in \N$.}
{Is there a partition of $X$ into $X_1, \dots, X_k$ such that $\Sigma(X_\ell) = t_\ell$ for every $\ell \in [k]$?}
{$T \coloneq \max_{1 \le i \le k} t_i$}

\Problem{\hypertarget{prob:kpartitionbounded}{\kpartitionBounded}}
{Set of $n$ integers $X \subset \N \cap [W(1 - 1/n^{10}), W]$ for some integer $W \in \mathbb N$
. }
{Is there a partition of $X$ into $X_1, \dots, X_k$ such that $\Sigma(X_1) = \ldots = \Sigma(X_k)$?}
{$T \coloneq \Sigma(X) / k$}


\paragraph*{Scheduling problems}
We use the Graham et al.~\cite{graham1979optimization} notation explained in the introduction to define the following scheduling problems. These problems are defined as optimisation problems; however each of them has a natural decision variant, in which the task is to decide if the objective value is at most $\Lambda$ for some given $\Lambda \in \N$.


\Problem{\hypertarget{prob:pcmax}{\PCmax}}{$n$ jobs with processing times $p_1, \dots, p_n \in \N$.}{What is the minimum value of $C_{\max} \coloneq \max_{1 \le j \le n} C_j$ for any schedule of the $n$ jobs on $k$ identical machines?}{$T \coloneq \sum_{1 \le j \le n} p_j$}

\Problem{\hypertarget{prob:qcmax}{\QCmax}}
{$n$ jobs with processing times $p_1, \dots, p_n \in \N$; machine speed $s_1, \dots, s_k \in (0, 1]$.}
{What is the minimum value of $C_{\max} \coloneq \max_{1 \le j \le n} C_j$ for any schedule of the $n$ jobs on $k$ uniform machines with speeds $s_1, \dots, s_k$?}
{$T \coloneq \sum_{1 \le j \le n} p_j$}

\Problem{\hypertarget{prob:prjcmax}{\PrjCmax}}
{$n$ jobs with processing times $p_1, \dots, p_n \in \N$ and release dates $r_1, \dots, r_n \in \N$.}
{What is the minimum value of $C_{\max} \coloneq \max_{1 \le j \le n} C_j$ for any schedule of the $n$ jobs on $k$ identical machines such that every job is processed after its release date?}
{$T \coloneq \sum_{1 \le j \le n} p_j$}

\Problem{\hypertarget{prob:psuj}{\PSUj}}
{$n$ jobs with processing times $p_1, \dots, p_n \in \N$ and due dates $d_1, \dots, d_n \in \N$.}
{What is the minimum value of $\sum_{1 \le j \le n} U_j$ for any schedule of the $n$ jobs on $k$ identical machines, where $U_j \in \{0, 1\}$ indicates whether job $j$ completes after its due date $d_j$?}
{$T \coloneq \sum_{1 \le j \le n} p_j$}

\Problem{\hypertarget{prob:pspjuj}{\PSpjUj}}
{$n$ jobs with processing times $p_1, \dots, p_n \in \N$ and due dates $d_1, \dots, d_n \in \N$.}
{What is the minimum value of $\sum_{1 \le j \le n} p_j U_j$ for any schedule of the $n$ jobs on $k$ identical machines, where $U_j \in \{0, 1\}$ indicates whether job $j$ completes after its due date $d_j$?}
{$T \coloneq \sum_{1 \le j \le n} p_j$}

\Problem{\hypertarget{prob:pslmax}{\PLmax}}
{$n$ jobs with processing times $p_1, \dots, p_n \in \N$ and due dates $d_1, \dots, d_n \in \N$.}
{What is the minimum value of $\max_{1 \le j \le n} L_j$ for any schedule of the $n$ jobs on $k$ identical machines, where $L_j \coloneq C_j - d_j$ is the lateness of job $j$?}
{$T \coloneq \sum_{1 \le j \le n} p_j$}

\Problem{\hypertarget{prob:pstmax}{\PTmax}}
{$n$ jobs with processing times $p_1, \dots, p_n \in \N$ and due dates $d_1, \dots, d_n \in \N$.}
{What is the minimum value of $\max_{1 \le j \le n} T_j$ for any schedule of the $n$ jobs on $k$ identical machines, where $T_j \coloneq \max\{0, C_j - d_j\}$ is the lateness of job $j$?}
{$T \coloneq \sum_{1 \le j \le n} p_j$}

\Problem{\hypertarget{prob:pswjcj}{\PSwjCj}}
{$n$ jobs with processing times $p_1, \dots, p_n \in \N$ and weights $w_1, \dots, w_n \in \N$.}
{What is the minimum value of $\max_{1 \le j \le n} w_j C_j$ for any schedule of the $n$ jobs on $k$ identical machines?}
{$T \coloneq \sum_{1 \le j \le n} p_j$}



\paragraph*{\pbtwoYes}
Finally, we introduce the following variant of \kpartition, which is used as an intermediate problem to show SETH-hardness of various scheduling problems. In \pbtwoYes, the given integers are grouped into equal-sized groups, and the goal is to partition all integers into parts of equal sum such that each group is partitioned into equal-sized parts. Intuitively, the goal is to pack all items into $k$ \emph{bins} of equal capacity, as in \kpartition, with the additional constraint that each \emph{bin} needs to contain an equal number of items from each group. 

\Problem{\hypertarget{prob:pbtwoyes}{\pbtwoYes}}
{Sets of integers $G_1, \dots, G_q \subset \N \cap [W(1 - 1/n^{10}), W]$ with $|G_i| = ks$ for every $i \in [q]$ and $n = k s q$ for integers $s, q, W \in \N$.
}
{Decide whether for all $i \in [q]$ there is a partition of $G_i$ into $k$ subsets $S_{i, 1}, \dots, S_{i, k}$ such that
\begin{enumerate}
	\item $|S_{i, \ell}| = s$ for all $i \in [q]$ and $\ell \in [k]$,
	\item $\sum_{i \in [q]} \Sigma(S_{i, \ell}) = \frac{1}{k} \sum_{i \in [q]} \Sigma(G_i)$ for all $\ell \in [k]$.
\end{enumerate}
}
{$W$}
We weaken the above requirement in \pbtwo where the goal is to differentiate between two non-exhaustive cases, referred to as the \textsc{Yes} and \textsc{No} cases. 
Note that \pbtwo is a promise problem where the given instance is guaranteed to be either in the \textsc{Yes} or in the \textsc{No} case. We point out that the two cases handle $k-1$ disjoint subsets of each group instead of a partition into $k$ parts. 
\Problem{\hypertarget{prob:pbtwo}{\pbtwo}}
{Sets of integers $G_1, \dots, G_q \subset \N \cap [W(1 - 1/n^{10}), W]$ with $|G_i| = ks$ for every $i \in [q]$ and $n = k s q$ for integers $s, q, W \in \N$.
}
{Decide which of the following cases hold.
\begin{yesenum}
	\item \label{enum:pb2_yes} $\forall i \in [q]$, there are disjoint subsets $S_{i, 1}, \dots, S_{i, k-1} \subset G_i$ such that
	\begin{yesenum}
		\item $|S_{i, \ell}| = s$ for all $i \in [q]$ and $\ell \in [k-1]$,
		\item $\sum_{i \in [q]} \Sigma(S_{i, \ell}) = \frac{1}{k} \sum_{i \in [q]} \Sigma(G_i)$ for all $\ell \in [k-1]$.
	\end{yesenum}
\end{yesenum}
\begin{noenum}
	\item \label{enum:pb2_no} It does not hold that $\forall i \in [q]$, there are disjoint subsets $S_{i, 1}, \dots, S_{i, k-1} \subset G_i$ such that
	\begin{noenum}
		\item $|S_{1, \ell}| + \dots + |S_{i, \ell}| \leq i \cdot s$ for all $i \in [q]$ and $\ell \in [k-1]$,
		\item $\sum_{i\in [q], \ell \in [k-1]} |S_{i, \ell}| \geq (k-1) q s$,
		\item $\sum_{i \in [q]} \Sigma(S_{i, \ell}) = \frac{1}{k} \sum_{i \in [q]} \Sigma(G_i)$ for all $\ell \in [k-1]$.
	\end{noenum}
\end{noenum}
}
{$W$}
Naturally, \pbtwo generalises to \pbtwoMultisets, where the given integers form multisets. An additional variant to consider is when each bin is required to sum to a given value, called the target. 
\Problem{\hypertarget{prob:pbtwotargetsmultisets}{\pbtwoTargetsMultisets}}
{Multisets $G_1, \dots, G_q \subset \N \cap [W(1 - 1/n^{10}), W]$ with $|G_i| = ks$ for every $i \in [q]$ and $n = k s q$ for integers $s, q, W \in \N$; and targets $t_1, \dots, t_{k-1} \in \N$.}
{Decide which of the following cases hold.
\begin{yesenum}
	\item \label{enum:pb2targets_yes} $\forall i \in [q]$, there are disjoint subsets $S_{i, 1}, \dots, S_{i, k-1} \subset G_i$ such that:
	\begin{yesenum}
		\item $|S_{i, \ell}| = s$ for all $i \in [q]$ and $\ell \in [k-1]$,
		\item $\sum_{i \in [q]} \Sigma(S_{i, \ell}) = t_\ell$ for all $\ell \in [k-1]$.
	\end{yesenum}
\end{yesenum}
\begin{noenum}
	\item \label{enum:pb2targets_no} It does not hold that $\forall i \in [q]$, there are disjoint subsets $S_{i, 1}, \dots, S_{i, k-1} \subset G_i$ such that:
	\begin{noenum}
		\item $|S_{1, \ell}| + \dots + |S_{i, \ell}| \leq i \cdot s$ for all $i \in [q]$ and $\ell \in [k-1]$,
		\item $\sum_{i\in [q], \ell \in [k-1]} |S_{i, \ell}| \geq (k-1) q s$,
		\item $\sum_{i \in [q]} \Sigma(S_{i, \ell}) = t_\ell$ for all $\ell \in [k-1]$.
	\end{noenum}
\end{noenum}
}
{$W$}

\paragraph*{Subset Sum}

We consider the classical \subsetsum problem and the \vectorsubsetsum variant.

\Problem{\hypertarget{prob:subsetsum}{\subsetsum}}
{Set of $n$ integers $X \subset \N$, target value $t \in \N$.}
{Is there a subset $S \subset X$ such that $\Sigma(S) = t$?}
{$t$}

\Problem{\hypertarget{prob:vectorsubsetsum}{\vectorsubsetsum}}
{Set of $n$ vectors $X \subset \N^k$, target vector $t \in \N^k$.}
{Is there a subset $S \subset X$ such that $\Sigma(S) = t$?}
{$T \coloneq \lVert t \rVert_{1}$}

\section{Further Implications of ETH-Tight Lower Bound for Bin Packing}\label{app:further-applications}

We list problems for which hardness results were proven via a reduction from \binpacking, using the $(nT)^{o(k/\log k)}$ lower bound for \binpacking by Jansen et al.~\cite{jansen2013bin}. Since we strengthen that lower bound in \cref{thm:bin-packing-lowerbound}, we directly strengthen the lower bounds for the following problems. 


\medskip
Lafond et al.~\cite{lafond2025cluster} studied the {Cluster Editing} problem 
on cographs, where the task is to edit at most $k$ edges in a given cograph $G$ such that the resulting graph is a disjoint union of at most~$p$ cliques. 
They showed that the problem can be solved in $n^{\Oh(p)}$ time and by a reduction from Bin Packing ruled out time $f(p) \cdot n^{o(p/\log p)}$ assuming ETH for any computable function $f(\cdot)$. 
By replacing the Bin Packing lower bound of Jansen et al.~\cite{jansen2013bin} with our \cref{thm:bin-packing-lowerbound}, their reduction immediately yields a higher lower bound.
\begin{corollary}\label{cor:cluster-editing}
	There is no $f(p) \cdot n^{o(p)}$-time algorithm for Cluster Editing on cographs parameterised by the number of cliques $p$ for any computable function $f(\cdot)$, unless ETH fails.
\end{corollary}

Javadi and Nikabadi~\cite{JAVADI20241} designed a $k^{t} \cdot n^{\Oh(k)}$-time algorithm for the $k$-Sparsest Cut problem in graphs of treewidth $t$, where the task is to partition the vertices of a graph into $k$ parts $S_1,\ldots,S_k$ to minimise $\max_{i \in [k]} |\delta(S_i)|/|S_i|$, where $\delta(S_i)$ is the set of edges with exactly one endpoint in $S_i$.
They complemented this result by showing a reduction from Bin Packing, which, by Jansen et al.~\cite{jansen2013bin} rules out time $n^{o((k+t)/\log(k+t))}$ assuming ETH~\cite[Theorem 4.3]{JAVADI20241}. By using \cref{thm:bin-packing-lowerbound} in their reduction, we improve their lower bound.
\begin{corollary}\label{cor:sparsest-cut}
	There is no $f(k,t) \cdot n^{o(k+t)}$-time algorithm for $k$-Sparsest Cut parameterised by the number of parts $k$ and the treewidth $t$ of the input graph, for any computable function $f(\cdot)$, unless ETH fails.
\end{corollary}

Bla{\v{z}}ej et al.~\cite[Theorem 1.2]{blavzej2024parameterized} reduced Bin Packing to the Eulerian Strong Component Arc Deletion (ESCAD) problem, where the task is to delete a minimum number of arcs in a directed graph so that every strongly connected component is Eulerian. They consequently excluded $f(v) \cdot n^{o(v/\log{v})}$-time algorithm for ESCAD assuming ETH, where $v$ is the vertex cover number of the input graph. 
Using~\cref{thm:bin-packing-lowerbound} immediately gives a higher lower bound.

\begin{corollary}\label{cor:escad}
	There is no $f(v) \cdot n^{o(v)}$-time algorithm for Eulerian Strong Component Arc Deletion parameterised by the vertex cover number $v$ of the input graph, for any computable function $f(\cdot)$, unless ETH fails.
\end{corollary}

Bliem et al.~\cite{bliem2016complexity} considered a problem of EEF-Allocation,
where we are given a set of $n$ agents and a set of $m$ resources, the task is
to find an allocation of the resources to the agents that is \emph{Pareto
efficient} and \emph{envy-free}~(see~\cite{bliem2016complexity} for an exact
definition). By a reduction from Bin Packing, they showed that EEF-Allocation with monotonic additive preferences
bounded by $z$ is $W[1]$-hard parameterised by the number of agents $b$ and excluded $f(b) \cdot z^{o(b/\log{b})}$-time algorithms for any computable function $f(\cdot)$ assuming ETH. With~\cref{thm:bin-packing-lowerbound} their reduction immediately yields a
$z^{\Omega(b)}$ lower bound.

\begin{corollary}\label{cor:eef-allocation}
	There is no $f(b) z^{o(b)}$-time algorithm for EEF-Allocation with monotonic additive preferences bounded by $z$ parameterised by the number of agents $b$ for any computable function $f(\cdot)$, unless ETH fails.
\end{corollary}

Bla{\v{z}}ej et al.~\cite{blavzej2024equitable} considered the Equitable Connected Partition (ECP) problem where for a given graph $G$ the task is to find a partition of its vertices into $p$ parts such that each part induces a connected subgraph of $G$ and the size of each pair of parts differ by at most $1$. They considered the parameter $\ell = p + \text{dist}_\mathcal{G}(G)$, where $\mathcal{G}$ is a graph family and $\text{dist}_\mathcal{G}(G)$ is the size of the minimum modulator of $G$ to $\mathcal{G}$.\footnote{A set $M \subseteq V(G)$ is a modulator to $\mathcal{G}$ if $G \setminus M$ is in $\mathcal{G}$.}
Bla{\v{z}}ej et al.~\cite[Theorem 33]{blavzej2024equitable} showed a reduction from Bin Packing, which allowed them to exclude $f(\ell) \cdot n^{o(\ell/\log{\ell})}$ algorithm for ECP parameterised by $\ell$ assuming ETH.
If instead of using the result of Jansen et al.~\cite{jansen2013bin} we use \cref{thm:bin-packing-lowerbound}, their reduction excludes $f(\ell)n^{o(\ell)}$-time algorithms for any computable function $f(\cdot)$. 
\begin{corollary}\label{cor:ecp}
	There is no $f(\ell) \cdot n^{o(\ell)}$-time algorithm for Equitable Connected Partition parameterised by $\ell = p + \text{dist}_\mathcal{G}(G)$, for any computable function $f(\cdot)$, unless ETH fails.
\end{corollary}


Dreier et al.~\cite{dreier2019complexity} studied the {Exact Path Packing} problem, where the task is to decide whether a given graph $G$ contains an embedding of $k$ vertex-disjoint paths $p_1,\ldots,p_k$ such that each edge of $G$ is covered exactly once. 
Assuming ETH, by a reduction from Bin Packing, they showed that the problem does not admit $f(k) \cdot n^{o(k/\log{k})}$-time algorithms for any computable function $f(\cdot)$, even when the treewidth of $G$ is 2. By using our lower bound (\cref{thm:bin-packing-lowerbound}), their reduction immediately yields a tight lower bound.
\begin{corollary}\label{cor:path-packing}
	There is no $f(k) \cdot n^{o(k)}$-time algorithm for Exact Path Packing parameterised by the number of paths $k$ on graphs of treewidth two for any computable function $f(\cdot)$, unless ETH fails.
\end{corollary}
Heeger et al.~\cite{heeger2023single} considered the scheduling problems $1 | \bar{d_j} | \Sigma w_j C_j$ and $1 | \bar{d_j} | \Sigma U_j$, where the task is to schedule $n$ jobs with $k$ distinct due dates on a single machine to minimise (i) the weighted completion times or (ii) the number of tardy jobs. They showed that both problems can be solved in $T^{\Oh(k)}$-time and showed that assuming ETH no $f(k) \cdot T^{o(k/\log{k})}$ algorithm exists by a reduction from Bin Packing (here $T$ is the total processing time and $f(\cdot)$ is any computable function). Using our lower bound for Bin Packing (\cref{thm:bin-packing-lowerbound}) in their reduction shows optimality of their result.
\begin{corollary}\label{cor:single-machine}
	There are no $f(k) T^{o(k)}$-time algorithms for the scheduling problems $1 | \bar{d_j} | \Sigma w_j C_j$ and $1 | \bar{d_j} | \Sigma U_j$ parameterised by the number of different due dates $k$ for any computable function $f(\cdot)$, unless ETH fails.
\end{corollary}

\section{SETH-hardness of Vector Subset Sum}
\label{app:vector-subset-sum}

In this section we prove a tight SETH-based lower bound for Vector Subset Sum, as we mentioned in \Cref{sec:techoverview_vectorsubsetsum}. See \Cref{app:problem-definitions} for a formal problem definition. We prove the following reformulation of \cref{thm:SETH_to_VectorSubsetSum} via the equivalence of \subsetsum and \vectorsubsetsum shown in \cref{lem:SubsetSum_equiv_VectorSubsetSum}.



\begin{theorem}[Reformulation of \Cref{thm:SETH_to_VectorSubsetSum}]\label{thm:SETH_to_VectorSubsetSum2}
	Assuming SETH, for every $\eps > 0$ and integer $k \geq 1$ there exists $\delta > 0$ such that $k$-dimensional \vectorsubsetsum cannot be solved in time $\Oh(2^{\delta n} T^{k-\eps})$, where $n$ is the number of items and $T = \|t\|_1$ is the 1-norm of the target vector $t$.
\end{theorem}

\begin{lemma}\label{lem:SubsetSum_equiv_VectorSubsetSum}
	There exists $\eps > 0$ and $k \geq 1$ such that for every $\delta > 0$, $k$-dimensional \vectorsubsetsum can be solved in time $\Oh(2^{\delta n} T^{k(1-\eps)})$ time if and only if
	there exists $\eps > 0$ such that for every $\delta > 0$ 
	\subsetsum can be solved in time $\Oh(2^{\delta n} t^{1-\eps})$ time.
\end{lemma}

\begin{proof}[Proof of \cref{thm:SETH_to_VectorSubsetSum2}]
	Assuming SETH, for every $\eps > 0$, there exists $\delta >0$ such that \subsetsum cannot be solved in time $\Oh(2^{\delta n} \cdot t^{1-\eps})$ by~\cite[Theorem 1]{AbboudBHS19}. \Cref{lem:SubsetSum_equiv_VectorSubsetSum}  implies that 
	for any $\eps > 0$ and $k \geq 1$, there exists $\delta > 0$ such that, $k$-dimensional \vectorsubsetsum cannot be solved in time $\Oh(2^{\delta n} \cdot T^{k(1-\eps)})$. Rescaling to $\eps' \coloneqq k\eps$, we equivalently obtain that for any $k \ge 1$ and $\eps' > 0$, there exists $\delta > 0$ such that $k$-dimensional \vectorsubsetsum cannot be solved in time $\Oh(2^{\delta n} \cdot T^{k-\eps'})$.
\end{proof}


\begin{proof}[Proof of \cref{lem:SubsetSum_equiv_VectorSubsetSum}]
%
\textbf{\vectorsubsetsum $\rightarrow$ \subsetsum}\;\;
Let $x_1, \dots, x_n, t \in \N^k$ be an instance of \vectorsubsetsum.
Observe that we can discard $x_i$ if there exists $\ell \in [k]$ with $x_i[\ell] > t[\ell]$ since such items can never be a part of a solution. Therefore, let $B \coloneq 2 n \cdot \lVert t \rVert_{\infty}$ and construct the integers 
\begin{align*}
\hat x_i \coloneq \sum_{\ell = 1}^k x_i[\ell] \cdot B^{\ell-1} \text{\qquad and \qquad}
\hat t \coloneq \sum_{\ell = 1}^k t[\ell] \cdot B^{\ell-1}
\end{align*}
for any $i \in [n]$. Clearly, a solution to the \vectorsubsetsum instance $x_1, \dots, x_n, t$ implies a solution to the \subsetsum instance $\hat x_1, \dots, \hat x_n, \hat t$. On the other hand, since for any $\ell \in [k]$ we have $\sum_{i=1}^n x_i[\ell] < B$, any \subsetsum solution $y_1, \dots, y_n \in \{0, 1\}$ such that $\sum_{i = 1}^n y_i \cdot \hat x_i = \hat t$ necessarily satisfies  
$
	\sum_{i = 1}^n y_i \cdot x_i[\ell] = t[\ell]
$
for all $\ell \in [k]$,
and thus is also a solution to \vectorsubsetsum.

Note that $\hat t \le B^k = 2^k n^k \lVert t \rVert_{\infty}^k \le 2^k n^k T^k$, where $T = \lVert t \rVert_{1}$. Hence, if there exists $\eps >0$ such that for every $\delta>0$ \subsetsum can be solved in time $\Oh{(2^{\delta n} {\hat t}^{1-\eps})}$, then \vectorsubsetsum can be solved in time $\Oh{(2^{\delta n} n^{k(1-\eps)} T^{k(1-\eps)})} \le \Oh{(2^{\delta' n} T^{k(1-\eps)})}$ for $\delta' = f(k, \eps) \cdot \delta$, where $f(k, \eps)$ is a function depending only on $k$ and $\eps$.

\subparagraph*{\subsetsum $\rightarrow$ \vectorsubsetsum}
Let $x_1, \dots, x_n, t \in \N$ be an instance of \subsetsum. We can assume that $x_i \le t$ for every $i \in [n]$ as items $x_i >t$ can never be part of a solution.
For any fixed dimension $k \geq 1$, we encode the integers in a $B$-ary system for $B \coloneq \lceil t^{1/k} \rceil$: let $t_\ell, x_{i, \ell} \in \{0, 1, \dots, B-1\}$ for all $i \in [n]$ and $\ell \in [k]$ such that 
$$x_i = \sum_{\ell = 1}^k x_{i, \ell} \cdot B^{\ell -1} \text{\qquad and \qquad} t = \sum_{\ell =1}^k t_\ell \cdot B^{\ell-1}.$$
We denote by $C$ the set of all sequences $(c_0,\ldots,c_k) \in \{0, 1, \dots, n-1\}^{k+1}$ with $c_0 = c_k = 0$ (this is the set of all possible carry sequences). For any $c \in C$ we construct the $k$-dimensional \vectorsubsetsum instance given by 
\begin{align*}
\hat x_i & \coloneq (x_{i, 1}, x_{i, 2}, \dots, x_{i, k}) \quad \text{for any } i \in [n], \\
\hat t_c  &\coloneq (t_1 + c_1 \cdot B - c_0, t_2 + c_2 \cdot B - c_1, \dots, t_k + c_k \cdot B - c_{k-1}).
\end{align*}

We claim that there exists $c \in C$ such that $(\hat x_1, \dots, \hat x_n, \hat{t}_c)$ is a \textsc{Yes}-instance of \vectorsubsetsum if and only if $(x_1, \dots, x_n, t)$ is a \textsc{Yes}-instance of \subsetsum.
Indeed, assume there exist $c \in C$ and 
$y_1, \dots, y_n \in \{0, 1\}$ such that $\sum_{i = 1}^n y_i \hat x_i = \hat{t}_c$. Then in particular, for every $\ell \in [k]$, we have $\sum_{i =1}^n y_i \cdot x_{i, \ell} = \sum_{i =1}^n y_i \cdot \hat x_i[\ell] = \hat{t}_c[\ell] = t_\ell + c_\ell \cdot B - c_{\ell -1}$. Thus,
\begin{align*}
	\sum_{i =1}^n y_i \cdot x_i &
	= \sum_{i=1}^n y_i \cdot \left(\sum_{\ell =1}^k x_{i, \ell} \cdot B^{\ell -1} \right) \\
	&= \sum_{\ell =1}^k B^{\ell-1} \left(\sum_{i=1}^n y_i \cdot x_{i, \ell} \right) \\
	&= \sum_{\ell=1}^k B^{\ell-1} (t_\ell + c_\ell \cdot B - c_{\ell -1}) \\
	&= \sum_{\ell=1}^k t_\ell \cdot B^{\ell-1} + \sum_{\ell=1}^k c_\ell \cdot B^{\ell} - \sum_{\ell=1}^k c_{\ell -1} \cdot B^{\ell-1} \\
	&= t + c_k \cdot B^k - c_0 = t. \tag{since $c_0 =c_k=0$}
\end{align*}
Hence, $(x_1, \dots, x_n, t)$ is a \textsc{Yes}-instance of \subsetsum.

\smallskip
Conversely, suppose there exist $y_1, \dots, y_n \in \{0, 1\}$
such that $\sum_{i = 1}^n y_i x_i = t$. 
Consider the carries $c_0,c_1,\ldots,c_k$ that arise when summing up $\sum_{i = 1}^n y_i x_i$ in base $B$; more precisely, let $c_0 \coloneq 0$ and for any $\ell \in [k]$ let 
$$c_\ell \coloneq \Big\lfloor \Big(c_{\ell-1} + \sum_{i = 1}^n y_i x_{i,\ell}\Big) / B \Big\rfloor.$$ 
Then clearly $c_\ell \ge 0$ for all $\ell \in [k]$. 
As we have $c_0 < n$, we can prove by induction that $c_\ell < n$ for all $\ell \in [k]$: the observation that $\sum_{i = 1}^n y_i x_{i,\ell} \le n (B-1)$ and the induction hypothesis $c_{\ell-1} < n$ imply that $(c_{\ell-1} + \sum_{i = 1}^n y_i x_{i,\ell}) / B < n$. 
It follows that $c_0, \dots, c_k \in \{0,\ldots,n-1\}$.

Furthermore, $\sum_{i = 1}^n y_i x_i = t$ implies that $\sum_{i = 1}^n y_i \hat{x}_i = \hat{t}_c$ for each $\ell \in [k]$. Indeed, this follows from the correctness of addition in base $B$: For any $\ell \in [k]$, the $\ell$-th digit of the summation result $t_\ell$ equals the previous carry $c_{\ell-1}$ plus the sum of $\ell$-th digits $\sum_{i = 1}^n y_i x_{i,\ell}$, reduced modulo $B$. In other words, we have $t_\ell = c_{\ell-1} + \sum_{i = 1}^n y_i x_{i,\ell} \pmod B$. The new carry $c_\ell = \lfloor (c_{\ell-1} + \sum_{i = 1}^n y_i x_{i,\ell}) / B \rfloor$ is the same by using integer division by $B$ instead of reducing modulo $B$. 
The identity $x = (x \bmod y) + y \cdot \lfloor x/y \rfloor$ now implies $c_{\ell-1} + \sum_{i = 1}^n y_i x_{i,\ell} = t_\ell + B \cdot c_\ell$, which is equivalent to $\sum_{i = 1}^n y_i x_{i,\ell} = t_\ell + c_\ell \cdot B - c_{\ell-1}$, for each $\ell \in [k]$. In other words, we have $\sum_{i = 1}^n y_i \hat{x}_i = \hat{t}_c$.

Again, by correctness of addition, we can see that the last carry is always 0, i.e.~$c_k = 0$. More formally, we can argue that
\begin{align*}
  0 &= t - \sum_{i = 1}^n y_i x_i  \tag{by definition of $y_1, \dots, y_n$}\\
  &= \sum_{\ell=1}^k t_\ell B^{\ell-1} - \sum_{i = 1}^n y_i \sum_{\ell=1}^k x_{i,\ell} B^{\ell-1}  \\
  &= \sum_{\ell=1}^k \Big(t_\ell - \sum_{i = 1}^n y_i x_{i,\ell}\Big) B^{\ell - 1}  \\
  &= \sum_{\ell=1}^k (c_{\ell - 1} - c_\ell \cdot B) B^{\ell-1}  \tag{since $\sum_{i = 1}^n y_i x_{i,\ell} = t_\ell + c_\ell \cdot B - c_{\ell-1}$}  \\
  &= c_0 - c_k B^k = - c_k B^k  \tag{since $c_0=0$} 
\end{align*}
and thus $c_k = 0$.
We have thus proven that $c \coloneq (c_0,\ldots,c_k)$ is contained in the set $C$, i.e.~$(\hat x_1, \dots, \hat x_n, \hat{t}_c)$ is among the constructed \vectorsubsetsum instances. 
Since we also showed $\sum_{i = 1}^n y_i \hat{x}_i = \hat{t}_c$, $(\hat x_1, \dots, \hat x_n, \hat{t}_c)$ is a \textsc{Yes}-instance of \vectorsubsetsum.

\smallskip
Note that $T \coloneq \| \hat{t}_c \|_1 \leq k n B  = \Oh(n t^{1/k})$. Hence, if 
there exists $\eps > 0$ and $k \ge 1$ such that for every $\delta > 0$,
$k$-dimensional \vectorsubsetsum can be solved in time $\Oh(2^{\delta n} T^{k(1-\eps)})$ then \subsetsum can be solved in time $\Oh(n^{k-1} \cdot 2^{\delta n} (n t^{1/k})^{k(1-\eps)}) \le \Oh(2^{\delta' n} t^{1-\eps})$ for some $\delta' = f(k,\eps) \cdot \delta$, where the factor $n^{k-1}$ comes from iterating over all $c \in C$ and $f(k,\eps)$ is a function depending only on $k$ and $\eps$. 
\end{proof}

Doron{-}Arad et al.~\cite{ddimknapsack-itcs26} considered in the arXiv version of their paper~\cite{DBLP:journals/corr/abs-2407-10146} a
$k$-dimensional Knapsack problem, in which we are given a set of items $I$, a
profit function $p : I \to \nat$, a $k$-dimensional cost function $c : I \to
\nat^k$ and a budget $t \in \nat^k$. They observed that $k$-dimensional Knapsack
can be solved in time $\Oh(nW^k)$, where $n$ is the number of items and $W = \norm{t}_\infty$ is the maximum budget of any dimension. They complemented this by a lower bound excluding time $(n+W)^{o(k/\log{k})}$
assuming ETH. 
Note that we can reduce \vectorsubsetsum to $k$-dimensional Knapsack as follows. Let $x_1, \dots, x_n, t \in \N^k$ be a \vectorsubsetsum instance. For each $i \in [n]$, create an item with profit $p(i) = \norm{x_i}_1$ and cost $c(i)= x_i$, and let $t$ be the budget. This $k$-dimensional Knapsack instance has a solution of optimal profit $\norm{t}_1$ if and only if the \vectorsubsetsum admits a solution. 
Hence, we immediately get the following corollary that improves upon the lower bound of Doron{-}Arad et al.

\begin{corollary}
		Assuming SETH, for every $k \in \nat$ and $\eps > 0$ the $k$-dimensional Knapsack problem cannot be solved in time $2^{o(n)} W^{k-\varepsilon}$. Assuming ETH, $k$-dimensional Knapsack cannot be solved in time $2^{o(n)} W^{o(k)}$.
\end{corollary}

\section{Further Observations on Scheduling Problems}
\label{app:folklore}

\begin{lemma}\label{obs:PrjCmax_wlog_small_release_dates}
  For \PrjCmax we can assume without loss of generality that $C \le T$, where $C = \max_{1 \le j \le n} C_j$ is the makespan of an optimal schedule and $T = \sum_{1 \le j \le n} p_j$ is the total processing time. 
  Moreover, if the jobs are sorted by release dates $r_1 \le \ldots \le r_n$ we can assume without loss of generality that $r_j \le \sum_{j' < j} p_{j'}$ for every $j \in [n]$.
\end{lemma}
\begin{proof}
	By shifting any positive offset and sorting, we can assume that $0 = r_1 \le r_2 \le \ldots \le r_n$. 
	Let $t_j = \sum_{j' < j} p_{j'}$ for every $j \in [n]$. In particular, $r_1 = 0 = t_1$.

	Observe that if $r_j \leq t_j$ for every job $j$, then we can schedule every job $j$ starting at time $t_j$ on the first machine. This is a valid schedule with makespan $C = t_n + p_n = T$. Hence, if $C > T$ then there exists a job $j$ with $r_j > t_j$. In this case, we show how to transform the instance into an equivalent one where all jobs have $r_j \leq t_j$. 
	
	Let $j$ the smallest index such that $r_j > t_j$, i.e.~for every $j' <j$ we have $r_{j'} \leq t_{j'}$. Then we can schedule each job $j' \in \{1, \dots, j-1\}$ starting at time $t_{j'}$ on the first machine. This is a valid (partial) schedule with makespan $t_{j-1} + p_{j-1} = t_j < r_j$. 
	Now we shift all later release dates by setting $\bar r_{j'} \coloneq r_{j'} - r_j + t_j$ for any $j' \ge j$, and keeping $\bar r_{j'} \coloneq r_{j'}$ for any $j' < j$. This removes the idle time $[t_j,r_j]$ from consideration, and thus yields an equivalent instance of \PrjCmax. Repeating this argument removes all jobs with $r_j > t_j$. 
\end{proof}

\begin{lemma}\label{lem:PSUj_folkloreUB}
	\PSUj and \QSUj can be solved in time $\Oh(n^2 T^{k-1})$.
\end{lemma}
\begin{proof}
	This is a folklore result, which we include here for completeness. 
	We present a dynamic programming algorithm that solves any instance of \QSUj with $n$ jobs of processing times $p_1, \dots, p_n$, due dates $d_1, \dots, d_n$ and $k$ parallel machines of speeds $s_1, \dots, s_k \in (0, 1]$. In case of parallel identical machines, i.e.~\PSUj, we have $s_1 = \dots = s_k$. 
	
	Note that we can assume without loss of generality that all tardy jobs are not scheduled (they can be scheduled in an arbitrary way after all the non-tardy jobs). Additionally, we can assume that the (non-tardy) jobs are scheduled in Earliest Due Date (EDD) order, i.e.~on each machine, the jobs are scheduled in order of non-decreasing due dates. So sort the jobs according to their due date such that $d_1 \leq d_2 \leq \dots \leq d_n$. Let $T \coloneq \sum_{j=1}^n p_j$ be the sum of processing times. 
	
	The state of the dynamic programming table is a tuple $(j, t, \ell_1, \dots, \ell_{k-1})$, where $j,t \in \{0,\ldots,n\}$ and $\ell_1, \dots, \ell_{k-1} \in \{0,\ldots,T\}$. The entry $\texttt{DP}[j, t, \ell_1, \dots, \ell_{k-1}]$ of the dynamic programming table stores the minimum 
	value $\ell_k \in \{0,\ldots,T\}$ such that we can schedule the jobs $\{1, \dots, j\}$ onto $k$ machines such that $t$ jobs are tardy, and the total processing time of non-tardy jobs on the $i$-th machine is $\ell_i$, for every $i \in [k]$. If such a schedule does not exist, then the table of that state has value $\infty$. 
	The optimal schedule of all jobs $\{1, \dots, n\}$ has $t^*$ tardy jobs, where $t^*$ is the minimum $t \in \{0, 1, \dots, n\}$ such that there exist $\ell_1, \dots, \ell_{k-1} \in [T]$ with $\texttt{DP}[n, t, \ell_1, \dots, \ell_{k-1}] \neq \infty$. 

	To compute the table, first initialise every entry to $\infty$ except the base-case $\texttt{DP}[0, 0, 0, \dots, 0] = 0$. Next, we will fill the table by iterating over $j = 1, \dots, n$ and for each $j$ we will iterate over all possible values of $t, \ell_1, \dots, \ell_{k-1}$.
	We set $\texttt{DP}[j, t, \ell_1, \dots, \ell_{k-1}]$ to the value 
	\begin{equation}\label{eq:DP_QSUj}
	\min  
	\begin{cases}
		\texttt{DP}[j-1, t-1, \ell_1, \dots \ell_{k-1}] & \text{if } t \ge 1,\\
		\texttt{DP}[j-1, t, \ell_1, \dots, \ell_i - p_j \dots, \ell_{k-1}] & \text{for } i \in [k-1] \text{ such that } \ell_i / s_i \leq d_j \text{ and } p_j \le \ell_i,\\
		\texttt{DP}[j-1, t, \ell_1, \dots, \ell_{k-1}] + p_j & \text{if } (\texttt{DP}[j-1, t, \ell_1, \dots, \ell_{k-1}]  +p_j)/ s_i \leq d_j.\\
	\end{cases}
	\end{equation}
	Here we interpret the minimum over an empty set as $\infty$. 
	Since there are $\Oh(n^2 T^{k-1})$ dynamic programming states and each entry of the dynamic programming table takes time $\Oh(k)$ to compute, the algorithm takes total time $\Oh(n^2 T^{k-1})$ as $k$ is a fixed constant.

	It remains to show correctness of our algorithm, which we do by an induction on $j$. 
	The base case $j=0$ holds trivially. Hence, consider $j \ge 1$ and assume that the table is correctly filled for $j-1$.  
	Fix $t \in \{0, 1, \dots, n\}$ and $\ell_1, \dots, \ell_{k-1} \in \{0, 1, \dots, T\}$.
	Note that any schedule of the first $j$ jobs can be constructed from a schedule of the first $j-1$ jobs by scheduling $j$ either as tardy, on one of the first $k-1$ machines, or on the last machine. Hence, if there is no schedule corresponding to $(j, t, \ell_1, \dots, \ell_{k-1})$, then the minimum in (\ref{eq:DP_QSUj}) is taken over the empty set and $\texttt{DP}[j, t, \ell_1, \dots, \ell_{k-1}]$ correctly contains value $\infty$. Otherwise, the correct value of $\texttt{DP}[j, t, \ell_1, \dots, \ell_{k-1}]$ is at most the value of (\ref{eq:DP_QSUj}). We now show that it is at least the value of (\ref{eq:DP_QSUj}). 

	Let $\ell_k \in \{0, 1, \dots, T\}$ be the correct value of $\texttt{DP}[j, t, \ell_1, \dots, \ell_{k-1}]$, i.e.~$\ell_k$ is minimal such that there exists a partition $\{1, \dots, j\} = X_1 \uplus \dots \uplus X_k \uplus B$ of the first $j$ jobs where $B$ contains $t = |B|$ tardy jobs, and $X_i$ contains all non-tardy jobs scheduled on the $i$-th machine with $\sum_{j' \in X_i} p_{j'} = \ell_i$, for each $i \in [k]$. 
	Since $X_1 \uplus \ldots \uplus X_k \uplus B$ is a partition of $\{1, \dots, j\}$, we have the following three cases: either (i) $j \in B$, (ii) $j \in X_1 \uplus \ldots \uplus X_{k-1}$, or (iii) $j \in X_k$.
	In case (i), the job $j$ is tardy. In particular, $|B| = t \geq 1$. Consider the same schedule restricted to the first $j-1$ jobs: the number of tardy jobs decreases by 1 but the total processing time of non-tardy job on each machine remains unchanged. Thus, in that case, by induction hypothesis, we have $\ell_k \geq \texttt{DP}[j-1, t-1, \ell_1, \dots, \ell_{k-1}]$.
	In case (ii) and (iii), the job $j$ is non-tardy and scheduled on the $i$-th machine for some $i \in [k]$. In particular, $\ell_i/s_i \leq d_j$ and $p_j \leq \sum_{j' \in X_i } p_{j'} = \ell_i$. Consider the same schedule restricted to the first $j-1$ jobs: the number of tardy jobs and the total processing time of non-tardy jobs on each machine remains unchanged, except for the total processing time of non-tardy jobs on the $i$-th machine, which decreases by $p_j$. By induction hypothesis, we deduce that if $i \in [k-1]$, i.e.~case (ii), then $\ell_k \geq \texttt{DP}[j-1, t, \ell_1, \dots, \ell_i - p_j , \dots, \ell_{k-1}]$; and otherwise, $i = k$, i.e.~case (iii), and we have $\ell_k \geq \texttt{DP}[j- 1, t, \ell_1, \dots, \ell_{k-1}] + p_k$.
	This proves that $\ell_k$ is at least the value of (\ref{eq:DP_QSUj}), and therefore the correct value $\ell_k$ of $\texttt{DP}[j, t, \ell_1, \dots, \ell_{k-1}]$ is equal to (\ref{eq:DP_QSUj}).
\end{proof}
 
\paragraph{Maximum lateness and maximum tardiness}
For a given schedule, the \emph{lateness} (respectively, \emph{tardiness}) of job $j$ is defined as $L_j \coloneq C_j - d_j$ (respectively, $T_j \coloneq \max\{0, C_j - d_j\}$), where $C_j$ and $d_j$ denote the completion time and due date of job $j$. We consider the scheduling problem \PLmax (respectively, \PTmax), where one is given $n$ jobs with processing times $p_1, \dots, p_n$ and due dates $d_1, \dots, d_n$ and a target objective $\ell \in \Z$ (respectively, $t \in \N$) and the task is to decide whether there exists a schedule of the $n$ jobs on $k$ identical parallel machines such that the maximum lateness $L_{\max} \coloneq \max_{j\in [n]} L_j$ is at most $\ell$ (respectively, the maximum tardiness $T_{\max} \coloneq \max_{j\in [n]} T_j$ is at most $t$). 
We show in \cref{lem:PSUj0_Lmax_Tmax} that both those problems are parameter-preserving equivalent to the special case of \PSUj asking for $\sum_j U_j =0$, and thus to \PrjCmax by \cref{lem:GroupedkPart_to_PrjCmax}. The reductions are folklore. We include them for completeness in order to verify that they are indeed parameter-preserving. 

\begin{lemma}\label{lem:PSUj0_Lmax_Tmax}
	The special case of \PSUj asking for $\sum_j U_j =0$ is parameter-preserving equivalent to \PLmax and to \PTmax.
\end{lemma}
\begin{proof}
	Note that $\sum_j U_j = 0$ if and only if $T_{\max} = 0$, so the special case of \PSUj asking for $\sum_j U_j =0$ reduces to \PTmax. Next, note that $T_{\max} = \max \{0, L_{\max}\}$, so any schedule that minimises $L_{\max}$ also minimises $T_{\max}$, i.e.~\PTmax reduces to \PLmax. We now reduce \PLmax to the special case of \PLmax asking for $L_{\max} \leq 0$. Since $L_{\max} \leq 0$  if and only if $\sum_j U_j = 0$, this closes the cycle of reductions. 
	
	Let $p_1, \dots, p_n$ and $d_1, \dots, d_n$ be the processing times and due dates of $n$ given jobs. Let $\ell \in \Z$ be the target objective, i.e.~the task is to decide whether there exists a schedule of the $n$ jobs such that $L_{\max} \leq \ell$. Consider the instance $p_1', \dots, p_n'$ and $d_1', \dots, d_n'$ where $p_j' \coloneq p_j$ and $d_j' \coloneq d_j + \ell$. 
	Note that we can assume without loss of generality that $|\ell| \leq \min_{j \in [n]} \{d_j\}$, as otherwise the given instance is trivially in the NO case. Hence, the constructed due dates are non-negative integers. Additionally, the order of the due dates is preserved and thus for the same schedule we have $L_{\max}' = L_{\max} - \ell$, i.e.~$L_{\max} \leq \ell$ if and only if $L_{\max}' \leq 0$. As the processing times remain unchanged, this is indeed a parameter-preserving reduction. 
\end{proof}

\section{Technical Lemmas}
\label{app:technical-lemmas}

\lembehrend*
\begin{proof}
	It suffices to adapt Behrend's classic construction~\cite{behrend1946sets} by adding an extra coordinate to account for the number of summed items. For completeness, we provide the details below.

	Let $u, r, d \in \mathbb N$ be parameters that we choose later. Consider the set of vectors $R \coloneq \{y \in \{0, 1, \dots, u\}^d \ : \ \norm{y} = r\}$ where $\norm{\cdot}$ denotes the Euclidean norm. 
	Then for any $\ell \geq 0$ and vectors $y_1, \dots, y_\ell, y \in R$ we have $\norm{\ell \cdot y} = \ell \cdot r$ and $\norm{y_1 + \dots + y_\ell} = \norm{y_1} + \dots + \norm{y_\ell} = \ell \cdot r$ if and only if $y_1 = \dots = y_\ell$. 
	Embed the $d$-dimensional vectors of $R$ into integers using a base $2ku$ encoding: 
	$$B \coloneq \left\{(2ku)^d + \sum_{i=1}^d y[i] \cdot (2ku)^{i-1} \ : \ y
    \in R\right\}.$$ We claim that the set $B$ is strong $k$-average-free. Indeed, let $0\leq a, b \leq k$ and consider integers $x_1, \dots, x_{a}, x \in B$ with their corresponding vectors $y_1, \dots, y_{a}, y \in R$.
	Then 
	$$x_1 + \dots + x_{a} = a \cdot (2ku)^d + \sum_{i=1}^d (y_1 + \dots + y_{a})[i] \cdot (2ku)^{i-1} \;\;\text{and}\;\; b x = b \cdot (2ku)^{d} + \sum_{i=1}^{d} b \cdot y[i] \cdot (2ku)^{i-1}.$$ 
	Since $0 \leq a, b \leq k$ and $0 \leq y[i] \leq u$ for any vector $y \in R$ and any $i \in [d]$, there is no overflow between terms corresponding to different exponents of $2ku$. It follows that $x_1 + \dots + x_{a} = b x$ if and only if $a = b$ and $y_1 + \dots + y_{a} = b y$. In particular, this implies that $\norm{y_1 + \dots + y_{a}} = \norm{a y} = a \cdot r$, which by the above observation is equivalent to $y_1 = \dots = y_{a}$. So we deduce that $y_1 = \dots = y_{a} = y$. 

	It remains to tune the parameters $u, r, d \in \mathbb N$ such that $B$ has size at least $n$. 
	We choose $d \coloneq \lceil 2 / \mu \rceil +2$ so that $d \ge 3$. We choose $r$ that maximises the size of $B$, i.e.~$|B| = |R| = \max_{r \in \mathbb N} |\{y \in \{0, 1, \dots, u\}^d \ : \ \norm{y} = r\}|$.
	For any $y \in R$, we have $\norm{y}^2  = y[1]^2 + \dots + y[d]^2 \in \{0, 1, \dots, d u^2\}$, so there are at most $d u^2+1$ choices for $r$. Therefore, 
	$|B| \geq (u+1)^d / (du^2+1) \geq u^{d-2} / d$. 
	Hence, by setting $u \coloneq \lceil (dn)^{1/(d-2)} \rceil$ we get $|B| \geq n$. We can bound the integers in $B$ by
	$
	U \coloneq \max(B) \leq 2 (2ku)^{d} \leq 
	 n^{1+\mu} k^{\Oh(1/\mu)}
	$, where we plugged in the values of $d$ and $u$ and bounded $2 (2k)^d \le k^{\Oh(d)} \le k^{\Oh(1/\mu)}$, $d^{d/(d-2)} \le d^3 \le (1/\mu)^{\Oh(1)} \le 2^{\Oh(1/\mu)} \le k^{\Oh(1/\mu)}$, and $n^{d/(d-2)} = n^{1 + 2/\lceil 2/\mu \rceil} \le n^{1+\mu}$.
\end{proof}

\genrohwedder*
\begin{proof}
	We assume that $\tau \ge k^2$, as otherwise $\tau$ copies of the integer $1$ satisfy the claim.
	For $n \in \N$, let $\text{bit}_i(n) \in \{0,1\}$ denote the $i$-th bit of the binary representation of $n$.
	Let $h \coloneq \lfloor \log (\tau/k^2) \rfloor$, $b \coloneq 2^h - 1$, $a \coloneq \tau - (k-1) b$
	and $c \coloneq \floor{a/2^h}$. Define the sets
	$$
	A \coloneq \{2^i \ : \ i \in \{0, \dots, h-1\} \text{ and } \text{bit}_i(a) = 1\}
 \quad \text{  and  } \quad 
 B \coloneq \{2^i \ : \ i \in \{0, \dots, h-1\}\}.
	$$
	For every $i \in [k-1]$ we construct a copy of $B$ that we denote by $B^{(i)}$.
	We let $C$ be a multiset consisting of $c$ copies of the integer $2^h$.
	Finally, let $P$ be the union of multisets $A$, $B^{(1)},\ldots,B^{(k-1)}$ and
	$C$.
	This concludes the construction of $P$, which clearly takes time $\Oh(|P|)$. We now verify the properties of $P$. 
	
	\begin{property}
			$|P| = \Oh(k^2 \log{\tau})$.
	\end{property}
	\begin{proof}
		This follows from $|P| = |A| + (k-1)|B| + |C|$, and observing that $|A| \leq |B| \leq  h \le \log \tau$ and 
			\begin{displaymath}
					|C| = c = \floor{a/2^h} \le \tau/2^{\floor{\log(\tau/k^2)}} \le
					2\tau/2^{\log(\tau/k^2)} \le 2k^2.\qedhere
			\end{displaymath}\end{proof}
	\begin{property}
		$\Sigma(P) = \tau$.
	\end{property}
	\begin{proof}
		Observe that $\Sigma(A) \equiv a \Mod{2^h}$ and $\Sigma(A) \leq 2^h$. So we have 
		\begin{align*}
			c 2^h = \floor{\frac{a}{2^h}} 2^h = a - (a \bmod{2^h}) = a - \Sigma(A),
		\end{align*}
		and therefore
		\begin{align*}
			\Sigma(P) 
			&= \Sigma(A) + (k-1)\Sigma(B) + \Sigma(C) = \Sigma(A) + (k-1)b + c 2^h \tag{by construction of $B$ and $C$} \\
			&= (k-1) b + a = \tau.\tag*{\hfill\qedhere}
		\end{align*}
	\end{proof}

	\begin{property}
		For any integers $0 \le t_1 \le \dots \le t_k$ such that $t_1 + \dots + t_k = \tau$, we can partition $P = P_1 \uplus \ldots \uplus P_k$ such that $\Sigma(P_i) = t_i$ for every $i
 \in [k]$. 
	\end{property}
	\begin{proof}
	For every $i \in [k-1]$, let $t_i^{+} \coloneq 2^h \floor{t_i/2^h}$ and $t_i^{-} \coloneq t_i - t_i^{+} = t_i \bmod 2^h$. Notice that $t_i^{+} \equiv 0  \Mod{2^h}$ and $t_i^{-} < 2^h$. So we can choose $P_i^{-} \subseteq B^{(i)}$ to be the unique set such that $\Sigma(P_i^{-}) = t^{-}_i$. Let $P^{+}_i \subseteq C$ be a multiset of cardinality $t^{+}_i/2^h$, i.e.~$\Sigma(P_i^+) = 2^h \cdot t_i^+ / 2^h= t_i^+$. Then, since $P_i^+$ and $P_i^-$ are disjoint, the sets $P_i \coloneq P_i^- \cup P_i^+$ satisfy $\Sigma(P_i) = t_i$, for every $i \in [k-1]$.

	We claim that the sets $P_1, \dots, P_{k-1}$ are disjoint.  
	Notice that $P_1^-, \dots, P_{k-1}^-$ are disjoint by construction, so we focus on showing that we can construct disjoint multisets $P_1^+, \dots, P_{k-1}^+ \subset C$ for which it suffices to show that $\sum_i |P_i^+| \leq C$. Recall that $2^h = 2^{\floor{\log (\tau / k^2)}} \leq \tau / k^2$. 
	So we bound the size of $C$ by
	\begin{align*}
			|C| &= c =  \floor{\frac{a}{2^h}} 
			= \floor{\frac{\tau - (k-1)(2^h - 1)}{2^h}} 
			\ge \floor{\frac{\tau}{2^h} - k} \ge \floor{\frac{\tau - \tau/k}{2^h}} \tag{by $2^h \leq 
			\tau / k^2 $}\\
			&= \floor{\frac{\sum_{i=1}^k t_i - \tau/k}{2^h}} \ge \floor{\frac{\sum_{i=1}^{k-1} t_i}{2^h}}  \tag{since $t_k \ge \tau/k$}\\
			& \ge \floor{\frac{\sum_{i=1}^{k-1} t_i^{+}}{2^h}} 
			= \sum_{i=1}^{k-1} \frac{t_i^{+}}{2^h} \tag{since $t_i^+ \equiv 0 \Mod{2^h}$} = \sum_{i=1}^{k-1} |P_i^+|.
	\end{align*}
	Hence, $C$ contains enough elements to be distributed among $P_1^+, \dots, P_{k-1}^+$ such that $P_1, \dots, P_{k-1}$ are disjoint.
	Finally, set $P_k \coloneq P \setminus (\bigcup_{i=1}^{k-1} P_i)$. Then $\Sigma(P_k) = \Sigma(P) - \left( \Sigma(P_1) + \dots +\Sigma(P_{k-1}) \right) = \tau - (t_1 + \dots + t_{k-1}) = t_k$. Hence, $P_1 \uplus \dots \uplus P_k$ is a partition of $P$ with $\Sigma(P_i) = t_i$ for every $i \in [k]$. 
\end{proof}
As a consequence, we can construct the set $P$ with the required properties in time $\Oh(|P|) = \Oh(k^2 \log \tau)$. 
\end{proof}

\end{document}